\documentclass{lmcs} 
\pdfoutput=1
\usepackage[utf8]{inputenc}

\usepackage{lastpage}
\lmcsdoi{20}{2}{6}
\lmcsheading{}{\pageref{LastPage}}{}{}%
{Oct.~05,~2022}{May~08,~2024}{}


\usepackage{hyperref}
\usepackage{stmaryrd}
\usepackage{booktabs}
\usepackage{comment}
\usepackage{algorithm}
\usepackage{algpseudocode}

\newcommand{\oar}{\mathit{oar}}
\newcommand{\fcomp}{\circ}
\newcommand{\VarUniv}{\mathbb{V}} 
\newcommand{\var}{\mathit{vars}}
\newcommand{\flio}{\mathrm{FLIF}^{\rm io}}

\newcommand{\Sch}{\mathcal{S}}
\newcommand{\names}{\mathit{Names}}
\newcommand{\In}{I}
\newcommand{\Out}{O}

\newcommand{\arity}{\mathit{ar}}
\newcommand{\dom}{\mathbf{dom}}
\newcommand{\ar}{\arity{}}
\newcommand{\iar}{\mathit{iar}}
\newcommand{\FV}{\mathit{fvars}} 
\newcommand{\inst}{D}
\newcommand{\adom}{\mathrm{adom}}
\newcommand{\vex}[1]{${\VarSet} #1$-executable}

\newcommand{\VarSet}{\mathcal{V}} 
\newcommand{\sat}[3]{#1, #3 \models #2} 

\newcommand{\satD}[2]{\inst, #2 \models #1}
\newcommand{\nsatD}[2]{\inst, #2 \not \models #1}
\newcommand{\symdif}{\mathbin{\triangle}}
\newcommand{\nuin}{\nu_{\rm in}}
\newcommand{\nuout}{\nu_{\rm out}}
\newcommand{\semm}[2]{\llbracket #1 \rrbracket_{#2}}
\newcommand{\sem}[1]{\semm{#1}{\inst}}
\newcommand{\semV}[1]{\semm{#1}{\inst}^{\SVars}}
\newcommand{\conc}{\cdot}
\renewcommand{\setminus}{-} 
\newcommand{\comp}{\ensuremath \mathbin{;}}
\newcommand{\cons}{c}
\newcommand{\set}[2]{(#1 \mathop{:=} #2)}
\newcommand{\eq}[2]{(#1 = #2)}
\newcommand{\EvalV}[4]{\mathit{Eval}_{#1,#2}(#3,#4)}
\newcommand{\EvalL}[3]{\mathit{Eval}_{#1}(#2,#3)}
\newcommand{\EvalLV}[4]{\mathit{Eval}_{#1}^{#2}(#3,#4)}

\newcommand{\disj}[2]{#1 \cap #2 = \emptyset}
\newcommand{\access}{\overset{\rm access}{\bowtie}}
\newcommand{\MyIn}{\mathit{In}}

\newcommand{\compl}[1]{\overline{#1}}
\newcommand{\restr}[2]{#1|_{#2}}

\newcommand{\intersection}[1]{}

\newcommand{\SVars}{\mathbb{V}}

\newcommand{\fo}{\mathrm{FO}}
\newcommand{\FO}[1]{\fo[#1]}
\theoremstyle{plain} 

\begin{document}

\title[Executable FO \& FLIF]{Executable First-Order Queries in\texorpdfstring{\\}{} the Logic of Information Flows{\rsuper*}}
\titlecomment{{\lsuper*}This paper is the combined, extended, and fully revised
journal version of two papers presented at ICDT 2020 and ICDT 2021
\cite{flifexfo,AamerIO}.}

\thanks{
This work was partially supported by FWO project G0D9616N 
and by the Flanders AI Research Program.  Heba Aamer was supported 
by the Special Research Fund (BOF) (BOF19OWB16) while at Hasselt
University.  
Jan Van den Bussche is partially supported by the National 
Natural Science Foundations of China (61972455).}	


\author[H.~Aamer]{Heba Aamer\lmcsorcid{0000-0003-0460-8534}}[a]
\author[B.~Bogaerts]{Bart Bogaerts\lmcsorcid{0000-0003-3460-4251}}[a]
\author[D.~Surinx]{Dimitri Surinx\lmcsorcid{0000-0003-0718-8854}}[b]
\author[E.~Ternovska]{Eugenia Ternovska\lmcsorcid{0000-0003-0751-4031}}[c]
\author[J.~Van den Bussche]{Jan Van den Bussche\lmcsorcid{0000-0003-0072-3252}}[b]

\address{Vrije Universiteit Brussel, Belgium}	
\email{heba.mohamed@vub.be,bart.bogaerts@vub.be}  

\address{Hasselt University, Belgium}	
\email{surinxd@gmail.com,jan.vandenbussche@uhasselt.be}  

\address{Simon Fraser University, Canada}	
\email{ter@sfu.ca} 






\begin{abstract}
The logic of information flows (LIF) has recently been proposed as
a general framework in the field of knowledge representation.  In
this framework, tasks of procedural nature can still be modeled
in a declarative, logic-based fashion.  In this paper, we focus
on the task of query processing under limited access patterns, a
well-studied problem in the database literature.  We show that
LIF is well-suited for modeling this task.  Toward this goal, we
introduce a variant of LIF called ``forward'' LIF (FLIF), in a
first-order setting.  FLIF takes a novel graph-navigational
approach; it is an XPath-like language that nevertheless turns
out to be equivalent to the ``executable'' fragment of first-order 
logic defined by Nash and Lud\"ascher.  One can also classify 
the variables in FLIF expressions as inputs and outputs.  
Expressions where inputs and outputs are disjoint, 
referred to as io-disjoint FLIF expressions, allow a 
particularly transparent translation into 
algebraic query plans that respect the access limitations.  
Finally, we show that general FLIF expressions 
can always be put into io-disjoint form.  
\end{abstract}

\maketitle

\section*{Introduction}

An information source is said to have a limited access pattern if
it can only be accessed by providing values for a specified
subset of the attributes; the source will then respond with
tuples giving values for the remaining attributes.  A typical
example is a restricted telephone directory $\rm D(name;tel)$
that will show the phone numbers for a given name, but not the
other way around.  For another example, the public bus company 
may provide its weekdays schedule as a relation $\rm
Route(stop,\allowbreak interval;\allowbreak time,\allowbreak
line,\allowbreak next,\allowbreak duration)$ that, given a bus
stop and a time interval, outputs bus lines that stop there at a
time within the interval, together with the duration to the next
stop.  Note how we use a semicolon to separate the attributes
required to access the information source from the rest of the
attributes.

The topic of querying information sources with limited access
patterns was put on the research agenda in the mid 1990s
\cite{rsu_bindingpatterns}, and has
been intensively investigated since then, with recent work until
at least 2018
\cite{ylgcu_capabilities_lim_med,flms_qo_limited,
dgl_rec_int,lilimited,mhf_cont_int,nl_accesspatterns,nl_access_constraints,
cali_query_limit,cali_cont_limit,cali_opt_accpat,bgs_relevance_access,
bbb_access_constraints,blt_query_access_constraints,benedikt_planproofs_tods,
cali_deepweb_amw,cali_access_comput}.
The research is motivated by diverse applications, such as query
processing using indices, information integration, or querying
the Deep Web.  A review of the field was given by Benedikt et
al.\ \cite[Chapter 3.12]{benedikt_book}.

In this paper, we offer a fresh perspective on querying with
limited access patterns, based on the Logic of Information Flows
(LIF)\@.  This framework has been recently introduced in the
field of knowledge representation \cite{lif_amw,lif_frocos}.  The
general aim of LIF is to model how information propagates in
complex systems.  LIF allows machine-independent
characterizations of computation; in particular, it allows tasks
of a procedural nature to be modeled in a declarative fashion.  

In the full setting, LIF is a rich family of logics with
higher-order features.  The present paper is self-contained,
however, and we introduce here a lightweight, first-order
fragment of LIF, which we call \emph{forward} LIF (FLIF).  
Our goal then is to show that FLIF is suitable to query 
information sources with limited access patterns.

Specifically, we offer the following insights and contributions:
\begin{enumerate}
  \item
    We offer a new perspective on databases with access
    limitations, by viewing them as a graph. The nodes of the
    graph are valuations; the edges denote access to information
    sources.  The start node of an edge provides values to input
    variables, and the end node provides values to output
    variables.
  \item

    Our perspective opens the door to using a graph query
    language to query databases with access limitations.
    Standard navigational graph query languages
    \cite{nsparql,rafragments_ins,gxpath,nav_with_tc,
    angles_survey} have a logical foundation in Tarski's algebra
    of binary relations
    \cite{tarski_relcalc,maddux_originra,pratt_relcalc,catemarx_sigmodrecord}.
    However, in our situation, nodes in a graph are not abstract
    elements, but valuations that give values to variables.  

  \item

    Interestingly, LIF, in its first-order version, can be
    understood exactly as the desired extension of Tarski's algebra to
    binary relations of valuations.  LIF is a \emph{dynamic}
    logic: like first-order dynamic logic \cite{dynamiclogicbook}
    or dynamic predicate logic \cite{dplogic}, expressions of LIF
    are not satisfied by single valuations, but by \emph{pairs} of
    valuations. Such pairs represent transitions of information.
    However, LIF is very general
    and has operators, such as converse, or cylindrification,
    which do not rhyme with the limited access to information
    sources that we want to target in this work.
    Therefore, in this paper, we introduce FLIF, an instantiation
    of the LIF framework where information can only flow forward.
    Like navigational graph query languages,
    FLIF expressions define sets of pairs of valuations so that
    there is a path in the graph from the first valuation
    of the pair to the second.

  \item

    We show that FLIF is equivalent in expressive power to
    executable FO, an elegant syntactic fragment of first-order
    logic introduced by Nash and Lud\"ascher
    \cite{nl_accesspatterns}.  Formulas of executable FO can be
    evaluated over information sources in such a way that the
    limited access patterns are respected.  Furthermore, the
    syntactical restrictions are not very severe and become
    looser the more free variables are declared as inputs.

  \item

    Our equivalence result between FLIF and executable FO is
    interesting since FLIF is a simple compositional language,
    built from atomic expressions using just three navigational
    operators: composition, union, and difference.  These
    operators allow {one} to build paths, explore alternatives, and
    exclude paths. The atomic expressions are information
    accesses, tests, or variable assignments.  Thus, FLIF is a
    very different language from executable FO, where the
    classical first-order constructs (disjunction, conjunction, negation,
    quantification) are syntactically restricted to be ordered so
    as to respect the
    access limitations, and cannot simply be combined
    orthogonally.  FLIF, which directly navigates through the
    graph, is also different from other approaches in the
    literature where first the ``accessible part'' (up to some
    depth) of the database is retrieved, after which an arbitrary
    query can be evaluated on this part.

  \item

    We also specialize our result to FLIF expressions that are
    \emph{io-disjoint}.  This is a property
    coming from our companion paper
    where we analyze input and output sensitivity in LIF
    expressions \cite{lif-tocl}.  An expression $\alpha$ is
    io-disjoint if, whenever $\alpha$ can reach a
    valuation $\nu_{\rm out}$ from a valuation $\nu_{\rm in}$,
    the values of the variables in $\nu_{\rm out}$ depend only on the
    values of variables in $\nu_{\rm in}$ that have not changed in
    $\nu_{\rm out}$.  For io-disjoint expressions, the single
    valuation $\nu_{\rm out}$ contains all the relevant
    information: in this sense, the io-disjoint fragment of FLIF
    can be given a static (single-valuation) semantics as opposed
    to the dynamic semantics of full FLIF.

  \item

    We show three results on io-disjoint FLIF. First, when
    translating FLIF to executable FO, a
    more economical translation is possible if the FLIF
    expression is io-disjoint. Here, by ``economical'', we mean
    that fewer variables are needed in the FO formula, and the FO
    formula is closer in syntax to the FLIF expression.

  \item

    Second, we show that io-disjoint FLIF expressions can be
    translated into \emph{plans} in a particularly simple and
    transparent manner.  Plans are a standard way of formalizing
    query processing with limited access patterns
    \cite{benedikt_book}.  In such plans, database relations can
    only be accessed by joining them on their input attributes
    with a relation that is either given as input or has already
    been computed.  Apart from that, plans can use the usual
    relational algebra operations.  That executable FO can be
    translated into plans is well known, so, by the equivalence
    with FLIF, the same holds for FLIF\@.  However, the resulting
    plans can be rather complex, just like the classical
    translation from relational calculus to relational algebra
    \cite{ahv_book} can produce rather ugly algebra expressions
    in general. So, our result is that for io-disjoint FLIF, very
    simple plans can be produced.  The plans we generate do not
    need the renaming operator, and use only natural joins (no cartesian products or theta-joins).
    
  \item

    Third, we show that, actually, any FLIF expression can be
    simulated by an io-disjoint one.  The simulation requires
    auxiliary variables and variable renamings, and the
    correctness proof is quite intricate.  We see this
    result mainly as an expressiveness result, not as
    suggesting a practical way to evaluate arbitrary FLIF expressions.
    Indeed, these can be evaluated rather directly as is, since
    FLIF is an algebraic language in itself.
\end{enumerate}

This paper is further organized as follows.  We begin with 
some preliminaries in Section~\ref{secpre}.  Section~\ref{seclif} 
introduces the language FLIF\@.  In Section~\ref{secmodel}, 
we recall the basic setting of executable FO on databases 
with limited access patterns; furthermore, 
we prove the equivalence between FLIF and executable FO.  
In Section~\ref{secio}, we formally define the io-disjoint
fragment.
Then, in Section~\ref{seceval}, we give a translation
from that fragment to executable FO which improves upon the 
translation from FLIF from Section~\ref{secmodel}.
In Section~\ref{seceval}, we also give a translation from 
FLIF to its io-disjoint fragment.  In Section~\ref{sec:proofs2}, 
we give the correctness proofs of the translation
theorems from Sections~\ref{secmodel} and~\ref{seceval}.  
Section~\ref{secplan} discusses evaluation plans.  
Finally, we discuss related work and then conclude
in Sections \ref{secrel} and \ref{seconcl} respectively.  

\section{Preliminaries}\label{secpre}
Relational database schemas are commonly formalized as finite
relational vocabularies, i.e., finite collections of relation
names, each name with an associated arity (a natural number).
To model limited
access patterns, we additionally specify an \emph{input
arity} for each name.  For example, if $R$ has arity five and
input arity two, this means that we can only access $R$ by giving
input values, say $a_1$ and $a_2$, for the first two arguments; $R$
will then respond with all tuples $(x_1,x_2,x_3,x_4,x_5)$ in $R$
where $x_1=a_1$ and $x_2=a_2$.

Thus, formally, we define a \emph{database schema}
as a triple $\Sch=(\names,\ar,\iar)$, 
where $\names$ is a set of relation names; $\ar$ assigns a
natural number $\ar(R)$ to each name $R$ in $\names$, called the
arity of $R$; and $\iar$ similarly assigns an input arity to each
$R$, such that $\iar(R)\leq\ar(R)$.  In what follows, we use 
$\oar(M)$ (output arity) for $\ar(M) - \iar(M)$.

\begin{rem}\label{remarkmultiple}
In the literature, a more general notion of schema is often used,
allowing, for each relation name, several possible sets of
input arguments; each such set is called
an access method.  In this paper, we stick to the simplest setting
where there is only one access method per relation, consisting of
the first $k$ arguments, where $k$ is set by the input arity.
All subtleties and difficulties already show up in this setting.  
Nevertheless, our definitions and results can be easily generalized to the
setting with multiple access methods per relation.
\end{rem}

The notion of database instance remains the standard one.
Formally, we fix a countably infinite universe $\dom$ of atomic
data elements, also called \emph{constants}.  Now an \emph{instance}
$\inst$ of a schema $\Sch$ assigns to each relation name $R$ an
$\ar(R)$-ary relation $\inst(R)$ on $\dom$.  We say that $D$ is
\emph{finite} if every relation $D(R)$ is finite.  The
\emph{active domain} of $D$, denoted by $\adom(D)$, is the set of
all constants appearing in the relations of $D$.

The syntax and semantics of first-order logic (FO, relational
calculus) over $\Sch$ is well known \cite{ahv_book}.  
The set of free variables of an FO formula $\varphi$ is 
denoted by $\FV(\varphi)$.
Moreover, in formulas, we allow constants only in equalities 
of the form $x=c$, 
where $x$ is a variable and $c$ is a constant.  As we mentioned 
earlier, in writing relation atoms, we find it clearer to separate 
input arguments from output arguments by a semicolon.  Thus, we 
write relation atoms in the form $R(\bar{x}; \bar{y})$, where 
$\bar{x}$ and $\bar{y}$ are tuples of variables such that the 
length of $\bar{x}$ is $\iar(R)$ and the length of $\bar{y}$ is 
$\oar(R)$.  
For example, the relation atom $R(x,z;y,y,z)$
indicates that $R$ is a relation name with
$\ar(R)=5$ and $\iar(R)=2$; consequently, $\oar(R)=3$.

We use the ``natural'' semantics \cite{ahv_book} and let
variables in formulas range over the whole of $\dom$.  
Formally, an $X$-\emph{valuation} is a valuation 
defined on a set $X$ of variables, 
and precisely, it is a mapping $\nu : X \to \dom$.  
We will often not specify the set $X$ of variables a valuation 
is defined on when it is clear from context.
It is convenient to be able to apply valuations 
also to constants, agreeing that $\nu(c)=c$ for any valuation 
$\nu$ and any $c \in \dom$.  
Moreover, in general, for a 
valuation $\nu$, a variable $x$, and a constant $c$, we 
use $\nu[x:=c]$ for the valuation that is the same as 
$\nu$ except that $x$ is mapped to $c$.  
Additionally, we say that two valuations $\nu_1$ and $\nu_2$ agree on 
(outside) a set of variables $X$ when $\nu_1(x) = \nu_2(x)$ for 
every variable $x \in X$ ($x \not \in X$).  
Finally, given an instance $D$ of $\Sch$, 
an FO formula $\varphi$ over $\Sch$, and a valuation $\nu$ 
defined on $\FV(\varphi)$, the definition of when $\varphi$ 
is satisfied by $D$ and $\nu$, denoted by $\sat D\varphi\nu$, 
is standard.

\section{Forward LIF} \label{seclif}

In this section, we introduce the language
FLIF\@.\footnote{Pronounced as ``eff-lif''.}
The language itself is a form of dynamic logic.  
Indeed, the semantics of any FLIF expression is
defined as a set of \emph{pairs} of valuations. 
The operators are an algebraization of first-order logic 
connectives.  Although FLIF is a dynamic algebraic form of 
first-order logic, it is notable that it lacks 
quantification operators, which makes it
especially simple.

\paragraph{Syntax and semantics of FLIF: atomic expressions}
The central idea is to view a database instance as a graph.
The nodes of the graph are all possible valuations on 
some set of variables (hence the graph is infinite.)
The edges in the graph are labeled with \emph{atomic FLIF
expressions}.  
Some of the edges are merely tests (i.e., self-loops), 
while other edges represent a change in the state.

Syntactically, over a schema $\Sch$ and a set of variables $\SVars$, 
there are five kinds of
atomic expressions $\tau$, given by the following grammar:
\[
\tau ::= R(\bar x;\bar y) \mid (x=y) \mid (x=c) \mid (x:=y)
\mid (x:=c)
\]
Here, $R(\bar x;\bar y)$ is a relation atom over $\Sch$ 
as in first-order logic with $\bar x$ and $\bar y$ being 
tuples of variables in $\SVars$, 
$x$ and $y$ are variables from $\SVars$, and $c$ is a constant. 
The atomic expressions $(x = y)$ and $(x = c)$
are equality tests, while the expressions $(x := y)$ and $(x := c)$
are assignment expressions. 
From the grammar, we see that any atomic expression $\tau$ 
is defined such that 
$\var(\tau) \subseteq \SVars$ where $\var(\tau)$ is the set of 
variables used in $\tau$.

Semantically, given an instance $D$ of $\Sch$, a set of variables $\SVars$,
and an atomic expression $\tau$ over $\Sch$ and $\SVars$,
we define the set of $\tau$-labeled edges in the graph view 
of $D$ as a set $\semV\tau$ of ordered pairs of $\SVars$-valuations, 
as follows.
\begin{defi}\label{def:flif1}
\ 
\begin{enumerate}
\item
$\semV {R(\bar x;\bar y)}$ is the set of all pairs
$(\nu_1,\nu_2)$ of $\SVars$-valuations such that the concatenation
$\nu_1(\bar x)\conc\nu_2(\bar y)$ belongs to $D(R)$, and $\nu_1$
and $\nu_2$ agree outside the variables in $\bar y$.  

\item
$\semV {(x=y)}$ is the set of all identical pairs $(\nu, \nu)$ of 
$\SVars$-valuation such that $\nu(x)=\nu(y)$.
\item
Likewise, 
$\semV {(x=c)}$ is the set of all identical pairs $(\nu, \nu)$ of 
$\SVars$-valuation such that $\nu(x)=c$.
\item
$\semV {(x:=y)}$ is the set of all pairs
$(\nu_1,\nu_2)$ of $\SVars$-valuations such that $\nu_2 =
\nu_1[x:=\nu_1(y)]$.
Thus, $\nu_2(x)=\nu_1(y)$ and $\nu_2$ agrees with $\nu_1$ on all other
variables.
\item
Similarly,
$\semV {(x:=c)}$ is the set of all pairs
$(\nu_1,\nu_2)$ of $\SVars$-valuations such that $\nu_2 = \nu_1[x:=c]$.
\end{enumerate}
\end{defi}
Note that each $\semV\tau$, being a set of ordered pairs of valuations,
is a \emph{binary relation on valuations}.
When $\SVars$ is understood, we will feel free to omit
the superscript in $\semV{\tau}$.
\begin{exa}\label{ex:bustrain1}
Consider a set of variables $\SVars = \{x, y, z\}$ and a schema $\Sch$ 
with two binary relation names $B$ and $T$, both of input arity one.  
In the rest of the example, assume 
that $\dom \supseteq \{1,2,3,4,5\}$ and 
that we have an instance $\inst$ of $\Sch$ that 
assigns the relation names to the following binary relations:
\[\inst(B) = \{(1,2), (1,3), (2,3), (3,5)\} \text{ and }
\inst(T) = \{(1,4), (3,5)\}.\]
Intuitively, you could think of $B$ and $T$ as relations of 
source-destination pairs of stations that could be reached 
by bus ($B$) or train ($T$) respectively.

Figure~\ref{fig:dbExample} shows a tiny fragment of the graph view
of $D$.  It shows only three valuations and all labeled edges 
between these three valuations.  We depict valuations by 
three consecutive squares with the first being the value of $x$, the
second being the value of $y$, and the third being the value of $z$. 
\begin{figure}
    \centering
    \includegraphics[width=0.92\textwidth]{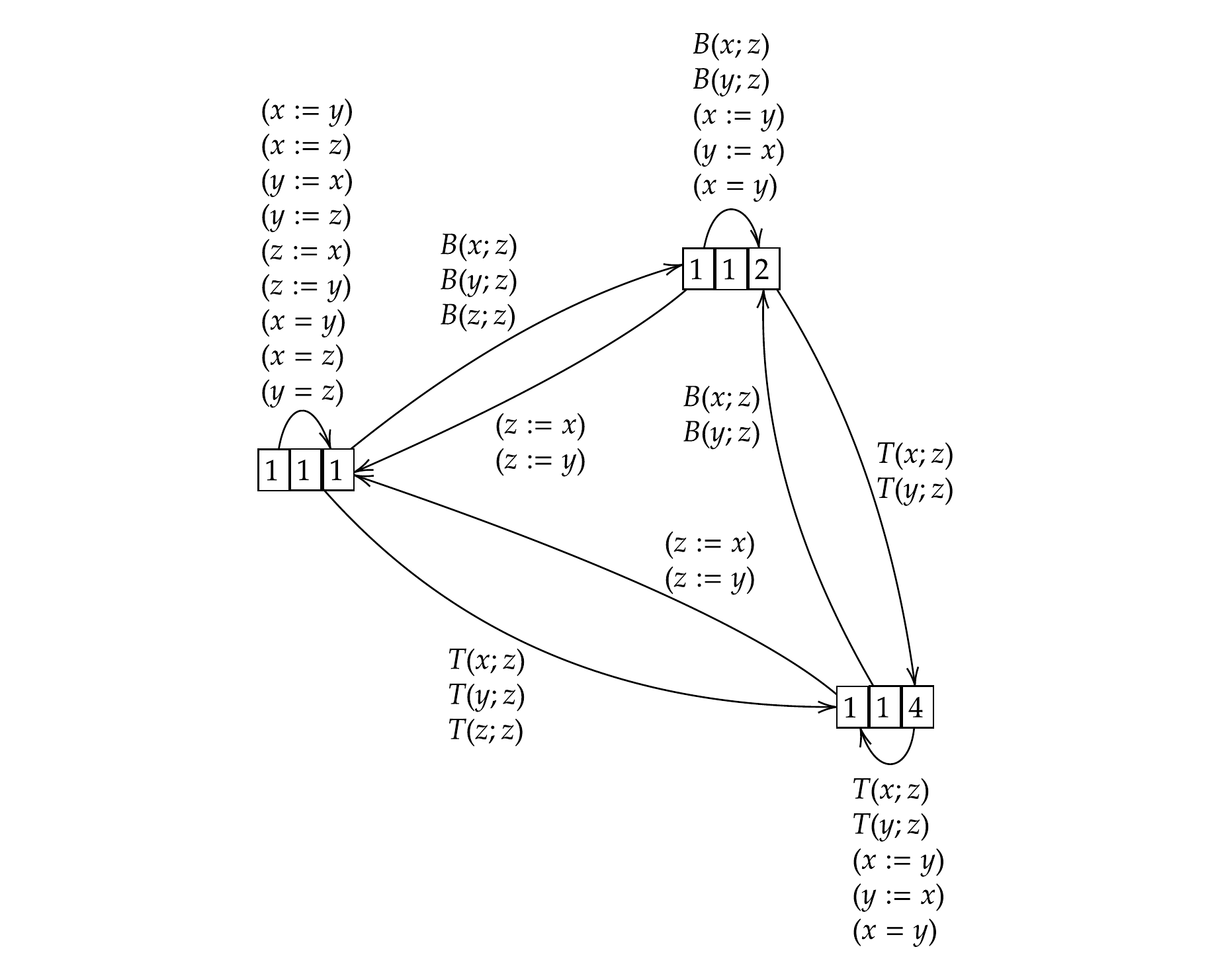}
    \caption{Part of graph view of the database considered
    in Example~\ref{ex:bustrain1}.}
    \label{fig:dbExample}
\end{figure}

For another illustration of the same graph view, let us consider the 
following atomic expressions: $B(x;x)$, $B(x;y)$, $T(y;z)$, $(x:=z)$, 
and $(x=z)$.  Figure~\ref{fig:tableview} depicts the entire 
binary relations on valuations $\semV{\tau}$ for these five 
atomic expressions $\tau$.
For each of these examples, we give a table below that shows its 
semantics (i.e., pairs of $\SVars$-valuations).  In this depiction, 
`$*$' means that the value of the variable could be anything in 
the domain, i.e., the variable in that valuation is not restricted 
to a specific value.  
Furthermore, when, in some pair, we put `$-$' in both slots for 
some variable $u$, we mean that the value for $u$ could be anything
on condition that it is the same on the left and right.
\begin{figure}
    \centering
\[
\begin{array}{ccc|ccc}
\multicolumn{6}{c}{\semV{B(x;x)}} \\
\multicolumn{3}{c}{\nu_1} & 
\multicolumn{3}{c}{\nu_2} \\
\toprule
x & y & z & x & y & z \\
\midrule
1 & - & - & 2 & - & - \\
1 & - & - & 3 & - & - \\
2 & - & - & 3 & - & - \\
3 & - & - & 5 & - & - \\
\bottomrule
\end{array}
\qquad
\begin{array}{ccc|ccc}
\multicolumn{6}{c}{\semV{B(x;y)}} \\
\multicolumn{3}{c}{\nu_1} & 
\multicolumn{3}{c}{\nu_2} \\
\toprule
x & y & z & x & y & z \\
\midrule
1 & * & - & 1 & 2 & - \\
1 & * & - & 1 & 3 & - \\
2 & * & - & 2 & 3 & - \\
3 & * & - & 3 & 5 & - \\
\bottomrule
\end{array}
\]
\[
\begin{array}{ccc|ccc}
\multicolumn{6}{c}{\semV{T(y;z)}} \\
\multicolumn{3}{c}{\nu_1} & 
\multicolumn{3}{c}{\nu_2} \\
\toprule
x & y & z & x & y & z \\
\midrule
- & 1 & * & - & 1 & 4 \\
- & 3 & * & - & 3 & 5 \\
\bottomrule
\end{array}
\]
\[
\begin{array}{ccc|ccc}
\multicolumn{6}{c}{\semV{(x:=z)}} \\
\multicolumn{3}{c}{\nu_1} & 
\multicolumn{3}{c}{\nu_2} \\
\toprule
x & y & z & x & y & z \\
\midrule
* & - & 1 & 1 & - & 1 \\
* & - & 2 & 2 & - & 2 \\
* & - & 3 & 3 & - & 3 \\
* & - & 4 & 4 & - & 4 \\
* & - & 5 & 5 & - & 5 \\
\bottomrule
\end{array}
\qquad
\begin{array}{ccc|ccc}
\multicolumn{6}{c}{\semV{(x=z)}} \\
\multicolumn{3}{c}{\nu_1} & 
\multicolumn{3}{c}{\nu_2} \\
\toprule
x & y & z & x & y & z \\
\midrule
1 & - & 1 & 1 & - & 1 \\
2 & - & 2 & 2 & - & 2 \\
3 & - & 3 & 3 & - & 3 \\
4 & - & 4 & 4 & - & 4 \\
5 & - & 5 & 5 & - & 5 \\
\bottomrule
\end{array}
\]
    \caption{Table view for the expressions considered
    in Example~\ref{ex:bustrain1}.}
    \label{fig:tableview}
\end{figure}
\end{exa}

\paragraph{Syntax and semantics of FLIF: operators}
The syntax of all FLIF expressions $\alpha$ (still over schema $\Sch$
and set of variables $\SVars$) is now given by the
following grammar:
\[\alpha ::= \tau \mid \alpha\comp \alpha \mid \alpha \cup \alpha
\intersection{\mid \alpha \cap \alpha }\mid \alpha - \alpha \]
Here, $\tau$ ranges over atomic expressions over $\Sch$ and $\SVars$,
as defined above.  The semantics of the composition operator `$;$' is
defined as follows:
\[\semV {\alpha_1 \comp \alpha_2} = \{(\nu_1,\nu_2) \mid \exists
\nu : (\nu_1,\nu) \in \semV{\alpha_1} \text{ and } 
(\nu,\nu_2) \in \semV{\alpha_2}\} \]
Note that we are simply taking the standard composition of 
two binary relations on valuations.
Similarly, the semantics of the set operations are standard union
\intersection{intersection }and set difference on binary relations on valuations.
\begin{exa}
Continuing Example~\ref{ex:bustrain1}, consider the expression 
$B(x;y) \comp T(y;x)$.
Intuitively, {this expression} takes $x$ as input and retrieves the possible values 
for $x$ and $y$ such that
\begin{enumerate}
\item you can go from station $x$ to station $y$ by a bus, 
and moreover,
\item you can go from station $y$ to \emph{a possibly different}
station $x$ by a train.
\end{enumerate}

The next table of pairs of valuations shows the semantics of that 
FLIF expression, i.e., $\semV{B(x;y) \comp T(y;x)}$.
\[
\begin{array}{ccc|ccc}
\multicolumn{3}{c}{\nu_1} & 
\multicolumn{3}{c}{\nu_2}  \\
\toprule
x & y & z & x & y & z \\
\midrule
1 & * & - & 5 & 3 & - \\
2 & * & - & 5 & 3 & - \\
\bottomrule
\vspace*{0.45cm}
\end{array}
\]
\end{exa}

We see that FLIF expressions describe paths in the graph, in the
form of source--target pairs.  Composition is used
to navigate through the graph, and to conjoin paths.
Paths can be branched using union, 
and excluded using set difference.

\begin{rem} Sometimes, in writing FLIF expressions, we omit 
parentheses around (sub)expressions involving composition since 
it is an associative operator.  Also, we give precedence to 
composition over the set operations. 
\end{rem}
\begin{exa} \label{friends1}
Consider a simple Facebook abstraction
with a single binary relation $F$ of input arity one.  When given 
a person as input, $F$ returns all their friends.
We assume that this relation is symmetric.

To illustrate the dynamic nature of FLIF, 
over just a single variable $\SVars=\{x\}$,
the expression $F(x;x); F(x;x); F(x;x)$ describes
all pairs $(\nu_1, \nu_2)$ such that there is a path
of length three from $\nu_1(x)$ to $\nu_2(x)$.

For another example, suppose, for an input person $x$ 
(say, a famous person),
we want to find all people who are friends with
at least two friends of $x$.  Formally, we want to navigate from
a valuation $\nu_1$ giving a value for $x$, to all valuations
$\nu_2$ giving values to variables $y_1$,
$y_2$, and $z$, such that
\begin{itemize}
\item
$\nu_2(y_1)$ and $\nu_2(y_2)$ are both friends with $\nu_1(x)$;
\item
$\nu_2(z)$ is friends with both $\nu_2(y_1)$ and $\nu_2(y_2)$;
and
\item
$\nu_2(y_1) \neq \nu_2(y_2)$.
\end{itemize}
This can be done by the FLIF expression 
$\alpha - (\alpha \comp (y_1=y_2))$, where
$\alpha$ is the expression
\[F(x;y_1)\comp F(x;y_2) \comp F(y_1;z) \comp F(y_2;z_1) \comp (z=z_1).\]
Note that using the extra variable $z_1$ is needed, since using 
$F(y_2;z)$ instead would result in overwriting the value of the variable 
$z$ set by the variable $y_1$.

Without the use of the extra variable, we could alternatively define 
$\alpha$ by the intersection $\alpha_1 \cap \alpha_2$,
where $\alpha_i$ is the expression
$F(x;y_1)\comp F(x;y_2) \comp F(y_i;z)$. 
We are using the intersection operator here, 
which is formally not part of FLIF as defined, 
but easily expressible as 
$\alpha_1 - (\alpha_1 - \alpha_2)$.
\end{exa}

\begin{rem}
In the above example, it would be more efficient to
simply write $\alpha \comp (y_1 \neq y_2)$.  For simplicity, we have
not added nonequality tests in FLIF as they are formally redundant
in the presence of set difference, but they can easily be added
in practice.  The purpose of this paper is to introduce the formal 
foundation of FLIF; clearly, a practical language based on FLIF
will include arithmetic comparisons and operations. 
\qed
\end{rem}

\paragraph{The evaluation problem for FLIF expressions}
Given that FLIF expressions navigate paths in the graph view
of a database instance $\inst$, {the} natural use of FLIF is to 
provide an input valuation $\nuin$ to an expression $\alpha$,
and ask for all output valuations $\nuout$ such that 
$(\nuin, \nuout) \in \sem\alpha$.  Formally, we define:
\begin{defi}[The evaluation problem $\EvalLV \alpha \SVars D \nuin$ for FLIF expression $\alpha$ over $\SVars$]
\label{defevalalphaV}
Given a database instance $D$ and a $\SVars$-valuation $\nuin$, 
the task is to compute the set 
\[\EvalLV \alpha \SVars D \nuin = \{\nuout \mid (\nuin,\nuout)
\in \semV\alpha\}.\]
\end{defi}

Assume we have effective access to the relations $R$ of $\inst$,
in the following sense.  If $R$ has input arity $i$ and 
output arity $o$, and given an $i$-tuple $t_1$, we can effectively 
retrieve the set of $o$-tuples $t_2$ such that $t_1 \cdot t_2$
belongs to $\inst(R)$.  Moreover, this set is assumed to be finite.
Assuming such effective access, which is needed for the evaluation
of atomic expressions of the form $R(\bar x; \bar y)$, it is now
obvious how more complex expressions can be evaluated.
Indeed, other atomic expressions are just assignments or tests,
and operations of FLIF are standard operations on binary relations.
In Section~\ref{secplan}, we will give an explicit description of this 
evaluation algorithm, for the ``io-disjoint'' fragment of FLIF,
in terms of relational algebra plans.  Nevertheless, the obvious
evaluation algorithm described informally above can always be applied,
also for FLIF expressions that are not io-disjoint.

\begin{exa}\label{friends3}
Recall the expression $F(x;x) \comp F(x;x) \comp F(x;x)$
from Example~\ref{friends1} over $\SVars = \{x\}$.
On input a valuation $\nuin$ on $\{x\}$, the evaluation will 
return all valuations $\nuout$ on $\{x\}$ such that there is
a path of length three from $\nuin(x)$ to $\nuout(x)$.

Next recall the expression 
$F(x;y_1)\comp F(x;y_2)\comp F(y_1;z) \comp F(y_2;z_1) \comp (z=z_1)$. 
On input valuation $\nuin$ on $\SVars = \{x,y_1,y_2,z,z_1\}$, the
evaluation will return all $\SVars$-valuations $\nuout$ such that
the tuples
$(\nuin(x),\nuout(y_1))$, $(\nuin(x),\nuout(y_2))$,
$(\nuout(y_1),\nuout(z))$, $(\nuout(y_2),\nuout(z))$
belong to relation $F$, and moreover $\nuout(z_1)=\nuout(z)$
and $\nuout(x)=\nuin(x)$.
Note in particular that the values
provided by $\nuin$ for $y_1$, $y_2$, $z$, and $z_1$
are irrelevant; only the input value $\nuin(x)$ counts.
Similarly, recall the expression
\[F(x;y_1)\comp F(x;y_2)\comp F(y_1;z) \cap 
F(x;y_1)\comp F(x;y_2)\comp F(y_2;z)\]
over $\SVars = \{x,y_1,y_2,z\}$.  On input a 
$\SVars$-valuation $\nuin$, the evaluation will return
all $\SVars$-valuations $\nuout$ such that the tuples
$(\nuin(x),\nuout(y_1))$, $(\nuin(x),\nuout(y_2))$,
$(\nuout(y_1),\nuout(z))$, $(\nuout(y_2),\nuout(z))$
belong to the relation $F$, and moreover $\nuout(x)=\nuin(x)$.

In contrast, consider the expression
\[F(x;y_1)\comp F(y_1;z) \cap F(x;y_2)\comp F(y_2;z).\]
Now on input a valuation $\nuin$ on $\SVars=\{x,y_1,y_2,z\}$,
the evaluation will return all $\SVars$-valuations $\nuout$
such that the tuples
$(\nuin(x),\nuin(y_1))$, $(\nuin(x),\nuin(y_2))$,
$(\nuin(y_1),\nuout(z))$, $(\nuin(y_2),\nuout(z))$
belong to the relation $F$, and moreover, $\nuout$ and $\nuin$
agree on $\{x,y_1,y_2\}$.  So for this expression, 
not just $\nuin(x)$, but also $\nuin(y_1)$ and $\nuin(y_2)$
are important values for the evaluation problem.
This behavior can be traced back to Definition~\ref{def:flif1},
which requires $\nu_1(y_2)=\nu_2(y_2)$ for any pair 
$(\nu_1, \nu_2) \in \semV{F(x;y_1)}$ as well as $\semV{F(y_1;z)}$.
Similarly, $\nu_1(y_1)=\nu_2(y_1)$ for any pair 
$(\nu_1, \nu_2) \in \semV{F(x;y_2)}$ as well as $\semV{F(y_2;z)}$.
\end{exa}

\section{Executable FO} \label{secmodel}
Let us recall the language known as executable FO (cf.~the Introduction).
Executability of formulas is a syntactic notion.  In the literature,
a lot of work has focused on the problem of trying to rewrite 
arbitrary FO formulas into executable form {\cite{nl_accesspatterns,rsu_bindingpatterns,lilimited,nl_access_constraints,cali_query_limit,blt_query_access_constraints,benedikt_planproofs_tods,
cali_deepweb_amw}}.
However, in this paper, we are focusing instead on using executable
FO as a gauge for accessing the expressiveness of our new language FLIF.
(Indeed, we will show that FLIF and executable FO are equivalent.)
Hence, in this paper, we work only with executable FO formulas and 
not with arbitrary FO formulas.

The notion of when a formula is executable is defined relative to
a set of variables $\VarSet$, which specifies the variables for which
input values are already given.  
Beware (in line with established work in the area \cite{lilimited,nl_accesspatterns}) that the notion of executability here is syntactic, and dependent on how subformulas are ordered within the formula.
One may think of the notion of executability discussed in 
this paper as a ``left-to-right'' executability, which shall be clear from the following examples.
Indeed, we begin with a few examples.

\begin{exa}\hfill
\begin{itemize}
\item
Let $\varphi$ be 
the formula $R(x;y)$. As mentioned above, this notation
makes clear that the input arity of $R$ is one. If we provide an
input value for $x$, then the database will give us all $y$ values 
such that $R(x,y)$ holds.  Indeed, $\varphi$ will turn out to be
$\{x\}$-executable.  Giving a value for the first argument of $R$
is mandatory, so $\varphi$ is neither
$\emptyset$-executable nor $\{y\}$-executable.
However, it is certainly allowed to provide input values for both $x$
and $y$; in that case we are merely testing if
$R(x,y)$ holds for the given pair $(x,y)$.
Thus, $\varphi$ is also $\{x,y\}$-executable.
In general, a $\VarSet$-executable formula will also be
$\VarSet'$-executable for any $\VarSet'\supseteq\VarSet$.
\item
Also, the formula $\exists y \, R(x;y)$ is $\{x\}$-executable.
In contrast, the formula $\exists x \, R(x;y)$ is not, because
even if a value for $x$ is given as input, it will be ignored due to the
existential quantification.  In fact, the latter formula is not
$\VarSet$-executable for any $\VarSet$.
\item
The formula $R(x;y) \land S(y;z)$ is $\{x\}$-executable,
intuitively because each $y$ returned by the formula $R(x;y)$ can
be fed into the formula $S(y;z)$, which is $\{y\}$-executable in itself.
In contrast, the \emph{semantically equivalent}
formula $S(y;z) \land R(x;y)$ is not $\{x\}$-executable,
because we need a value for $y$ to execute the formula $S(y;z)$.
However, the entire formula is $\{y, x\}$-executable.
\item
The formula $R(x;y) \lor S(x;z)$ is not $\{x\}$-executable,
because any $y$ returned by $R(x;y)$ would already satisfy the
formula, leaving the variable $z$ unconstrained.
This would lead to an infinite number of satisfying valuations.
The formula is neither $\{x,z\}$-executable;
if $S(x,z)$ holds for the given values for $x$ and $z$,
then $y$ is left unconstrained.  Of course, the formula is
$\{x,y,z\}$-executable.
\item
For a similar reason, $\neg R(x;y)$ is only $\VarSet$-executable for
$\VarSet$ containing $x$ and $y$.
\end{itemize}
\end{exa}


\paragraph{$\VarSet$-executable Formulas}
We now define, formally, for any set of variables $\VarSet$, the set of $\VarSet$-executable 
formulas are defined as follows.  Our definition closely follows
the original definition by Nash and Lud\"{a}scher {\cite{nl_accesspatterns}}; we only add equalities and constants to the language.
	\begin{itemize}
		\item An equality $x=y$, for variables $x$ and
		$y$, is \vex{} if at least one of $x$ and $y$ belongs to $\VarSet$.
		\item An equality $x=c$, for a variable $x$ and
		a constant $c$, is always \vex{}.
		\item A relation atom $R(\bar{x};\bar{y})$ is \vex{}
		if $X \subseteq \VarSet$, where $X$ is the set of
		variables from $\bar x$.
		\item A negation $\lnot \varphi$ is \vex{} if $\varphi$ is, and moreover $\FV(\varphi) \subseteq \VarSet$.
		\item A conjunction $\varphi \land \psi$ is
		\vex{} if $\varphi$ is, and moreover
		$\psi$ is \vex{\cup \FV(\varphi)}.
		\item A disjunction $\varphi \lor \psi$ is \vex{} if both $\varphi$ and $\psi$ are, and moreover $\FV(\varphi) \symdif \FV(\psi) \subseteq \VarSet$. Here, $\symdif$ denotes symmetric difference.
		\item An existential quantification $\exists x \, \varphi$ is \vex{} if $\varphi$ is \vex{\setminus \{x\}}.
\end{itemize}
Note that universal quantification is not part of the syntax of
executable FO.

\begin{exa}\label{friends2}
Recall the query considered in Example~\ref{friends1}, asking
for all triples $(y_1, y_2, z)$ such that, for some input $x$,
we have $F(x; y_1)$, $F(x; y_2)$, $F(y_1; z)$, $F(y_2; z)$,
and $y_1$ and $y_2$ are different.  The natural FO formula 
for this query is indeed $\{x\}$-executable:
\[F(x;y_1) \wedge F(x;y_2) \wedge F(y_1;z) \wedge F(y_2;z) \wedge \lnot (y_1=y_2).\]

Note that the above executable FO formula and FLIF expression from
Example~\ref{friends1} are quite similar in their structure.  The main 
difference is the use of the extra variable $z_1$ which was explained 
in Example~\ref{friends1}.
\end{exa}

\begin{rem} 
Continuing Remark~\ref{remarkmultiple}, in an extended setting where multiple access patterns are possible for the same relation, the simple syntax we use both in FLIF and in executable FO needs to be changed. Instead of relation atoms of the form $R(\bar x; \bar y)$ we would use adornments, which is a standard syntax in the literature on access limitations.  For example, if a ternary relation $R$ can be accessed by giving inputs to the first two arguments, or to the first and the third, then both $R^{iio}(x,y,z)$ and $R^{ioi}(x,y,z)$ would be allowed relation atoms.
\end{rem}

Given an FO formula $\varphi$ and a finite set
of variables $\VarSet$ such that $\varphi$ is $\VarSet$-executable,
we describe the following task:
\begin{defi}[The evaluation problem $\EvalV \varphi \VarSet D \nuin$ for $\varphi$
with input variables $\VarSet$] 
\label{defevalexfo}
Given a database instance $D$ and a valuation $\nuin$ on $\VarSet$,
compute the set of all valuations $\nu$ on $\VarSet \cup
\FV(\varphi)$ such that $\nuin \subseteq \nu$ and $\sat D\varphi \nu$.
\end{defi}

As mentioned in the Introduction, this problem is
known to be solvable by a relational algebra plan respecting the
access patterns.  In particular, if $D$ is finite, the output is always 
finite: each valuation $\nu$ in the output can be shown to 
take only values in $\adom(D) \cup \nuin(\VarSet)$.\footnote{
Actually, a stronger property can be shown:
only values that are ``accessible'' from $\nuin$ in $D$ can be
taken \cite{benedikt_book}, and if this accessible set is finite,
the output of the evaluation problem is finite.}

\subsection{From Executable FO to FLIF}

After introducing FLIF and executable FO, we observe that
executable FO formulas translate rather nicely to FLIF expression
as given by the following Theorem.
\begin{thm} \label{exfo2flif}
Let $\varphi$ be a $\VarSet$-executable formula over a schema $\Sch$.
There exists an FLIF expression $\alpha$ over $\Sch$ and a set of variables
$\SVars \supseteq \FV(\varphi) \cup \VarSet$ such that 
for every $D$, $\VarSet$-valuation $\nuin$, 
and $\SVars$-valuation $\nuin'$ with $\nuin' \supseteq \nuin$, we have
\[
\EvalV \varphi \VarSet D \nuin= 
\{ \nuout|_{\FV(\varphi) \cup \VarSet} \mid (\nuin', \nuout) \in \semV{\alpha}\}.
\]
\end{thm}

\begin{exa}
Before giving the proof, we give a few examples.
Note that in all the following examples, we only consider sets of input variables
$\VarSet$ with $\VarSet \subseteq \FV(\varphi)$.
\begin{itemize}
\item
Suppose $\varphi$ is $R(x;y)$ with input variable $x$.  
Then, as expected, $\alpha$ can be taken to be $R(x;y)$.
Suppose we have the same formula with $\VarSet = \{x, y\}$.
Intuitively, the formula asks for outputs $(u)$ where $u$ equals $y$.
Hence, $\alpha$ can be taken to be $R(x;u); (u=y)$.
Note that the FLIF expression $R(x;y)$ is not a correct translation 
since the value of $y$ may change from the value given by $\VarSet$.
\item
Now, consider $T(x;x,y)$, again with input variable $x$.
Intuitively, the formula asks for outputs $(u,y)$ where $u$ equals $x$.
Hence, a suitable FLIF translation is $T(x;u,y)\comp(u=x)$.
Note that the FLIF expression $T(x;x,y)$ is not semantically equivalent since
the value of $x$ is changeable due to the dynamic semantics of FLIF.
\item
If $\varphi$ is $R(x;y) \land S(y;z)$, still with input variable
$x$, we can take $R(x;y) \comp S(y;z)$ for $\alpha$.  The same
expression also serves for the formula $\exists y\, \varphi$.
\item
Suppose $\varphi$ is $R(x;x) \lor S(y;)$ with $\VarSet=\{x,y\}$.  For
$\VarSet \cap \FV(R(x;x))$, we translate $R(x;x)$ to $R(x;u)\comp(x=u)$.
Similarly, $S(y;)$ is translated to $S(y;)$.  Then, the final $\alpha$ 
can be taken to be $R(x;u)\comp(x=u) \cup S(y;)$.
\item
A new trick must be used for negation.  For example,
if $\varphi$ is $\neg R(x;y)$ with $\VarSet=\{x,y\}$, then $\alpha$
can be taken to be $(u:=42) \, - \, R(x;u)\comp(u=y)\comp(u:=42)$.
Composing each side of `$\setminus$' with the same dummy assignment to $u$
is required since the value of the $u$ in the second operand should
not affect the result of the needed negation.
\end{itemize}
\end{exa}
\begin{proof}[Proof Sketch of Theorem~\ref{exfo2flif}]
We only describe the translation; its correctness is proven in 
Section~\ref{proofexfo2flif}.

If $\varphi$ is a relation atom $R(\bar x;\bar y)$, then $\alpha$
is $R(\bar x;\bar z) \comp \xi$, where $\bar z$ is 
obtained from $\bar y$ by replacing each variable from $\VarSet$ by a 
fresh variable.  The expression $\xi$ consists of the composition 
of all equalities $(y_i = z_i)$ where $y_i$ is a variable from 
$\bar y$ that is in $\VarSet$ and $z_i$ is the corresponding fresh
variable.

If $\varphi$ is $x=y$, then $\alpha$ is $(x=y)$.

If $\varphi$ is $x=c$, then $\alpha$ is $(x=c)$.

If $\varphi$ is $\varphi_1 \land \varphi_2$, then by induction 
we have an expression $\alpha_1$ for $\varphi_1$ and 
$\VarSet \cap \FV(\varphi_1)$, and 
an expression $\alpha_2$ for $\varphi_2$ and 
$(\VarSet \cup \FV(\varphi_1)) \cap \FV(\varphi_2)$. 
Now $\alpha$ can be taken to be $\alpha_1 \comp \alpha_2$.

If $\varphi$ is $\exists x \, \varphi_1$, then without loss of 
generality we may assume that $x \notin \VarSet$.
By induction, we have an expression $\alpha_1$ for $\varphi_1$ and 
$\VarSet$.  This expression also works for $\varphi$.

If $\varphi$ is $\varphi_1 \lor \varphi_2$, then by induction we
have an expression $\alpha_i$ for $\varphi_i$ and $\VarSet$, 
for $i=1,2$.  Now $\alpha$ can be taken to be $\alpha_1 \cup \alpha_2$.

Finally, if $\varphi$ is $\neg \varphi_1$, then by induction we
have an expression $\alpha_1$ for $\varphi_1$ and $\VarSet$.  Fix an
arbitrary constant $c$, and let
$\xi$ be the composition of all expressions $(z:=c)$ for $z\in
\var(\alpha_1) - \VarSet$.
(If that set is empty, we add an extra fresh variable.)
Then $\alpha$ can be taken to be $\xi \; - \; \alpha_1\comp\xi$. 
\end{proof}

\subsection{From FLIF to executable FO}

The previous translation shows that FLIF is expressive enough, 
in the sense that executable FO formulas can be translated into FLIF
expressions such that they evaluate to the same set of valuations
starting from the same assignment. 
It turns out that the converse translation is also possible, 
so, FLIF exactly matches executable FO in expressive power.  

Actually, two distinct translations from FLIF to executable FO are 
possible:
\begin{enumerate}
    \item A somewhat rough translation, which translates every
    FLIF expression on a set of variables $\SVars$ to an
    equivalent $\SVars$-executable formula that uses thrice 
    the number of variables in $\SVars$; 
    \item A much finer translation, which often results in 
    $\VarSet$-executable formulas with a much smaller set $\VarSet$ 
    than the entire $\SVars$.  
    This set $\VarSet$ will consist of the ``input variables'' of 
    the given FLIF expression.  We will this idea further 
    in Sections~\ref{secio} and~\ref{seceval}.
\end{enumerate}
Next, we proceed with the rough translation.
Assume $\SVars=\{x_1,\ldots,x_n\}$.  Since the semantics of FLIF expressions
on $\SVars$ involves pairs of $\SVars$-valuations, we 
introduce a copy $\SVars_y = \{y_1,\ldots,y_n\}$ disjoint from $\SVars$.
For clarity, we also write $\SVars_x$ for $\SVars$.
By $\FO{k}$ we denote the fragment of first-order logic that uses 
only $k$ distinct variables~\cite{libkin_fmt}.

The following proposition is a variant of a result shown in our
companion paper on LIF~\cite[Proposition 7.9]{lif-tocl}.  That result
is for a larger language LIF, but it does not talk about executability.
\begin{prop}\label{prop:FLIF-BoundedFO}
Let $\Sch$ be a schema, and $\SVars_x$ a set of $n$ variables.
Then, for every FLIF expression $\alpha$ over $\Sch$ and $\SVars_x$, 
there exists a 
$\SVars_x$-executable $\FO{3n}$ formula $\varphi_\alpha$ over $\Sch$
with free variables in $\SVars_x\cup \SVars_y$ such that
\[
(\nu_1,\nu_2)\in \semV{\alpha} \quad \leftrightarrow \quad 
\inst{},(\nu_1\cup\nu_2')\models \varphi_\alpha,
\]
where $\nu_2'$ is the $\SVars_y$-valuation such that $\nu_2'(y_i) = \nu_2(x_i)$ 
for $i=1,\ldots,n$.
\end{prop}
\begin{proof}
 The proof is by induction on the structure of $\alpha$.
 First, we introduce a third copy $\SVars_z= \{z_1,\ldots,z_n\}$ of $\SVars$.
 Moreover, for every $u,v\in \{x,y,z\}$ we define $\rho_{uv}$ as follows:
\begin{align*}
   \rho_{uv} :\SVars_u \rightarrow \SVars_v: u_i \mapsto v_i
\end{align*}
Using these functions, we can translate a valuation $\nu$ on $\SVars=\SVars_x$ 
to a corresponding valuation on $\SVars_u$ with $u \in \{y,z\}$. 
Clearly, the function composition in $\nu \circ \rho_{ux}$ does this job.

In the first part of the proof, we actually show a stronger statement by induction, 
namely that for each $\alpha$ and for every $u\neq v \in \{x,y,z\}$ 
there is a formula $\varphi^{uv}_\alpha$ in $\FO{\SVars_x\cup \SVars_y\cup \SVars_z}$
with set of free variables equal to $\SVars_u\cup \SVars_v$
such that for every  $\inst{}$,
\[
(\nu_1,\nu_2)\in \semV{\alpha}
\quad \leftrightarrow \quad  \inst{},(\nu_1\circ \rho_{ux}\cup\nu_2\circ \rho_{vx})\models \varphi^{uv}_\alpha.
\]
Since the notations $x$, $y$, $z$, $u$ and $v$ are taken, 
we use notations $a$, $b$ and $c$ for variables and $d$ for constants.
\begin{itemize}
  \item $\alpha = R(\overline{a};\overline{b})$.
  Take $\varphi^{uv}_\alpha$ to be 
  $R(\rho_{xu}(\overline{a});\rho_{xv}(\overline{b}))\land 
  \bigwedge_{c \not \in \overline{b}} \rho_{xu}(c)= \rho_{xv}(c)$.
  \item $\alpha = (a = b)$.  Take $\varphi^{uv}_\alpha$ to be 
  $\rho_{xu}(a)= \rho_{xu}(b) \land 
  \bigwedge_{c \in \SVars_x} \rho_{xu}(c)= \rho_{xv}(c)$.
  \item $\alpha = (a = d)$.  
  Take $\varphi^{uv}_\alpha$ to be 
  $\rho_{xu}(a)= d \land 
  \bigwedge_{c \in \SVars_x} \rho_{xu}(c)= \rho_{xv}(c)$.
  \item $\alpha = (a := b)$.  
  Take $\varphi^{uv}_\alpha$ to be 
  $\rho_{xv}(a)= \rho_{xu}(b) \land 
  \bigwedge_{c \in \SVars_x\setminus\{a\}} \rho_{xu}(c)= \rho_{xv}(c)$.
  \item $\alpha = (a := d)$.  
  Take $\varphi^{uv}_\alpha$ to be 
  $\rho_{xv}(a)= d \land 
  \bigwedge_{c \in \SVars_x\setminus\{a\}} \rho_{xu}(c)= \rho_{xv}(c)$.
  \item $\alpha = \alpha_1 \cup \alpha_2$.
  Take $\varphi_\alpha^{uv}$ to be 
  $\varphi^{uv}_{\alpha_1} \lor \varphi^{uv}_{\alpha_2}$.
  \item $\alpha = \alpha_1 \setminus \alpha_2$.
  Take $\varphi_\alpha^{uv}$ to be 
  $\varphi_{\alpha_1}^{uv}\land\neg \varphi_{\alpha_2}^{uv}$.
  \item $\alpha = \alpha_1\comp{}\alpha_2$. 
  Let $w \in \{x,y,z\}\setminus \{u,v\}$. 
  Take $\varphi_\alpha^{uv}$ to be
  $\exists w_1\ldots\exists w_n\;(\varphi_{\alpha_1}^{uw}\land \varphi_{\alpha_2}^{wv})$.
\end{itemize}

In the rest of the proof, we verify that $\varphi^{uv}_\alpha$ is indeed
$\SVars_u$-executable.  As for the atomic FLIF expressions, this is clear.

In case $\alpha = \alpha_1 \cup \alpha_2$, we know by induction that both
$\varphi^{uv}_{\alpha_1}$ and $\varphi^{uv}_{\alpha_2}$ are 
$\SVars_u$-executable.  
For $\varphi^{uv}_{\alpha}$ to be $\SVars_u$-executable, 
it must be the case that
$\FV(\varphi^{uv}_{\alpha_1}) \symdif \FV(\varphi^{uv}_{\alpha_2}) \subseteq 
\SVars_u$ which is trivial since
$\FV(\varphi^{uv}_{\alpha_1}) = \SVars_u \cup \SVars_v =
\FV(\varphi^{uv}_{\alpha_2})$, so $\FV(\varphi^{uv}_{\alpha_1}) \symdif \FV(\varphi^{uv}_{\alpha_2}) = \emptyset$.

Now, consider the case $\alpha = \alpha_1 \setminus \alpha_2$.  We know
by induction that both $\varphi^{uv}_{\alpha_1}$ and $\varphi^{uv}_{\alpha_2}$ 
are $\SVars_u$-executable.  
For $\varphi^{uv}_{\alpha}$ to be $\SVars_u$-executable, 
it must be the case that
$\FV(\varphi^{uv}_{\alpha_2}) \subseteq \FV(\varphi^{uv}_{\alpha_1})$ 
which is true since
$\FV(\varphi^{uv}_{\alpha_1}) = \SVars_u \cup \SVars_v =
\FV(\varphi^{uv}_{\alpha_2})$.

Finally, consider the case $\alpha = \alpha_1 \comp \alpha_2$.  We know
by induction that $\varphi^{uw}_{\alpha_1}$ is $\SVars_u$-executable 
and $\varphi^{wv}_{\alpha_2}$ is $\SVars_w$-executable.  It is clear
that $\SVars_w \subseteq \FV(\varphi^{uw}_{\alpha_1})$, consequently,
the formula $\varphi_{\alpha_1}^{uw}\land \varphi_{\alpha_2}^{wv}$ is
$\SVars_u$-executable.  Hence, since $\SVars_u$ and $\SVars_u$ are
disjoint, the same formula is $(\SVars_u \setminus \SVars_w)$-executable 
which is sufficient to show
that $\varphi^{uv}_{\alpha}$ itself is $\SVars_u$-executable.
\qedhere
\end{proof}%

Although the previous translations show that FLIF and executable FO
are effectively equivalent in expressive power, 
the translation from FLIF to executable FO overlooks some of 
the interesting relations between both formalisms and moreover, 
it uses lots of variables unnecessarily.  This is best shown by
example.
\begin{exa}\label{ex:translation2}
Consider the FLIF expression 
$\alpha = R(x_1; x_1) \comp R(x_1; x_1) \comp R(x_1; x_1)$ where 
$\SVars_x = \{x_1\}$.
According to the procedure given in the proof of 
Proposition~\ref{prop:FLIF-BoundedFO}, 
the resultant $\varphi_\alpha$ would be 
\[
\exists z_1
\Bigl[ R(x_1;z_1) \land
\exists x_1
\bigl[ R(z_1;x_1) \land R(x_1;y_1) \bigr] \Bigr].
\]

In contrast, consider the FLIF expression 
$\alpha = R(x_1; x_2) \comp R(x_2; x_3) \comp R(x_3; x_4)$ where 
$\SVars_x = \{x_1, x_2, x_3, x_4\}$.
Now $\varphi_\alpha$ would be 
\begin{multline*}
\exists z_1 \exists z_2 \exists z_3 \exists z_4 
\Bigl[ R(x_1;z_2) \land (z_1=x_1) \land (z_3=x_3) \land (z_4=x_4) \land \\
\exists x_1 \exists x_2 \exists x_3 \exists x_4 
\bigl[ R(z_2;x_3) \land (x_1=z_1) \land (x_2=z_2) \land (x_4=z_4) \land \\
R(x_3;y_4) \land (y_1=x_1) \land (y_2=x_2) \land (y_3=x_3) \bigr] \Bigr].
\end{multline*}
However, it is clear that taking $\varphi_\alpha$ to be 
$R(x_1; x_2) \land R(x_2; x_3) \land R(x_3; x_4)$ would work fine,
in the sense that given an arbitrary $\{x_1\}$-valuation $\nu$ and any 
$\SVars_x$-valuation $\nu'$ that is an extension of $\nu$ 
(i.e., $\nu' \supseteq \nu$), $\alpha$ and $\varphi_\alpha$
would evaluate to the same set of $\SVars_x$-valuations as
stated below (where $D$ below is an arbitrary instance):
\[
\EvalV {\varphi_\alpha} {\{x_1\}} D \nu
=
\EvalLV \alpha {\SVars_x} D {\nu'}
\]
This shows that the values provided for $\{x_2, x_3, x_4\}$ in $\nu'$
to evaluate the expression $\alpha$ are not important since their
values would be overwritten regardless of what $\nu'$ sets them to.
Stated differently, variables $x_2$, $x_3$, $x_4$ are \emph{outputs} of 
$\alpha$, but not \emph{inputs}; the only input variables for $\alpha$
is $x_1$.
\end{exa}

In the next section, we develop the notions of input and output 
variables of FLIF expressions more formally.  Then in 
Section~\ref{secimproved}, we will give an improved translation 
from FLIF to executable FO taking inputs into account.

\section{Inputs and outputs of forward LIF} \label{secio}

In this section, we introduce inputs and outputs of 
FLIF expressions.  In every expression, we can identify
the \emph{input} and the \emph{output} variables.
Intuitively, the output variables are those that can change
value along the execution path; the input variables are those whose 
values at the beginning of the path are needed in order to know the
possible values for the output variables.  These intuitions will
be formalized below.  We first give some examples. \\
\begin{exa} \label{exio} \hfill
\begin{itemize}
\item
In both expressions given for $\alpha$ from Example~\ref{friends1}, 
the only input variable is $x$, and the other variables are output
variables.
\item
FLIF, in general, allows expressions where a variable is both
input and output.  For example, consider the relation $\mathit{Swap}$
of input arity two that holds of quadruples of the form $(a,b,b,a)$ for
$a$ and $b$ in the $\dom$,
so the values of the first two arguments are swapped in the second
two.  Then, using the expression $\mathit{Swap}(x,y;x,y)$ 
would result in having the values of $x$ and $y$ swapped.
Formally, this expression defines all pairs of valuations $(\nu_1,\nu_2)$
such that $\nu_2(x)=\nu_1(y)$ and $\nu_2(y)=\nu_1(x)$
(and $\nu_2$ agrees with $\nu_1$ on all other variables).
\intersection{\item
Consider the expression $R(x;y_1) \cap S(x;y_2)$.  Then not only
$x$, but also $y_1$ and $y_2$ are input variables.  Indeed,
the expression $R(x;y_1)$ will pair an input valuation $\nu_1$
with an output valuation $\nu_2$ that sets $y_1$ such that
$R(\nu_1(x),\nu_2(y_1))$ holds, but $\nu_2$ will have the same
value as $\nu_1$ on any other variable.  In particular,
$\nu_2(y_2)=\nu_1(y_2)$.  The expression $S(x;y_2)$ has a similar
behavior, but with $y_1$ and $y_2$ interchanged.  Thus,
the intersection expression checks two conditions on the
input valuation; formally, it defines all identical pairs
$(\nu,\nu)$ for which $R(\nu(x),\nu(y_1))$ and
$S(\nu(x),\nu(y_2))$ hold.  Since the expression only tests
conditions, it has no output variables.}
\item
On the other hand, for the expression $R(x;y_1) \cup S(x;y_2)$,
the output variables are
$y_1$ and $y_2$. Indeed, consider an input
valuation $\nu_1$ with $\nu_1(x)=a$. The expression
pairs $\nu_1$ either with a valuation
giving a new value for $y_1$, or with a valuation giving a new value
for $y_2$.  However, $y_1$ and $y_2$ are
also input variables (together with $x$).  Indeed, when pairing
$\nu_1$ with a valuation $\nu_2$ that sets $y_2$ to some $b$ for
which $S(a,b)$ holds, we must know the value of $\nu_1(y_1)$ so as
to preserve it in $\nu_2$.  A similar argument holds for $y_2$.
\qed
\end{itemize}
\end{exa}
The semantic properties that we gave above as intuitions for 
the notions of inputs and outputs are undecidable in general 
(see related work Section~\ref{secrel}).
Here, we will work with syntactic approximations.

\begin{defi}
For any FLIF expression $\alpha$, its sets $\In(\alpha)$ and 
$\Out(\alpha)$ of input and output variables are defined in Table~\ref{tabio}. 
\end{defi} 
Note that previously we have used $\var(\alpha)$ to denote
the set of all variables occurring in the expression $\alpha$.  Since FLIF has no explicit quantification, this
 is precisely the union of $I(\alpha)$ and $O(\alpha)$.
From now on, we will also refer to this set as the \emph{free variables}.

\begin{table}
\caption{Input and output variables of FLIF expressions.  In the
case of $R(\bar x;\bar y)$, the set $X$ is the set of variables
in $\bar x$, and the set $Y$ is the set of variables in $\bar
y$. Recall that $\symdif$ is symmetric difference.}
\label{tabio}
\centering
\begin{tabular}{c|l|l}
\toprule
$\alpha$ & $\In(\alpha)$ & $\Out(\alpha)$ \\
\toprule
$R(\bar{x};\bar{y})$ & $X$ & $Y$ \\
\midrule
$\eq{x}{y}$ & $\{x, y\}$ & $\emptyset$ \\
\midrule
$(x:=y)$ & $\{y\}$ & $\{x\}$ \\
\midrule
$\eq{x}{\cons}$ & $\{x\}$ & $\emptyset$ \\
\midrule
$\set{x}{\cons}$ & $\emptyset$ & $\{x\}$ \\
\midrule
$\alpha_1 ; \alpha_2$ &
$\In(\alpha_1) \cup (\In(\alpha_2) \setminus \Out(\alpha_1))$ &
$\Out(\alpha_1) \cup \Out(\alpha_2)$ \\
\midrule
$\alpha_1 \cup \alpha_2$ &
$\In(\alpha_1) \cup \In(\alpha_2) \cup (\Out(\alpha_1) \symdif \Out(\alpha_2))$ & $\Out(\alpha_1) \cup \Out(\alpha_2)$ \\
\midrule
$\alpha_1 \setminus \alpha_2$ &
$\In(\alpha_1) \cup \In(\alpha_2) \cup (\Out(\alpha_1) \symdif
\Out(\alpha_2))$ & $\Out(\alpha_1)$ \\
\bottomrule
\end{tabular}
\end{table}

Next we establish three propositions 
that show that our definition of inputs and outputs, which is 
purely syntactic, reflects actual properties of the semantics.
(See Section~\ref{secrel} on related work for their proofs.)

The first proposition confirms an intuitive property and
can be straightforwardly verified by induction.  

\begin{prop}[Inertia property] \label{propinertia}
If $(\nu_1,\nu_2) \in \sem\alpha$ then $\nu_2$ agrees with
$\nu_1$ outside $O(\alpha)$.
\end{prop}

The second proposition confirms, as announced earlier, that the
semantics of expressions depends only on the free variables;
outside $\var(\alpha)$, the binary relation $\sem \alpha$ is
\emph{cylindrical}, i.e., contains all possible data elements.\footnote{This terminology is borrowed from cylindrical set algebra \cite{il_cylindric,vdb_tarski}.}
An illustration of this was already given in Figure~\ref{fig:tableview}, using the asterisk indications.

\begin{prop}[Free variable property] \label{propfv}
Let $(\nu_1,\nu_2) \in
\sem\alpha$ and let $\nu_1'$ and $\nu_2'$ be valuations such that
\begin{itemize}
\item
$\nu_1'$ agrees with $\nu_1$ on $\var(\alpha)$, and
\item
$\nu_2'$ agrees with $\nu_2$ on $\var(\alpha)$, and agrees with $\nu_1'$
outside $\var(\alpha)$.
\end{itemize}
Then also $(\nu_1',\nu_2') \in \sem\alpha$.
\end{prop}

The third proposition is the most important one.  
It confirms that the values for the input 
variables determine the values for the output variables. 

\begin{prop}[Input-output determinacy] \label{propindet}
Let $(\nu_1,\nu_2) \in \sem\alpha$ and let $\nu_1'$ be a
valuation that agrees with $\nu_1$ on $I(\alpha)$.  Then there
exists a valuation $\nu_2'$ that agrees with $\nu_2$ on
$O(\alpha)$, such that $(\nu_1',\nu_2') \in \sem\alpha$.
\end{prop}
By inertia, we can see that the valuation
$\nu'_2$ given by the above proposition is unique.
Moreover, using the free variable property, we showed
the input-output determinacy property is equivalent to the following 
alternative form.
\begin{lem}[Input-output determinacy, alternative form]
\label{altidlemma}
Let $(\nu_1,\nu_2) \in \sem\alpha$ and let $\nu_1'$ be a
valuation that agrees with $\nu_1$ on $I(\alpha)$ as well as
outside $O(\alpha)$.  Then also $(\nu_1',\nu_2) \in \sem\alpha$.
\end{lem}
Intuitively, it is easier to work with the alternative form
since we have to consider only three valuations instead of
four in the original form.
\begin{exa}
Let us denote the expression $R(x;y) \cup S(x;z)$ by $\alpha$.
The definitions in Table~\ref{tabio} yield that $O(\alpha)=\{y,z\}$
and $I(\alpha)=\{x,y,z\}$.  Having $y$ and $z$ as input
variables may at first sight seem counterintuitive.
To see semantically why, say, $y$ is an input variable
for $\alpha$, consider an instance $D$ where
$S$ contains the pair $(1,3)$.  Consider the
valuation $\nu_1=\{(x,1),(y,2),(z,0)\}$, and let
$\nu_2=\nu_1[z:=3]$.  Clearly $(\nu_1,\nu_2) \in \sem{S(x;z)}
\subseteq \sem\alpha$.
However, if we change the value of $y$ in $\nu_1$, letting
$\nu_1'=\nu_1[y:=4]$, then $(\nu_1',\nu_2)$ neither belongs
to $\sem{S(x;z)}$ nor to $\sem{R(x;y)}$ (due to inertia).
Thus, input-output determinacy would be
violated if $y$ would not belong to $I(\alpha)$.
\end{exa}



We are now in a position to formulate a new version of the FLIF 
evaluation problem that takes the inputs into consideration.  
Given an expression $\alpha$, we consider the following
task:\footnote{For a valuation $\nu$ on a set of variables $X$
(possibly all variables), and a subset $Y$ of $X$, we use
$\nu|_Y$ to denote the restriction of $\nu$ to $Y$, 
i.e., $\nu|_Y$ is the function from the variables in $Y$ 
to $\dom$ that agrees with $\nu$ on $Y$.}

\begin{defi}[The evaluation problem $\EvalL \alpha D \nuin$ for $\alpha$]
\label{defevalalphainput}
Given a database instance $D$ and a valuation $\nuin$ on
$I(\alpha)$, the task is to compute the set 
\[\EvalL \alpha D \nuin = \{\nuout|_{\var(\alpha)} \mid \exists
\nuin' : \nuin \subseteq \nuin' \text{ and } (\nuin',\nuout)
\in \sem\alpha\}.\]
\end{defi}
By inertia and input-output determinacy, the choice of $\nuin'$ in the above
definition of the output does not matter.
We show this formally in the next Remark.
\begin{rem}
The above definition improves on Definition~\ref{defevalalphaV}
in that it is formally independent of the encompassing
universe $\SVars$ of variables; it intrinsically only depends 
on the input and output variables of $\alpha$.  Indeed, formally, 
given any FLIF expression $\alpha$ on $\SVars$, and any 
$\SVars$-valuation $\nu$, it is not hard to see that 
the following equivalence holds:
\[
\EvalL \alpha D {\nu|_{I(\alpha)}} = 
\{\nu'|_{\var(\alpha)} \mid \nu' \in \EvalLV \alpha \SVars D \nu \}
\]
\end{rem}
\begin{proof}
It suffices to show the `$\subseteq$' direction; the other direction
is clear from the definitions.  Let $\nu$ be a $\SVars$-valuation.
Suppose that there exists an arbitrary $\SVars$-valuation $\nuin$
such that $\nuin \supseteq \nu|_{I(\alpha)}$ (i.e., $\nuin = \nu$ 
on $I(\alpha)$) and $(\nuin, \nuout) \in \semV \alpha$.
We want to show that there exists a valuation $\nuout'$
such that $(\nu, \nuout') \in \semV \alpha$ and $\nuout = \nuout'$ 
on $\var(\alpha)$.

From the facts that $\nu = \nuin$ on $I(\alpha)$ and that
$(\nuin, \nuout) \in \semV \alpha$, it follows by input-output
determinacy that there exists a valuation $\nuout'$ such that
$(\nu, \nuout') \in \semV \alpha$ and $\nuout' = \nuout$ on $O(\alpha)$.
It remains to verify that $\nuout' = \nuout$ on $I(\alpha) - O(\alpha)$.
Indeed, this is true since $\nuout = \nuin = \nu = \nuout'$ on 
$I(\alpha) - O(\alpha)$, where the first and third equalities hold 
because of inertia and having both $(\nuin, \nuout)$ and 
$(\nu, \nuout')$ in $\semV \alpha$.  The middle equality 
follows from the fact that $\nu = \nuin$ on $I(\alpha)$.
\end{proof}

Consider an FLIF expression $\alpha$ for which the set $O(\alpha)$ is
disjoint from $I(\alpha)$. Then any pair $(\nu_1,\nu_2) \in
\sem\alpha$ satisfies that $\nu_1$ and $\nu_2$ are equal on
$I(\alpha)$.  Put differently, every $\nuout \in \EvalL \alpha D
\nuin$ is equal to $\nuin$ on $I(\alpha)$; all that the
evaluation does is expand the input valuation with output
values for the new output variables.  This makes the evaluation
process for expressions $\alpha$ where $I(\beta)\cap O(\beta)=\emptyset$,
for every subexpression $\beta$ of $\alpha$ (including $\alpha$
itself), very transparent in which input slots remain intact while
output slots are being filled.  We call such expressions
\emph{io-disjoint}.

\begin{exa} \label{exiod}
Continuing Example~\ref{friends1} (friends), the
expression $F(x;x)$ is obviously not io-disjoint.  Evaluating
this expression will overwrite the variable $x$ with a friend of the
person originally stored in $x$.  In contrast, both expressions given 
for $\alpha$ in Example~\ref{friends1} are io-disjoint.
Also the expression 
$R(x_1; x_2) \comp R(x_2; x_3) \comp R(x_3; x_4)$ 
from Example~\ref{ex:translation2} is io-disjoint.
Finally, the expression $R(x;y_1) \cup S(x;y2)$ already seen 
in Example~\ref{exio} is not io-disjoint.
\qed
\end{exa}

Formally, we have the following useful property, which follows from
inertia and input-output determinacy.
\begin{prop}[Identity property] \label{ip}
Let $\alpha$ be an io-disjoint expression and let $D$ be an
instance.  If $(\nu_1,\nu_2) \in \sem \alpha$, then also
$(\nu_2,\nu_2) \in \sem \alpha$.
\end{prop}

Intuitively, the identity property holds because, if in $\nu_1$
the output slots would accidentally already hold a correct
combination of output values, then there will exist an evaluation
of $\alpha$ that merely confirms these values.
This property can be interpreted to say that io-disjoint expressions
can be given a ``static'' semantics; we could say that a single
valuation $\nu$ satisfies $\alpha$ when $(\nu,\nu)$ belongs to
the dynamic semantics. This brings io-disjoint expressions closer
to the conventional static semantics (single valuations) of first-order
logic.  Indeed, this will be confirmed in the next Section.

\begin{exa}
The identity property clearly need not hold for expressions that
are not io-disjoint.  For example, continuing the friends
example, for the expression $F(x;x)$, a person need not be a
friend of themselves.
\end{exa}

The following proposition 
makes it easier to check if an expression is io-disjoint:
\begin{prop} \label{checkio}
The following alternative definition of io-disjointness is equivalent to the
definition given above:
\begin{itemize}
\item
An atomic expression
$R(\bar x;\bar y)$ is io-disjoint if $X\cap Y=\emptyset$, where
$X$ is the set of variables in $\bar x$, and
$Y$ is the set of variables in $\bar y$.
\item
Atomic expressions of the form $(x=y)$, $(x=c)$, $(x:=y)$ or $(x:=c)$ are
io-disjoint.
\item
A composition $\alpha_1\comp\alpha_2$ is io-disjoint if $\alpha_1$
and $\alpha_2$ are, and moreover $I(\alpha_1)\cap
O(\alpha_2)=\emptyset$.
\item
A union $\alpha_1 \cup \alpha_2$ is io-disjoint if $\alpha_1$ and
$\alpha_2$ are, and moreover $O(\alpha_1)=O(\alpha_2)$.
\intersection{\item
An intersection $\alpha_1 \cap \alpha_2$ is io-disjoint if $\alpha_1$ and
$\alpha_2$ are.}
\item
A difference $\alpha_1 - \alpha_2$ is io-disjoint if $\alpha_1$ and
$\alpha_2$ are, and moreover $O(\alpha_1)\subseteq O(\alpha_2)$.
\end{itemize}
\end{prop}

The fragment of io-disjoint expressions is denoted by $\flio$.
In the next section, we are going to show that $\flio$ is 
expressive enough, in the sense that FLIF expressions can be 
simulated by $\flio$ expressions that have the same set of input
variables. Furthermore, we will give the improved translation from $\flio$
to executable FO, which takes inputs into account.

\section{Io-disjoint FLIF} \label{seceval}
We begin this section by showing that any FLIF expression can be 
converted to an io-disjoint one.
We will first discuss the problem and its complications by means of 
illustrative examples.  After that, we formulate the precise theorem 
and give a constructive method to rewrite FLIF expressions into
io-disjoint ones.

\subsection{From FLIF to io-disjoint FLIF}

In this section, we are discussing a possible approach to
translate general FLIF expressions into io-disjoint ones that 
simulate the original expressions; we also discuss what
``simulate'' can mean.  For instance, we will see that we have to use extra 
variables in order to get io-disjointness.
An appropriate notion of simulation will then involve renaming of output variables.

For example, we rewrite $R(x;x)$ to
$R(x;y)$ and declare that the output value for $x$ can now be found
in slot $y$ instead.  This simple idea, however, is complicated
when handling the different operators of FLIF\@.  These
complications are discussed next.

\subsubsection{{Complications of Translation}}

When applying {the simple renaming} approach to the composition of two expressions, we must
be careful, as an output of the first expression can be taken as
input in the second expression.  In that case, when renaming the
output variable of the first expression, we must apply the
renaming also to the second expression, but only on the input
side.  For example, $R(x;x) \comp S(x;x)$ is rewritten to $R(x;y)
\comp S(y;z)$.  Thus, the output $x$ of the overall expression is
renamed to $z$; the intermediate output $x$ of the first
expression is renamed to $y$, as is the input $x$ of the second
expression.

Obviously, we must also avoid variable clashes.  For example, in
$R(x;x) \comp S(y;y)$, when rewriting the subexpression $R(x;x)$,
we should not use $y$ to rename the output $x$ to, as this
variable is already in use in another subexpression.

Another subtlety arises in the rewriting of set operations.
Consider, for example, the union $R(x;y) \cup S(x;z)$.  As discussed
in Example~\ref{exiod}, this expression is not io-disjoint: the output
variables are $y$ and $z$, but these are also input variables, in
addition to $x$.
To make the expression io-disjoint, it does not
suffice to simply rename $y$ and $z$, say, to $y_1$ and $z_1$.
We can, however, add assignments to both sides in such a way to
obtain a formally io-disjoint expression: \[R(x;y_1) \comp
(z_1:=z) \; \cup \; S(x;z_1) \comp (y_1:=y). \]

The above trick must also be applied to
intermediate variables.  For example, consider $T(;) \cup
(S(;y)\comp R(y;y))$.  Note that $T$ is a nullary relation.
This expression is not io-disjoint with
$y$ being an input variable as well as an output variable.
The second term is readily rewritten
to $S(;y_1) \comp R(y_1;y_2)$ with $y_2$ the new output
variable.  Note that $y_1$ is an intermediate variable.  The
io-disjoint form becomes \[T(;) \comp (y_1:=y) \comp (y_2:=y) \;
\cup \; S(;y_1) \comp R(y_1;y_2). \]
In general, it is not obvious that one can always find
a suitable variable to set intermediate variables from the
other subexpression to.  In our proof of the theorem we prove
formally that this is always possible.

\intersection{
Expressions involving intersection can be rewritten with the aid
of composition and equality test.  For example the expression
$R(x;y) \cap S(y;y)$ becomes \[R(x;y_1) \comp S(y;y_2) \comp
(y_1=y_2). \] The overall output $y$ is renamed to $y_1$, but in the
rewriting of the subexpression $S(y;y)$ we use an intermediate
output variable $y_2$, then test that the two outputs are the same as
required by the original expression.  Note that this usage of
intermediate variables is different from that used in the
treatment of composition expressions. There, the intermediate
variable was on the output side of the first subexpression and on
the input side of the second subexpression; here, it is on the
output side of the second subexpression.

An additional complication with intersection is that outputs that
are not common to the lhs and rhs subexpressions of the
intersection lose their status of output variable
(Table~\ref{tabio}), so must remain inertial for the overall
expression.  Hence, for the rewriting to have the
desired semantics, we must add the appropriate equality tests.
For example, the expression $R(x,y;x,y) \cap S(x,z;x,z)$ has only $x$
as an output, and $x$, $y$ and $z$ as inputs. Renaming the output
$x$ to $x_1$, it can be rewritten
in io-disjoint form as
\[R(x,y;x_1,y_1) \comp (y_1=y) \comp S(x,z;x_2,z_1) \comp (z_1=z)
\comp (x_2=x_1). \]
}

A final complication occurs in the treatment of difference.
Intermediate variables used in the rewriting must be reset to the
same value in both subexpressions, since the difference operator
is sensitive to the values of all variables.
For example, let $\alpha$ be the expression
$S(;x) \comp R(x;u,x) \: - \: T(;)$.
We have $I(\alpha)=O(\alpha)=\{x,u\}$.  Suppose we
want to rename the outputs $x$ and $u$ to $x_1$ and $u_1$
respectively.  As before, the subexpression on the lhs of the
difference operator is rewritten to
$S(;x_2) \comp R(x_2;u_1,x_1)$ introducing an intermediate
variable $x_2$.  Also as before, $x_1$ and $u_1$ need to be added
to the rewriting of $T(;)$ which does not have $x$ and $u$ as
outputs.  But the new complication is that $x_2$ needs to be
reset to a common value (we use $x$ here) for the difference of
the rewritten subexpressions to have the desired semantics.  We
thus obtain the overall rewriting
\[
S(;x_2) \comp R(x_2;u_1,x_1) \comp (x_2:=x) \: - \:
T(;) \comp (x_2:=x) \comp (u_1:=u) \comp (x_1:=x).
\]

\subsubsection{Statement of the theorem}

As the overall idea behind the above examples was to rename the
output variables, our aim is clearly the following theorem, with
$\rho$ playing the role of the renaming:\footnote{We
use $g \circ f$ for standard function composition (``$g$ after $f$'').
So, in the statement of the theorem, $\nu_2 \circ \rho :
O(\alpha) \to \VarUniv : x \mapsto \nu_2(\rho(x))$.}

\begin{thm} \label{main}
Let $\alpha$ be an FLIF expression and let $\rho$ be a bijection
from $O(\alpha)$ to a set of variables disjoint from $\var(\alpha)$.  There exists an $\flio$ expression $\beta$
such that
\begin{enumerate}
\item
$I(\beta) = I(\alpha)$;
\item
$O(\beta) \supseteq \rho(O(\alpha))$; and
\item
for every instance $D$ and
every valuation $\nu_1$, we have
\[
\{\nu_2|_{O(\alpha)} \mid (\nu_1,\nu_2) \in \semV \alpha\} =
\{\nu_2 \circ \rho \mid (\nu_1,\nu_2) \in \semV \beta\}.
\]
Here, $\SVars$ is any set of variables containing $\var(\alpha)$
and $\var(\beta)$.
\end{enumerate}
\end{thm}

In the above theorem, we must allow
$O(\beta)$ to be a superset of $\rho(O(\alpha))$ (rather than being
equal to it), because we must allow the introduction of auxiliary
(intermediate) variables.  For example, let $\alpha$ be the
expression $S(x;) - R(x;x)$.  Note that $O(\alpha)$ is empty.
Interpret $S$ as holding bus stops and $R$ as holding bus routes.  
Then $\alpha$ represents an information source with limited 
access pattern that takes as input $x$, and tests if $x$ is a 
bus stop to where the bus would not return if we would take the 
bus at $x$.  Assume, for the sake of contradiction, that
there would exist an io-disjoint expression $\beta$ as in the
theorem, but with $O(\beta)=O(\alpha)=\emptyset$.  Since
$I(\beta)$ must equal $I(\alpha)=\{x\}$, the only variable
occurring in $\beta$ is $x$.  In particular, $\beta$ can only
mention $R$ in atomic subexpressions of the form $R(x;x)$, which
is not io-disjoint.  We are forced to conclude that $\beta$
cannot mention $R$ at all.  Such an expression, however, can
never be a correct rewriting of $\alpha$.  Indeed, let $D$ be an
instance for which $\sem\alpha$ is
nonempty.  Hence $\sem\beta$ is nonempty as well.
Now let $D'$ be the instance with $D'(S)=D(S)$ but
$D'(R)=\emptyset$.  Then $\semm\alpha {D'}$ becomes clearly
empty, but $\semm\beta{D'}=\sem\beta$ remains nonempty
since $\beta$ does not mention $R$.

\subsubsection{Variable renaming}

In the proof of our theorem we need a rigorous way of
renaming variables in FLIF expressions.  The following lemma
allows us to do this.  It confirms that expressions behave under
variable renamings as expected.  The proof by structural
induction is straightforward.

As to the notation used in the lemma, recall that $\VarUniv$ is defined 
to be the universe of variables.  For a permutation $\theta$ of $\VarUniv$, and
an expression $\alpha$, we use $\theta(\alpha)$ for the
expression obtained from $\alpha$ by replacing every occurrence
of any variable $x$ by $\theta(x)$.

\begin{lem}[Renaming Lemma]\label{thm:subs}
Let $\alpha$ be an FLIF expression and let $\theta$ be a
permutation of $\VarUniv$. Then for every instance $D$, we have
\[
(\nu_1, \nu_2) \in \sem{\alpha} \iff
(\nu_1 \fcomp \theta, \nu_2 \fcomp \theta) 
\in \sem{\theta(\alpha)}.
\]
\end{lem}
\subsubsection{Rewriting procedure} \label{method}

In order to be able to give a constructive proof of
Theorem~\ref{main} by
structural induction, a stronger induction hypothesis is needed.
Specifically,
to avoid clashes, we introduce a set $W$ of forbidden variables.
So we will actually prove the following statement:

\begin{lem} \label{lem}
Let $\alpha$ be an FLIF expression, let $W$ be a set of
variables, and let $\rho$ be a bijection
from $O(\alpha)$ to a set of variables disjoint from $\var(\alpha)$.  There exists an $\flio$ expression $\beta$
such that
\begin{enumerate}
\item
$I(\beta) = I(\alpha)$;
\item
$O(\beta) \supseteq \rho(O(\alpha))$ and $O(\beta) -
\rho(O(\alpha))$ is disjoint from $W$;
\item
for every instance $D$ and
every valuation $\nu_1$, we have
\[
\{\nu_2|_{O(\alpha)} \mid (\nu_1,\nu_2) \in \sem \alpha\} =
\{\nu_2 \circ \rho \mid (\nu_1,\nu_2) \in \sem \beta\}.
\]
\end{enumerate}
\end{lem}

We proceed to formally describe an inductive rewriting procedure to produce
$\beta$ from $\alpha$ as prescribed by the above lemma.
The procedure formalizes and generalizes the situations
encountered in the examples discussed in the previous section.
The correctness of the method is proven in Section~\ref{secproof}.

\paragraph*{Terminology}
A bijection from a set of variables $X$ to another set of
variables is henceforth called a \emph{renaming of $X$}.

\paragraph*{Relation atom} If $\alpha$ is of the form $R(\bar
x;\bar y)$, then $\beta$ equals $R(\bar x;\rho(\bar y))$.

\paragraph*{Variable assignment} If $\alpha$ is of the form
$x := t$, then $\beta$ equals $\rho(x) := t$.

\paragraph*{Equality test} If $\alpha$ is an equality test,
we can take $\beta$ equal to $\alpha$.

\paragraph*{Nullary expressions} An expression $\alpha$ is
called \emph{nullary} if it contains no variables, i.e., $\var(\alpha)$ is empty.  Trivially, for nullary $\alpha$, the desired
$\beta$ can be taken to be $\alpha$ itself.  We will consider
this to be an extra base case for the induction.

\intersection{
\paragraph*{Intersection}
If $\alpha$ is of the form $\alpha_1 \cap \alpha_2$
then $\beta$ equals $ \beta_1 \comp \gamma_1 \comp \beta_2 \comp \gamma_2
\comp \eta, $ where the constituent expressions are defined as
follows.
\begin{itemize}
\item
Let $W_1 = W \cup \var(\alpha)$ and let $\rho_1$ be a renaming of
$\Out(\alpha_1)$ that is an extension of $\rho$ such that the
image of $\rho_1-\rho$ is disjoint from $W_1$.  By induction,
there exists an io-disjoint rewriting of $\alpha_1$
for $W_1$ and $\rho_1$; this yields $\beta_1$.
\item
Let $W_2 = W_1 \cup \Out(\beta_1)$ and let $\rho_2$ be a bijection
from $O(\alpha_2)$ to a set of variables that is disjoint from
$W_2$.  By induction, there exists an io-disjoint rewriting of $\alpha_2$
for $W_2$ and $\rho_2$; this yields $\beta_2$.
\item $\gamma_1$ is the composition of all $(\rho_1(y) = y)$ for
$y \in \Out(\alpha_1) - \Out(\alpha_2)$. If
$O(\alpha_1)-O(\alpha_2)$ is empty, $\gamma_1$ can be dropped
from the expression; this qualification applies to
similar situations below.
\item
$\gamma_2$ is defined symmetrically.
\item $\eta$ is the composition of all $(\rho_1(y) = \rho_2(y))$
for $y \in O(\alpha)$.
\end{itemize}
}

\paragraph*{Composition} If $\alpha$ is of the form $\alpha_1
\comp \alpha_2$ then $\beta$ equals $\beta_1 \comp
\theta(\beta_2)$, where the constituents are defined as follows.
\begin{itemize}
\item
Let $W_2 = W \cup \var(\alpha) \cup \rho(O(\alpha_1))$, and let
$\rho_2$ be the restriction of $\rho$ to $O(\alpha_2)$.  By
induction, there exists an io-disjoint rewriting of $\alpha_2$
for $W_2$ and $\rho_2$; this yields $\beta_2$.
\item
Let $W_1 = W \cup \var(\alpha)$, and let $\rho_1$
be a renaming of $O(\alpha_1)$ such that
\begin{itemize}
\item
on $O(\alpha_1) \cap O(\alpha_2) \cap I(\alpha_2)$, the image of
$\rho_1$ is disjoint from $\var(\alpha) \cup O(\beta_2)$ as well as
from the image of $\rho$;
\item
elsewhere, $\rho_1$ agrees with $\rho$.
\end{itemize}
By induction, there exists an io-disjoint renaming of $\alpha_1$
for $W_1$ and $\rho_1$; this yields $\beta_1$.
\item
$\theta$ is the permutation of $\VarUniv$ defined as follows.  For
every $y \in I(\alpha_2) \cap O(\alpha_1)$, we have \[
\theta(y)=\rho_1(y) \text{ and } \theta(\theta(y)) = y. \]
Elsewhere, $\theta$ is the identity.
\end{itemize}

\paragraph*{Union} If $\alpha$ is of the form $\alpha_1 \cup
\alpha_2$ then $\beta$ equals $ (\beta_1 \comp \gamma_1 \comp
\eta_1) \cup (\beta_2 \comp \gamma_2 \comp \eta_2) $ where the
constituent expressions are defined as follows.
\begin{itemize}
\item
Let $W_1 = W \cup \var(\alpha) \cup \rho(O(\alpha_2))$ and let
$\rho_1$ be the restriction of $\rho$ on $O(\alpha_1)$.  By
induction, there exists an io-disjoint rewriting of $\alpha_1$
for $W_1$ and $\rho_1$; this yields $\beta_1$.
\item
Let $W_2 = W \cup \var(\alpha) \cup O(\beta_1)$ and let $\rho_2$
be the restriction of $\rho$ on $O(\alpha_2)$.  By
induction, there exists an io-disjoint rewriting of $\alpha_2$
for $W_2$ and $\rho_2$; this yields $\beta_2$.
\item
$\gamma_1$ is the composition of all $(\rho(y):=y)$ for $y\in
O(\alpha_2)-O(\alpha_1)$, and $\gamma_2$ is defined
symmetrically.
\item
If $O(\beta_2)-\rho_2(O(\alpha_2))$ (the set of ``intermediate''
variables in $\beta_2$) is empty, $\eta_1$ can be dropped from the
expression.  Otherwise,
$\eta_1$ is the composition of all $(y := z)$ for $y \in
O(\beta_2) - \rho(O(\alpha_2))$, with $z$ a fixed variable
chosen as follows.
\begin{enumerate}
\item[(a)]
If $O(\beta_1)$ is nonempty, take $z$ arbitrarily from there.
\item[(b)]
Otherwise, take $z$ arbitrarily from $\var(\alpha_2)$.
We know $\var(\alpha_2)$ is nonempty, since
otherwise $\alpha_2$ would be nullary, so $\beta_2$ would equal
$\alpha_2$, and then $O(\beta_2)$ would be empty as well (extra
base case), which is not the case.
\end{enumerate}
\item
$\eta_2$ is defined symmetrically.
\end{itemize}

\paragraph*{Difference}
If $\alpha$ is of the form $\alpha_1 - \alpha_2$ then $\beta$
equals $ (\beta_1 \comp \gamma_1 \comp \eta_1 \comp \eta_2) -
(\beta_2 \comp \gamma_2 \comp \eta_1 \comp \eta_2) $ where the
constituent expressions are defined as follows.
\begin{itemize}
\item
Let $W_1 = W \cup \var(\alpha)$ and let $\rho_1=\rho$.
By induction, there exists an
io-disjoint rewriting of $\alpha_1$ for
$W_1$ and $\rho$; this yields $\beta_1$.
\item
Let $W_2 = W_1 \cup O(\beta_1)$ and
let $\rho_2$ be a renaming of $O(\alpha_2)$ that agrees with
$\rho$ on $O(\alpha_1) \cap O(\alpha_2)$, such that
the image of $\rho_2-\rho_1$ is disjoint from
$W_2$.  By induction, there exists an io-disjoint
rewriting of $\alpha_2$ for $W_2$ and $\rho_2$; this yields
$\beta_2$.
\item
$\gamma_1$ is the composition of all $(\rho_2(y):=y)$ for $y \in
O(\alpha_2) - O(\alpha_1)$, and $\gamma_2$ is defined
symmetrically.
\item
If $O(\beta_2) - \rho_2(O(\alpha_2))$ is empty,
$\eta_1$ can be dropped from the
expression. Otherwise, $\eta_1$ is the composition of all $(y:=z)$
for $y \in O(\beta_2) - \rho_2(O(\alpha_2))$, with $z$ a fixed
variable chosen as follows.
\begin{enumerate}
\item[(a)]
If $O(\alpha_1) \cap O(\alpha_2)$ is nonempty, take $z$
arbitrarily from $\rho(O(\alpha_1) \cap O(\alpha_2))$.
\item[(b)]
Otherwise, take $z$ arbitrarily from $\var(\alpha_2)$ (which is
nonempty by the same reasoning as given for the union case).
\end{enumerate}
\item
$\eta_2$ is defined symmetrically.
\end{itemize}

\subsubsection{Necessity of variable assignment}

Our rewriting procedure intensively uses variable assignment.  Is
this really necessary?  More precisely, suppose $\alpha$ itself
does not use variable assignment.  Can we still always find an
io-disjoint rewriting $\beta$ such that $\beta$ does not use
variable assignment either?  Below, we answer this question
negatively; in other words, the ability to do variable assignment
is crucial for io-disjoint rewriting.

For our counterexample we
work over the schema consisting of a nullary relation name
$S$ and a binary relation name $T$ of input arity one.
Let $\alpha$ be the expression $S(;) \cup T(x;x)$ and let $\rho$
rename $x$ to $x_1$.  Note that our rewriting procedure would
produce the rewriting \[S(;) \comp (x_1:=x) \: \cup \: T(x;x_1),
\]indeed using a variable assignment ($x_1:=x$) to ensure an io-disjoint
expression.

For the sake of contradiction, assume there
exists an expression $\beta$ according to Theorem~\ref{main} that
does not use variable assignment.  Fix $D$ to the instance where
$S$ is nonempty but $T$ is empty.  Then $\sem\alpha$ consists of
all identical pairs of valuations.  Take any valuation $\nu$ with
$\nu(x)\neq\nu(x_1)$.  Since $(\nu,\nu) \in \sem\alpha$,
there should exist a valuation $\nu'$ with $\nu'(x_1)=\nu(x)$
such that $(\nu,\nu') \in \sem\beta$.  Note that $\nu'\neq\nu$,
since $\nu(x_1)\neq\nu(x)$.  However, this contradicts the
following two observations.  Both
observations are readily verified by induction.
(Recall that $D$ is fixed as defined above.)
\begin{enumerate}
\item
For every expression $\beta$ without variable assignments,
either $\sem \beta$ is empty, or $\sem \beta = \sem \gamma$ for
some expression $\gamma$ that does not mention $T$ and that has no
variable assignments.
\item
For every expression $\gamma$ that does not mention $T$ and that
has no variable assignments, and any $(\nu_1,\nu_2) \in
\sem\gamma$, we have $\nu_1=\nu_2$.
\end{enumerate}

\subsection{Improved translation from io-disjoint FLIF to Executable FO} \label{secimproved}

We now turn to the translation from $\flio$ to executable FO.  
Here, a rather straightforward equivalence is possible, since 
executable FO has an explicit quantification operation 
which is lacking in FLIF.
Recall the evaluation problem for executable FO
(Definition~\ref{defevalexfo}, 
and the evaluation problem for $\alpha$ 
(Definition~\ref{defevalalphainput}).

\begin{thm} \label{flio2exfo}
Let $\alpha$ be an $\flio$ expression over a schema $\Sch$.  There
exists an $I(\alpha)$-executable FO formula $\varphi_\alpha$ over
$\Sch$, with $\FV(\varphi_\alpha)=\var(\alpha)$,
such that for every $D$ and $\nuin$, we have
$\EvalL \alpha D \nuin = \EvalV {\varphi_\alpha} {I(\alpha)} D \nuin$.
The length of $\varphi_\alpha$ is linear in the length of
$\alpha$.
\end{thm}
\begin{exa}
To illustrate the proof,
consider the $\flio$ expression $R(x;y,u) \comp S(x;z,u)$.  
Procedurally, to evaluate the first expression, we 
retrieve values for the variables $y$ and $u$ that match the value
given for the variable $x$ in the relation $R$.
We proceed to retrieve a $(z,u)$-binding from $S$ for the given
$x$, effectively overwriting the previous binding for $u$.  Thus,
a correct translation into executable FO is $(\exists u\,
R(x;y,u)) \land S(x;z,u)$.

Interpreting relations as functions, this example can be likened to the following piece of code in Python:
\begin{verbatim}
    y,u = R(x) ; z,u = S(x)
\end{verbatim}
Indeed, formalisms such as FLIF, as well as its mother framework LIF \cite{lif_frocos}, dynamic logic \cite{dynamiclogicbook}, and dynamic predicate logic \cite{dplogic} provide logical foundations for such programming constructs (and even natural language constructs).

For another example, consider the assignment $(x:=y)$.
This translates to $x=y$ considered as a $\{y\}$-executable
formula.  The equality test $(x=y)$ also translates to $x=y$, but
considered as an $\{x,y\}$-executable formula.
\end{exa}
\begin{proof}[Proof Sketch of Theorem~\ref{flio2exfo}]
Table~\ref{tabtrans} shows
the translation, which is almost an isomorphic embedding,
except for the case of composition. The correctness of the
translation for composition again hinges on inertia and input-output determinacy.
The formal correctness proof, including the
verification that $\varphi_\alpha$ is indeed
$I(\alpha)$-executable, is given in Section~\ref{proofflio2exfo}.
\end{proof}

\begin{table}
\caption{Translation showing how $\flio$ embeds in
executable FO\@.  In the table, $\varphi_i$
abbreviates $\varphi_{\alpha_i}$ for $i=1,2$.}
\label{tabtrans}
\centering
\begin{tabular}{c|l}
\toprule
$\alpha$ & $\varphi_\alpha$ \\
\toprule
$R(\bar{x};\bar{y})$ & $R(\bar{x};\bar{y})$ \\
\midrule
$(x = y)$ & $ x = y$ \\
\midrule
$(x := y)$ & $ x = y$ \\
\midrule
$x = c$ & $x = c$ \\
\midrule
$x := c$ & $x = c$ \\
\midrule
$\alpha_1 ; \alpha_2$ & $(\exists x_1 \dots \exists x_k \, \varphi_1)
\land \varphi_2$ where $\{x_1,\dots,x_k\}=O(\alpha_1)\cap
O(\alpha_2)$\\
\midrule
$\alpha_1 \cup \alpha_2$ & $\varphi_1 \lor \varphi_2$ \\
\midrule
\intersection{$\alpha_1 \cap \alpha_2$ & $\varphi_1 \land \varphi_2$ \\
\midrule}
$\alpha_1 \setminus \alpha_2$ & $\varphi_1 \land \lnot \varphi_2$ \\
\bottomrule
\end{tabular}
\end{table}

\intersection{
Notably, in the proof of Theorem~\ref{exfo2flio},
we do not need
the intersection operation.  Hence, by translating $\flio$ to
executable FO and then back to $\flio$, we obtain that
intersection is redundant in $\flio$, in the following sense:

\begin{cor}\label{cor:noInt}
For every $\flio$ expression $\alpha$ there exists an $\flio$
expression $\alpha'$ with the following properties:
\begin{enumerate}
\item
$\alpha'$ does not use the intersection operation.
\item
$I(\alpha')=I(\alpha)$.
\item
$O(\alpha')\supseteq O(\alpha)$.
\item
For every $D$ and $\nuin$, we have
$ \EvalL \alpha D\nuin = \pi_{\var(\alpha)}(
\EvalL {\alpha'} D\nuin)$.
\end{enumerate}
\end{cor}

\begin{rem}
One may wonder whether
the above corollary directly follows from the equivalence
between
$\alpha_1 \cap \alpha_2$ and $\alpha_1 - (\alpha_1 -
\alpha_2)$.  While these two expressions are semantically
equivalent and have the same input variables, they do not have
the same output variables, so a simple inductive proof
eliminating intersection while preserving the guarantees of the
above corollary does not work.  Moreover, the
corollary continues to hold for
the positive fragment of $\flio$ (without the difference
operation).  Indeed, positive $\flio$ can be translated into
executable FO without negation, which can then be
translated into positive $\flio$ without intersection.
\end{rem}
From the translations between $\flio$ and executable FO, we can 
see that intersection can be expressed in terms of composition
and equality tests.  As an example, consider the $\flio$ 
expression $\alpha$ which is defined to be 
$R_1(x; y, z) \cap R_2(u; v, y)$.  We can see that 
$I(\alpha) = \{x,u,z,v\}$ and $O(\alpha)=\{y\}$.
Using our translation for Theorem~\ref{flio2exfo}, we see 
that $\varphi_\alpha$ is a $\{x,u,z,v\}$-executable formula
and can simply be taken to be $R_1(x; y,z) \land R_2(u; v, y)$.
Now, from our translation for Theorem~\ref{exfo2flio}, we can 
take $\alpha'$ to be 
\[
R_1(x; y, z_1)\comp (z_1 = z)\comp (u = u)\comp (v = v)\comp
R_2(u; v_2, y_2)\comp (v_2 = v)\comp (y_2 = y)\comp (x = x)\comp
(z = z).
\]
We can verify that $I(\alpha') = \{x, z, u, v\} = I(\alpha)$ and
$O(\alpha') = \{y, z_1, v_2, y_2\} \supseteq O(\alpha)$.  Formally 
verifying the last condition of Corollary~\ref{cor:noInt}
is tedious, nonetheless, we can at least intuitively see that it
holds for $\alpha'$.
}

\section{Correctness Proofs of Translation Theorems}\label{sec:proofs2}

\subsection{From Executable FO to FLIF}
\label{proofexfo2flif}

In this section we prove Theorem~\ref{exfo2flif}, which is
reformulated below for convenience.

\newtheorem*{T1}{Theorem~\ref{exfo2flif}}
\begin{T1} \label{exfo2flifHer}
Let $\varphi$ be a $\VarSet$-executable formula over a schema $\Sch$.
There exists an FLIF expression $\alpha$ over $\Sch$ 
{and a set of variables $\SVars \supseteq \FV(\varphi)\cup \VarSet$} 
such that for every $D$, valuation $\nuin$ on $\VarSet$, 
and valuation $\nuin'$ {on $\SVars$}
with $\nuin' \supseteq \nuin$, we have
\[
\{ \nu|_{\FV(\varphi) \cup \VarSet} \mid \nuin \subseteq \nu \text{ and } 
\satD{\varphi}{\nu}\} 
= 
\{ \nuout|_{\FV(\varphi) \cup \VarSet} \mid (\nuin', \nuout) \in {\semV{\alpha}}\}.
\]
\end{T1}
\begin{proof}
By structural induction.
The containment from left to right is referred to as 
\emph{completeness}, and the containment from right to left
as \emph{soundness}. 

{In the proof, we will omit the explicit definition of the set
$\SVars$ and we take it to be the set of all variables mentioned in the
constructed expression $\alpha$.} 
It is {also} worth noting that it follows from the statement of the theorem
that $\alpha$ cannot change the values of the variables in $\VarSet$.
Precisely, for every $\nu_1$, $\nu_2$ such that 
$(\nu_1, \nu_2) \in \sem{\alpha}$, it must be the case that 
$\nu_1$ agrees with $\nu_2$ on $\VarSet$.

\paragraph{Atoms}
If $\varphi$ is a relation atom $R(\bar x;\bar y)$, then $\alpha$
is $R(\bar x;\bar z) \comp \xi$, where $\bar z$ is obtained from 
$\bar y$ by replacing each variable from $\VarSet$ by a
fresh variable.  The expression $\xi$ consists of the composition
of all equalities $(y_i = z_i)$ where $y_i$ is a variable from
$\bar y$ that is in $\VarSet$ and $z_i$ is the corresponding fresh
variable.  In what follows, let $X$, $Y$, and $Z$ be the
variables in $\bar{x}$, $\bar{y}$, and $\bar{z}$; respectively. 
Moreover, take $\nuin$ to be an arbitrary valuation on $\VarSet$, and
$\nuin'$ to be any valuation such that $\nuin' \supseteq \nuin$. 

We first prove completeness.  Let $\nu$ to be a valuation on 
$Y \cup \VarSet$ such that $\nuin \subseteq \nu$.  Now, suppose that 
$\satD{\varphi}{\nu}$.  We want to verify that there exists a valuation 
$\nuout$ such that $(\nuin', \nuout) \in \sem{\alpha}$ and 
$\nuout|_{\FV(\varphi) \cup \VarSet} = \nu$, which is clear when
taking $\nuout$ to be the valuation that agrees with $\nuin$ on $\VarSet$,
agrees with $\nu$ on $\FV(\varphi) - \VarSet$, agrees with $\nuin'$
outside $\var(\alpha)$, and satisfies $\nuout(z) = \nu(y)$ for every 
$y \in (Y \cap \VarSet)$ and its corresponding $z \in Z$. 

To show soundness, suppose that there exists a valuation $\nuout$ 
such that $(\nuin', \nuout) \in \sem{\alpha}$.  We want to verify that 
$\satD{\varphi}{\nuout|_{\FV(\varphi) \cup \VarSet}}$, which is clear
given the semantics of $\alpha$.

The cases where $\varphi$ is of the form $x=y$ or $(x=c)$ are
handled as already shown in the previous section; correctness is 
clear.

\paragraph{Conjunction}
If $\varphi$ is $\varphi_1 \land \varphi_2$, then by induction we
have an expression $\alpha_1$ for 
$\varphi_1$ and $\VarSet$, and an expression
$\alpha_2$ for $\varphi_2$ and $\VarSet \cup
\FV(\varphi_1)$ (since $\varphi_2$ is $\VarSet \cup
\FV(\varphi_1)$-executable).
We show that $\alpha$ can be taken to be $\alpha_1 \comp \alpha_2$.
Take $\nuin$ to be an arbitrary valuation on $\VarSet$, and
$\nuin'$ to be any valuation such that $\nuin' \supseteq \nuin$. 

We first prove completeness.  Let $\nu$ be a valuation on 
$\FV(\varphi) \cup \VarSet$ such that $\nuin \subseteq \nu$. Now, 
suppose that $\satD{\varphi}{\nu}$.  We want to verify that there 
exists a valuation $\nuout$ such that $(\nuin', \nuout) \in \sem{\alpha}$ 
and $\nuout|_{\FV(\varphi) \cup \VarSet} = \nu$. 
Clearly, $\satD {\varphi_1} \nu$ and $\satD {\varphi_2} \nu$.  
By induction, there exists $\nu_1$ such that 
$(\nuin',\nu_1) \in \sem{\alpha_1}$ and $\nu_1=\nu$ on
$\FV(\varphi_1)\cup\VarSet$.  From the last equality and 
also from induction, there exists 
$\nuout$ such that $(\nu_1,\nuout)\in\sem{\alpha_2}$ and 
$\nuout=\nu_1=\nu$ on $\FV(\varphi_2) \cup \VarSet \cup
\FV(\varphi_1)=\FV(\varphi)\cup\VarSet$. 

We next show soundness.  Suppose that there exists a valuation 
$\nuout$ such that $(\nuin', \nuout) \in \sem{\alpha}$.  
We want to verify that 
$\satD{\varphi}{\nuout|_{\FV(\varphi) \cup \VarSet}}$, and
that $\nuout \supseteq \nuin$. 
Clearly, there exists a valuation $\nu$ such that
$(\nuin',\nu)\in\sem{\alpha_1}$ and $(\nu,\nuout)\in\sem{\alpha_2}$.
By induction, $\satD {\varphi_1} {\nu|_{\FV(\varphi_1) \cup \VarSet}}$
and $\nu \supseteq \nuin$.  Also by induction, 
$\satD {\varphi_2} {\nuout|_{\FV(\varphi_2) \cup \FV(\varphi_1) \cup \VarSet}}$
and $\nuout \supseteq \nu|_{\FV(\varphi_1) \cup \VarSet}$.
From the latter, we obtain $\satD{\varphi}{\nuout|_{\FV(\varphi) \cup \VarSet}}$.
Showing that $\nuout \supseteq \nuin$ is clear.

\paragraph{Disjunction}
If $\varphi$ is $\varphi_1 \lor \varphi_2$, then by induction we
have an expression $\alpha_i$ for $\varphi_i$ and 
$(\FV(\varphi_1) \symdif \FV(\varphi_2)) \cup \VarSet$ for $i = 1,2$
(since $\FV(\varphi_1) \symdif \FV(\varphi_2) \subseteq \VarSet$).
We show that $\alpha$ can be taken to be $\alpha_1 \cup \alpha_2$.
Take $\nuin$ to be an arbitrary valuation on $\VarSet$, and
$\nuin'$ to be any valuation such that $\nuin' \supseteq \nuin$. 

We first prove completeness.  Let $\nu$ be a valuation on 
$\FV(\varphi) \cup \VarSet$ such that $\nuin \subseteq \nu$. Now, 
suppose that $\satD{\varphi}{\nu}$.  We want to verify that there 
exists a valuation $\nuout$ such that $(\nuin', \nuout) \in \sem{\alpha}$ 
and $\nuout|_{\FV(\varphi) \cup \VarSet} = \nu$.  We only consider the case 
when $\satD{\varphi_1}{\nu}$; the other is symmetric.  By induction, there exists 
a valuation $\nuout$ such that $(\nuin', \nuout) \in \sem{\alpha_1}$ and 
$\nuout|_{\FV(\varphi_1) \cup \VarSet} = \nu$.  
Clearly, $\FV(\varphi_1) \cup \VarSet = \FV(\varphi) \cup \VarSet$ 
from the conditions on $\VarSet$, and we are done.

To show soundness, suppose that there exists a valuation 
$\nuout$ such that $(\nuin', \nuout) \in \sem{\alpha}$.  
We want to verify that 
$\satD{\varphi}{\nuout|_{\FV(\varphi) \cup \VarSet}}$, and
that $\nuout \supseteq \nuin$.  Again, we only consider the case 
when $(\nuin', \nuout) \in \sem{\alpha_1}$; the other is symmetric.
By induction, $\satD {\varphi_1} {\nuout|_{\FV(\varphi_1) \cup \VarSet}}$
and $\nuout \supseteq \nuin$.  Again, 
$\FV(\varphi_1) \cup \VarSet = \FV(\varphi) \cup \VarSet$ 
from the conditions on $\VarSet$, and we are done.

\paragraph{Existential Quantification}
If $\varphi$ is $\exists x \, \varphi_1$,
then without loss of generality we may assume that $x \notin V$.
By induction, we have an expression $\alpha_1$ for $\varphi_1$ and 
$\VarSet$.  We show that this expression also works for $\varphi$.
Take $\nuin$ to be an arbitrary valuation on $\VarSet$, and
$\nuin'$ to be any valuation such that $\nuin' \supseteq \nuin$. 

We first prove completeness.  Let $\nu$ be a valuation on 
$\FV(\varphi) \cup \VarSet$ such that $\nuin \subseteq \nu$. Now, 
suppose that $\satD{\varphi}{\nu}$, and hence, 
$\satD{\varphi_1}{\nu \cup \nu_x}$ where $\nu_x$ is a valuation on $\{x\}$.  
We want to verify that there exists a valuation $\nuout$ such that 
$(\nuin', \nuout) \in \sem{\alpha_1}$ and 
$\nuout|_{\FV(\varphi) \cup \VarSet} = \nu$.
By induction, we know that such $\nuout$ exists but with 
$\nuout|_{\FV(\varphi_1) \cup \VarSet} = \nu \cup \nu_x$.
Since $x$ belongs neither to $\FV(\varphi)$ nor to $\VarSet$,
we easily obtain $\nuout|_{\FV(\varphi) \cup \VarSet} = \nu$, 
and we are done.

To show soundness, let $(\nuin',\nuout)\in\sem{\alpha_1}$.  By
induction, $\satD {\varphi_1} {\nuout|_{\FV(\varphi_1) \cup \VarSet}}$, 
so certainly $\satD \varphi {\nuout|_{\FV(\varphi) \cup \VarSet}}$. 
What remains to show is that $\nuout \supseteq \nuin$ which is clear
from the induction step.

\paragraph{Negation}
Finally, if $\varphi$ is $\neg \varphi_1$, then by induction we
have an expression $\alpha_1$ for $\varphi_1$ and $\VarSet$.  Fix an
arbitrary constant $c$, and a fresh variable $u$.  Let $Z$ denote 
$(\var(\alpha_1) \setminus \VarSet)\cup\{u\}$, and let $\xi$ be
the composition of all expressions $(z:=c)$ for $z\in Z$.  
We show that $\alpha$ can be taken to be $\xi \; - \; \alpha_1\comp\xi$.  
Note that $\FV(\varphi)=\FV(\varphi_1)\subseteq\VarSet$ 
(by the $\VarSet$-executability of $\varphi$).

We first prove completeness.  Suppose that $\satD \varphi \nuin$.
We want to verify that there 
exists a valuation $\nuout$ such that $(\nuin', \nuout) \in \sem{\alpha}$ 
and $\nuout|_{\VarSet} = \nuin$.
Take $\nuout'$ to be the valuation that agrees with $\nuin'$ outside $Z$
(and hence, $\nuout'|_{\VarSet} = \nuin$), 
and moreover, it assigns the value $c$ for every $z\in Z$.  It is clear
that $(\nuin', \nuout') \in \sem{\xi}$.  For the sake of contradiction,
suppose $(\nuin',\nuout')\in\sem{\alpha_1\comp\xi}$.  Then, by definition,
there exists $\nu$ such that $(\nuin', \nu) \in \sem{\alpha_1}$.
By induction, we know that $\satD {\varphi_1} {\nu|_{\VarSet}}$ 
and $\nu|_{\VarSet} = \nuin$.  It follows that $\nsatD \varphi \nuin$,
which is a contradiction.  Thus, $(\nuin',\nuout')\not\in\sem{\alpha_1\comp\xi}$, 
whence, $(\nuin',\nuout')\in\sem{\alpha}$, as desired. 

To show soundness, suppose that there exists a valuation 
$\nuout$ such that $(\nuin', \nuout) \in \sem{\alpha}$.  
We want to verify that $\satD{\varphi}{\nuin}$. 
By the semantics of $\alpha$, we obtain that $(\nuin', \nuout) \in \sem{\xi}$, 
and $(\nuin', \nuout) \not \in \sem{\alpha_1\comp\xi}$.  From the former,
we obtain that $\nuout =  \nuin'$ outside $Z$ which is disjoint from 
$\VarSet$.  For the sake of contradiction, assume that 
$\satD{\varphi_1}{\nuin}$.  Then, by induction, there is a 
$\nuout'$ such that $(\nuin', \nuout') \in \sem{\alpha_1}$ 
and $\nuout' = \nuin' = \nuout$ on $\VarSet$.  What remains to show is
that $(\nuout', \nuout) \in \sem{\xi}$ yielding the contradiction.
It is not hard to see that $\nuin' = \nuout'$ outside $\var(\alpha_1) - \VarSet$
which contains all the set of variables outside $Z$.
Thus, $\nuout = \nuout'$ outside $Z$, whence, $(\nuout', \nuout) \in \sem{\xi}$
as desired.
\end{proof}

\subsection{From FLIF to io-disjoint FLIF}
\label{secproof}

We prove that $\beta$ constructed by the method described in
Section~\ref{method} satisfies the statement of
Lemma~\ref{lem}.  The base cases are straightforwardly verified.
For every inductive case, we need to verify several things:
\begin{description}
\item[Inputs] $I(\beta)=I(\alpha)$.
\item[io-disjointness] Every subexpression of $\beta$, including
$\beta$ itself, must have disjoint inputs and outputs.
\item[Outputs] $O(\beta) \supseteq \rho(O(\alpha))$.
\item[No clashes] $O(\beta) - \rho(O(\alpha))$ is disjoint from $W$.
\item[Completeness] For any instance $D$ and
$(\nu_1,\nu_2) \in \sem \alpha$, we want to find $\nu$ such that
$(\nu_1,\nu) \in \sem \beta$ and $\nu(\rho(y))=\nu_2(y)$ for $y
\in O(\alpha)$.
\item[Soundness] For any $(\nu_1,\nu_2) \in \sem \beta$, we want
to find $\nu$ such that $(\nu_1,\nu) \in \sem \alpha$ and
$\nu(y)=\nu_2(\rho(y))$ for $y \in O(\alpha)$.
\end{description}

\intersection{
\subsection{Intersection}

\subsubsection*{Inputs} Let $\{i,j\}=\{1,2\}$.  We begin by
analyzing $\gamma_i$.  By definition
(Table~\ref{tabio}) an equality test has no outputs, so a
composition of equality tests has no outputs and has the union of
the inputs as inputs:
\begin{align*}
I(\gamma_i) & = (O(\alpha_i)-O(\alpha_j)) \cup
(\rho_i(O(\alpha_i)) - O(\alpha_j)) \\
O(\gamma_i) & = \emptyset
\end{align*}
Since the part of $I(\gamma_i)$ that is contained in
$\rho_i(O(\alpha_i))$ is a subset of $O(\beta_i)$, and
$I(\beta_i)=I(\alpha_i)$ by induction, we obtain:
\begin{align*}
I(\beta_i \comp \gamma_i) &= I(\alpha_i) \cup
(O(\alpha_i)-O(\alpha_j)) \\
O(\beta_i \comp \gamma_i) &= O(\beta_i)
\end{align*}
Since $O(\beta_1 \comp \gamma_1) = O(\beta_1)$ is disjoint from
$\var(\alpha)$ and thus from $I(\beta_2 \comp \gamma_2)$, the
inputs for the composition of $\beta_1 \comp \gamma_1$ and
$\beta_2 \comp \gamma_2$ is simply the union of the inputs, which
equals $I(\alpha_1) \cup I(\alpha_2) \cup (O(\alpha_1) \symdif
O(\alpha_2))$, which equals $I(\alpha)$:
\begin{align*}
I(\beta_1 \comp \gamma_1 \comp \beta_2 \comp \gamma_2) &=
I(\alpha) \\
O(\beta_1 \comp \gamma_1 \comp \beta_2 \comp \gamma_2) &=
O(\beta_1) \cup O(\beta_2)
\end{align*}
Looking at $\eta$, and noting that $O(\alpha)=O(\alpha_1)\cap
O(\alpha_2)$, we see that $I(\eta) = \rho_1(O(\alpha)) \cup
\rho_2(O(\alpha))$, which is entirely contained in $O(\beta_1)
\cup O(\beta_2)$ by induction.  We conclude that
$I(\beta)=I(\alpha)$ as desired.

\subsubsection*{Outputs} We have $O(\beta)=O(\beta_1) \cup
O(\beta_2)$ which includes $\rho_1(O(\alpha_1)) \cup
\rho_2(O(\alpha_2))$ by induction, which includes
$\rho(O(\alpha))$ as desired.

\subsubsection*{io-disjointness} Let $i=1$ or $2$.  The
subexpressions $\gamma_i$ and $\eta$ are trivially io-disjoint.  The
subexpression $\beta_i \circ \gamma_i$ is io-disjoint since, by
construction, $O(\beta_i)$ is disjoint from $\var(\alpha)$.  The
subexpression $\beta_1 \comp \gamma_1 \comp \beta_2 \comp
\gamma_2$, and the final expression $\beta$,
are io-disjoint for the same reason.  

\subsubsection*{No clashes}
Note that $\rho_1$ agrees with $\rho$ on $O(\alpha)$.
Furthermore, by induction and definition of $\rho_2$, the set
$O(\beta_2)$ is disjoint from $W_2$ and thus from $O(\beta_1)$,
which in turn includes $\rho_1(O(\alpha_1))$ by induction.
Hence:
\begin{align*}
O(\beta)-\rho(O(\alpha)) &= (O(\beta_1) \cup O(\beta_2)) -
\rho_1(O(\alpha_1) \cap O(\alpha_2)) \\
&= (O(\beta_1) - \rho_1(O(\alpha))) \cup O(\beta_2) \\
&\subseteq 
(O(\beta_1) - \rho_1(O(\alpha_1))) \cup
\rho_1(O(\alpha_1)-O(\alpha)) \cup O(\beta_2).
\end{align*}
The three terms of the last union are disjoint from $W$ by
construction and the induction hypothesis.

\subsubsection*{Completeness} Since $(\nu_1,\nu_2) \in
\sem{\alpha_1} \cap \sem{\alpha_2}$, by induction, we have
$\nu_3$ and $\nu_4$ such that 
\begin{itemize}
\item
$(\nu_1,\nu_3) \in \sem{\beta_1}$ and $\nu_3(\rho_1(y))=\nu_2(y)$
for $y \in O(\alpha_1)$;
\item
$(\nu_1,\nu_4) \in \sem{\beta_2}$ and $\nu_4(\rho_2(y))=\nu_2(y)$
for $y \in O(\alpha_2)$.
\end{itemize}

Since, by inertia, $\nu_1$ agrees with $\nu_2$ outside
$\Out(\alpha_1) \cap \Out(\alpha_2)$ and $\nu_1$ agrees with
$\nu_3$ outside $\Out(\beta_1)$ which is disjoint from
$\var(\alpha)$, we have $\nu_3(\rho_1(y)) = \nu_2(y) = \nu_1(y) =
\nu_3(y)$ for every $y \in \Out(\alpha_1) - \Out(\alpha_2)$,
whence, $(\nu_3, \nu_3) \in \sem{\gamma_1}$.  Consequently,
$(\nu_1, \nu_3) \in \sem{\beta_1 \comp \gamma_1}$. 

Reasoning similarly, we also have $(\nu_1, \nu_4) \in
\sem{\beta_2 \comp \gamma_2}$.  As noted in the previous
paragraph, $\nu_1$ and $\nu_3$ agree on $\var(\alpha)$ so in
particular on $I(\beta_2\comp\gamma_2)$. Hence, by input
determinacy, there exists $\nu_5$ such that $(\nu_3,\nu_5) \in
\sem{\beta_2 \comp \gamma_2}$ and $\nu_5$ agrees with $\nu_4$ on
$O(\beta_2\comp\gamma_2)=O(\beta_2)$.

We next check that $(\nu_5,\nu_5) \in \sem\eta$.  By inertia,
$\nu_5$ agrees with $\nu_3$ outside $O(\beta_2)$, which by
construction is disjoint from $O(\beta_1)$, so $\nu_5$ agrees
with $\nu_3$ on $O(\beta_1)$. Let $v \in \Out(\alpha_1) \cap
\Out(\alpha_2)$.  By induction, $\rho_2(v) \in O(\beta_2)$ and
$\rho_1(v) \in O(\beta_1)$.  Hence we have $\nu_5(\rho_2(v)) =
\nu_4(\rho_2(v)) = \nu_2(v) = \nu_3(\rho_1(v)) =
\nu_5(\rho_1(v))$ as desired.
 
All in all, we have $(\nu_1,\nu_5) \in
\sem{\beta_1\comp\gamma_1\comp\beta_2\comp\gamma_2}$, and
$\nu_5(\rho(y))=\nu_5(\rho_1(y))=\nu_2(y)$ for $y \in O(\alpha)$
was already noted in the previous paragraph.  Thus $\nu_5$ yields
the valuation needed to establish completeness.

\subsubsection*{Soundness}
Since $(\nu_1,\nu_2) \in \sem \beta$, we know that $\nu_2 \models
\eta$.\footnote{Here and below, for a composition $\eta$ of equality
tests, and a valuation $\nu$, we use $\nu \models \eta$ to
indicate that $\nu$ satisfies every equality from $\eta$.}  We
also know there exists $\nu_3$ such that $(\nu_1,\nu_3) \in
\sem{\beta_1\comp\gamma_1}$ and $(\nu_3,\nu_2) \in
\sem{\beta_2\comp \gamma_2}$.  Note that $(\nu_1,\nu_3) \in
\sem{\beta_1}$ and $\nu_3 \models \gamma_1$.

Since, by inertia, $\nu_3$ agrees with $\nu_1$ outside
$O(\beta_1\comp \gamma_1)=O(\beta_1)$ which is disjoint from
$\var(\alpha)$, we know that $\nu_1$ agrees with $\nu_3$ on
$\var(\alpha)$ so in particular on $I(\beta_2\comp\gamma_2)$.  
Hence, by input-output determinacy, there exists $\nu_4$ such that
$(\nu_1,\nu_4) \in \sem{\beta_2\comp \gamma_2}$ and $\nu_4$
agrees with $\nu_2$ on $O(\beta_2\comp \gamma_2)=O(\beta_2)$.
Note that $(\nu_1,\nu_4) \in \sem{\beta_2}$ and $\nu_4 \models \gamma_2$.

By induction, there exists $\nu$ such
that $(\nu_1,\nu) \in \sem{\alpha_1}$ and
$\nu(y)=\nu_3(\rho_1(y))$ for $y \in O(\alpha_1)$.  Similarly,
there exists $\nu'$ such that $(\nu_1,\nu') \in \sem{\alpha_2}$
and $\nu'(y)=\nu_4(\rho_2(y))$ for $y \in O(\alpha_2)$.  We are
going to show that $\nu=\nu'$.  By inertia, outside
$O(\alpha_1)\cup O(\alpha_2)$, both $\nu$ and $\nu'$ agree with
$\nu_1$, so we can focus on $O(\alpha_1) \cup O(\alpha_2)$.  We
consider three cases: $O(\alpha_1) \cap O(\alpha_2)$;
$O(\alpha_1)-O(\alpha_2)$; and $O(\alpha_2)-O(\alpha_1)$.

First, consider $y \in O(\alpha_1) \cap O(\alpha_2)$.  Then
\[
\nu(y)=\nu_3(\rho_1(y))=\nu_2(\rho_1(y)) 
{} =
\nu_2(\rho_2(y))=\nu_4(\rho_2(y))=\nu'(y).\]
The first and last equalities are by definition of $\nu$ and
$\nu'$.  The second equality holds because $(\nu_3,\nu_2) \in
\sem{\beta_2}$, so by inertia $\nu_2$ and $\nu_3$ agree outside
$O(\beta_2)$ which is disjoint from $O(\beta_1)$, which contains
$\rho_1(y)$. The third equality holds because $\nu_2 \models
\eta$. The fourth equality holds because $\nu_4$ agrees with
$\nu_2$ on $O(\beta_2)$ which contains $\rho_2(y)$.

Second, consider $y \in O(\alpha_1) - O(\alpha_2)$. Then
\[\nu(y) = \nu_3(\rho_1(y)) = \nu_3(y) = \nu_1(y) = \nu'(y). \]
The second equality holds because $\nu_3 \models \gamma_1$.  The
third equality holds because $(\nu_1,\nu_3) \in \sem{\beta_1}$,
so by inertia $\nu_1$ and $\nu_3$ agree outside $O(\beta_1)$
which is disjoint from $\var(\alpha)$, which contains $y$. The final
equality holds because $(\nu_1,\nu') \in \sem{\alpha_2}$, so by
inertia $\nu_1$ and $\nu'$ agree outside $O(\alpha_2)$ and $y
\notin O(\alpha_2)$.

The third case is similar to the second, so $\nu = \nu'$.
We conclude that $(\nu_1,\nu) \in \sem{\alpha_1} \cap
\sem{\alpha_2}$.  Moreover, in the discussion of the first case
above we already noted
that $\nu(y)=\nu_2(\rho_1(y))$ for $y
\in O(\alpha)=O(\alpha_1)\cap O(\alpha_2)$.  Since
$\rho_1(y)=\rho(y)$ for such $y$, we are done.
}

\subsubsection{Composition}

Henceforth, for any expression $\delta$, we will use the notation
\[\nu_1 \xrightarrow \delta \nu_2 \]to indicate that
$(\nu_1,\nu_2) \in \sem\delta$.

\paragraph{Inputs} We first analyze inputs and outputs for
$\theta(\beta_2)$.  Inputs pose no difficulty (note that
$I(\beta_2)=I(\alpha_2)$).  As to outputs,
$\theta$ only changes variables in $I(\alpha_2)$
and $\beta_2$ is io-disjoint by induction, so $\theta$
has no effect on $O(\beta_2)$.  Hence:
\begin{align*}
I(\theta(\beta_2)) &= (I(\alpha_2) - O(\alpha_1)) \cup
\rho_1(I(\alpha_2) \cap O(\alpha_1)) \\
O(\theta(\beta_2)) &= O(\beta_2)
\end{align*}
Calculating $I(\beta)$, the part of $I(\theta(\beta_2))$ that
is contained in $\rho_1(O(\alpha_1))$ disappears, because
$\rho_1(O(\alpha_1))$ is contained in $O(\beta_1)$.
Also, $I(\beta_1)=I(\alpha_1)$ by induction.
Thus $I(\beta) = I(\alpha_1) \cup (I(\alpha_2)-O(\alpha_1)) =
I(\alpha)$ as desired.

\paragraph{Outputs} We verify:
\begin{align*}
\rho(O(\alpha)) &= \rho(O(\alpha_1)) \cup \rho(O(\alpha_2)) \\
&= \rho(O(\alpha_1)-O(\alpha_2)) \cup \rho(O(\alpha_2)) \\
&= \rho_1(O(\alpha_1)-O(\alpha_2)) \cup \rho_2(O(\alpha_2)) \\
&\subseteq O(\beta_1) \cup O(\beta_2) \\
&= O(\beta).
\end{align*}

\paragraph{io-disjointness} Expression $\beta$ is
io-disjoint since $O(\beta_1)$ and $O(\beta_2)$ are disjoint from
$\var(\alpha)$ by construction.  For subexpression
$\theta(\beta_2)$, recall $I(\theta(\beta_2))$ and
$O(\theta(\beta_2))$ as calculated above.  The part contained in
$I(\alpha_2)$ is disjoint from $O(\beta_2)$ since
$I(\alpha_2)=I(\beta_2)$ and $\beta_2$ is io-disjoint by
induction.  We write the other part as $\rho_1(I(\alpha_2) \cap
O(\alpha_1) \cap O(\alpha_2)) \cup \rho((I(\alpha_2) \cap
O(\alpha_1))-O(\alpha_2))$.  The first term is disjoint from
$O(\beta_2)$ by definition of $\rho_1$.

The second term is
dealt with by the more general claim that
\emph{$\rho(O(\alpha_1)-O(\alpha_2))$ is disjoint from
$O(\beta_2)$.} Towards a proof, let $y \in O(\alpha_1)-O(\alpha_2)$ and
assume for the sake of contradiction that $\rho(y) \in
O(\beta_2)$.  Then $\rho(y) \in O(\beta_2) - \rho(O(\alpha_2))$,
which by induction is disjoint from $W_2$, which includes
$\rho(O(\alpha_1))$.  However, since $y \in O(\alpha_1)$, this is
a contradiction.

\paragraph{No clashes}  We have
\begin{align*}
O(\beta) - \rho(O(\alpha)) &= (O(\beta_1) \cup O(\beta_2)) -
\rho(O(\alpha_1) \cup O(\alpha_2)) \\
&\subseteq (O(\beta_1) - \rho(O(\alpha_1)) \cup (O(\beta_2) -
\rho(O(\alpha_2)).
\end{align*}
By induction,
the latter two terms are disjoint from $W_1 \supseteq W$ and $W_2
\supseteq W$, respectively.

\paragraph{Completeness}  Since $(\nu_1,\nu_2) \in
\sem\alpha$, there exists $\nu$ such that \[\nu_1
\xrightarrow{\alpha_1} \nu \xrightarrow{\alpha_2} \nu_2. \]
By induction, there exists $\nu_3$ such that $(\nu_1,\nu_3) \in
\sem{\beta_1}$ and $\nu_3(\rho_1(y))=\nu(y)$ for $y \in
O(\alpha_1)$.  Also by induction, there exists $\nu_4$ such that
$(\nu,\nu_4) \in \sem{\beta_2}$ and $\nu_4(\rho_2(y))=\nu_2(y)$ for
$y \in O(\alpha_2)$.  By the Renaming Lemma (\ref{thm:subs}), we
have $(\nu \circ \theta,\nu_4 \circ \theta) \in
\sem{\theta(\beta_2)}$.

We claim that $\nu_3$ agrees with $\nu \circ \theta$ on
$I(\theta(\beta))$.  Recalling that the latter equals
$(I(\alpha_2) - O(\alpha_1)) \cup \rho_1(I(\alpha_2) \cap
O(\alpha_1))$, we verify this claim as follows.
\begin{itemize}
\item
We begin by verifying that $\theta$ is the identity on $I(\alpha_2)
- O(\alpha_1)$.  Indeed, let $u \in I(\alpha_2)-O(\alpha_1)$.
Note that $\theta$ is the identity outside
$(I(\alpha_2) \cap O(\alpha_1)) \cup 
\rho_1(I(\alpha_2) \cap O(\alpha_1))$.  Clearly $u$ does not
belong to the first term.  Also $u$ does not belong to the second
term, since the image of $\rho_1$ is disjoint from $\var(\alpha)$.
\item
Now let $u \in I(\alpha_2) - O(\alpha_1)$.  Then $\theta(u)=u$, so
$(\nu \circ \theta)(u)=\nu(u)$.  Since $\nu_1
\xrightarrow{\alpha_1} \nu$ and $u$ does not belong to $O(\alpha_1)$, we have
$\nu(u)=\nu_1(u)$.  Also,  $\nu_1 \xrightarrow{\beta_1} \nu_3$ and
$u$ does not belongs to $O(\beta_1)$ since $O(\beta_1)$ is disjoint
from $\var(\alpha)$.  Hence, $\nu_1(u)=\nu_3(u)$ so we get $\nu(u)=\nu_3(u)$.
\item
Let $u \in I(\alpha_2) \cap O(\alpha_1)$.  Then $(\nu \circ
\theta)(\rho_1(u)) = \nu(\theta(\theta(u))) \allowbreak = \nu(u)$. The latter
equals $\nu_3(\rho_1(u))$ by definition of $\nu_3$.
\end{itemize}

We can now apply input-output determinacy and obtain $\nu_5$ such that
$(\nu_3,\nu_5) \in \sem{\theta(\beta_2)}$
and $\nu_5$ agrees with $\nu_4 \circ \theta$ on
$O(\beta_2)$.  It follows that $(\nu_1,\nu_5) \in \sem\beta$, so
we are done if we can show that $\nu_5(\rho(y))=\nu_2(y)$ for $y
\in O(\alpha)$.  We distinguish two cases.

First, assume $y \in O(\alpha_2)$.  Then $\nu_5(\rho(y)) =
\nu_5(\rho_2(y)) = (\nu_4 \circ \theta)(\rho_2(y))$ by definition
of $\nu_5$.  Now observe that $\theta(\rho_2(y)) = \rho_2(y)$.
Indeed, $\rho_2(y)$ belongs to $O(\beta_2)$, while $\theta$ is
the identity outside $(I(\alpha_2) \cap O(\alpha_1)) \cup 
\rho_1(I(\alpha_2) \cap O(\alpha_1))$.  The first term is
disjoint from $O(\beta_2)$ since $O(\beta_2)$ is disjoint from
$\var(\alpha)$.  The second term is disjoint from $O(\beta_2)$ as
already shown in the io-disjointness proof.  So, we obtain
$\nu_4(\rho_2(y))$, which equals $\nu_2(y)$ by definition of
$\nu_4$.

Second, assume $y \in O(\alpha_1) - O(\alpha_2)$.  Since
$(\nu_3,\nu_5) \in \sem{\theta(\beta_2)}$ and
$O(\theta(\beta_2))=O(\beta_2)$ is disjoint from
$\rho(O(\alpha_1)-O(\alpha_2))$ as seen in the disjointness
proof, $\nu_5(\rho(y))=\nu_3(\rho(y))$.  Since $y \notin
O(\alpha_2)$, we have $\nu_3(\rho(y)) = \nu_3(\rho_1(y))$, which equals
$\nu(y)$ by definition of $\nu_3$. Now $\nu(y)=\nu_2(y)$ since
$(\nu,\nu_2) \in \sem{\alpha_2}$ and $y \notin O(\alpha_2)$.

\paragraph{Soundness}  The proof for soundness is remarkably
symmetrical to that for completeness.  Such symmetry is not
present in the proofs for the other operators.  We cannot yet
explain well why the symmetry is present onlu for composition.

Since $(\nu_1,\nu_2) \in \sem\beta$, there exists $\nu$ such that
\[\nu_1 \xrightarrow{\beta_1} \nu \xrightarrow{\theta(\beta_2)}
\nu_2. \]By induction, there exists $\nu_3$ such that
$(\nu_1,\nu_3) \in \sem{\alpha_1}$ and $\nu_3(y)=\nu(\rho_1(y))$
for $y \in O(\alpha_1)$.  By the Renaming Lemma, we have $(\nu
\circ \theta,\nu_2 \circ \theta) \in \sem{\beta_2}$ (note that
$\theta^{-1} = \theta$).  By induction, there exists $\nu_4$ such
that $(\nu \circ \theta,\nu_4) \in \sem{\alpha_2}$ and
$\nu_4(y)=(\nu_2 \circ \theta)(\rho_2(y))$ for $y \in
O(\alpha_2)$.

Using analogous reasoning as in the completeness proof, it can be
verified that $\nu_3$ agrees with $\nu \circ \theta$ on
$I(\alpha_2)$.  Hence, by input-output determinacy, there exists $\nu_5$
such that $(\nu_3,\nu_5) \in \sem{\alpha_2}$ and $\nu_5$ agrees
with $\nu_4$ on $O(\alpha_2)$.  It follows that $(\nu_1,\nu_5)
\in \sem\alpha$, so we are done if we can show that
$\nu_5(y)=\nu_2(\rho(y))$ for $y \in O(\alpha)$.  This is shown by
analogous reasoning as in the completeness proof. 
\newpage
\subsubsection{Union} \label{subunion}

\paragraph{Inputs} Let $\{i,j\}=\{1,2\}$. We begin by
noting:
\begin{align*}
I(\gamma_i) &= O(\alpha_j) - O(\alpha_i) \\
O(\gamma_i) &= \rho(O(\alpha_j) - O(\alpha_i)) 
\end{align*}
Note that $I(\gamma_i)$, being a subset of $\var(\alpha)$,
is disjoint from $O(\beta_i)$, so $I(\beta_i \comp \gamma_i)$ is
simply $I(\beta_i) \cup I(\gamma_i)$. By induction,
$I(\beta_i)=I(\alpha_i)$ and $O(\beta_i)$ contains
$\rho(O(\alpha_i))$.  Hence:
\begin{align*}
I(\beta_i \comp \gamma_i) &= I(\alpha_i) \cup
(O(\alpha_j)-O(\alpha_i)) \\
O(\beta_i \comp \gamma_i) &= O(\beta_i) \cup \rho(O(\alpha_j))
\end{align*}

We next analyze $\eta_i$. Recall that this expression was defined
by two cases.
\begin{enumerate}
\item[(a)] If $O(\beta_i)$ is nonempty, $I(\eta_i) \subseteq
O(\beta_i)$.
\item[(b)] Otherwise, $I(\eta_i) \subseteq \var(\alpha_j)$.
However, if $O(\beta_i)$ is empty then $O(\alpha_i)$ is too, so
that $I(\alpha) = I(\alpha_i) \cup I(\alpha_j) \cup O(\alpha_j) =
I(\alpha_i) \cup \var(\alpha_j)$. Hence, in this case, $I(\eta_i)
\subseteq I(\alpha)$.
\end{enumerate}
The output is the same in both cases:
\[O(\eta_i) = O(\beta_j) - \rho(O(\alpha_j)) \]

Composing $\beta_i \comp \gamma_i$ with $\eta_i$, we continue
with the two above cases.
\begin{enumerate}
\item[(a)] In this case $I(\eta_i)$ is contained in
$O(\beta_i\comp \gamma_i)$, so $I(\beta_i\comp \gamma_i \comp
\eta_i) = I(\beta_i \comp \gamma_i)$.
\item[(b)] In this case $I(\eta_i)$ is disjoint from
$O(\beta_i\comp \gamma_i)$, and $I(\beta_i \comp \gamma_i \comp
\eta_i)$ equals $I(\beta_i \comp \gamma_i)$ to which some element
of $I(\alpha)$ is added.
\end{enumerate}
In both cases, we can state that \[I(\alpha_i) \cup
(O(\alpha_j) - O(\alpha_i)) \subseteq I(\beta_i \comp \gamma_i
\comp \eta_1) \subseteq I(\alpha). \]
For outputs, we have \[O(\beta_i \comp \gamma_i \comp \eta_i) =
O(\beta_1) \cup O(\beta_2). \]

The set of inputs of the final expression
$\beta = (\beta_1 \comp \gamma_1 \comp \eta_1) \cup
(\beta_2 \comp \gamma_2 \comp \eta_2)$ equals the
union of inputs of the two top-level subexpressions, since these
two subexpressions have the same outputs ($O(\beta_1)\cup
O(\beta_2)$).  Hence
\[I(\alpha_1) \cup I(\alpha_2) \cup (O(\alpha_1) \symdif
O(\alpha_2)) \subseteq I(\beta) \subseteq I(\alpha). \]Since the
left expression equals $I(\alpha)$ by definition, we obtain that
$I(\beta)=I(\alpha)$ as desired.

\paragraph{Outputs} From the above we have
$O(\beta)=O(\beta_1) \cup O(\beta_2)$. Since $O(\beta_i)
\supseteq \rho(O(\alpha_i))$ by induction, we obtain $O(\beta)
\supseteq \rho(O(\alpha_1) \cup O(\alpha_2)) = \rho(O(\alpha))$
as desired.

\paragraph{io-disjointness} Let $i=1,2$.  Expression
$\gamma_i$ is io-disjoint since the image of $\rho$ is disjoint
from $\var(\alpha)$.  Then $\beta_i \comp \gamma_i$ is io-disjoint
because both $O(\beta_i)$ and the image of $\rho$ are disjoint
from $\var(\alpha)$.  For the same reason, $\beta_i \comp \gamma_i
\comp \eta_i$ and $\beta$ are io-disjoint.  We still need to look
at $\eta_i$.  In case (b), $I(\eta_i) \subseteq I(\alpha)$ so
io-disjointness follows again because $O(\beta_j)$ is disjoint
from $\var(\alpha)$. In case (a), we look at $i=1$ and $i=2$
separately.  For $i=1$ we observe that
$O(\eta_1)=O(\beta_2)-\rho(O(\alpha_2))$ is disjoint from $W_2$,
which includes $O(\beta_1)$.  For $i=2$ we write $O(\beta_2) =
\rho(O(\alpha_2)) \cup (O(\beta_2)-\rho(O(\alpha_2)))$.  The first
term is disjoint from $O(\eta_2)=O(\beta_1) - \rho(O(\alpha_1))$
since the latter is disjoint from $W_1$ which includes
$\rho(O(\alpha_2))$.  The second term is disjoint from
$O(\beta_1)$ as we have just seen.

\paragraph{No clashes} We verify:
\begin{align*}
O(\beta) - \rho(O(\alpha)) &= (O(\beta_1) \cup O(\beta_2)) -
\rho(O(\alpha_1) \cup O(\alpha_2)) \\
&\subseteq (O(\beta_1)-\rho(O(\alpha_1))) \cup
(O(\beta_2)-\rho(O(\alpha_2))). 
\end{align*}
By induction, both of the latter terms are disjoint from $W$, which confirms
that there are no clashes.

\paragraph{Completeness}

Assume $(\nu_1,\nu_2) \in \sem{\alpha_1}$; the reasoning for
$\alpha_2$ is analogous.  By induction, there exists $\nu_3$ such
that $(\nu_1,\nu_3) \in \sem{\beta_1}$ and
$\nu_3(\rho(y))=\nu_2(y)$ for $y \in O(\alpha_1)$.

Note that each of the expressions $\gamma_i$ and $\eta_i$ for
$i=1,2$ is a composition of variable assignments. For any such
expression $\delta$ and any valuation $\nu$ there always exists a
unique $\nu'$ such that $(\nu,\nu') \in \sem\delta$ (even
independently of $D$).

Now let \[\nu_3 \xrightarrow{\gamma_1} \nu_4
\xrightarrow{\eta_1} \nu_5,
\]so that $(\nu_1,\nu_5) \in \sem\beta$.  If we can show that
$\nu_5(\rho(y))=\nu_2(y)$ for $y \in O(\alpha)$ we are done.
Thereto, first note that $\eta_1$ does not change variables in
$\rho(O(\alpha))$.  Indeed, for $\rho(O(\alpha_2))$ this is
obvious from $O(\eta_1)=O(\beta_2)-\rho(O(\alpha_2))$;
for $\rho(O(\alpha_1))$ this follows because by induction,
$O(\beta_2) - \rho(O(\alpha_2))$ is disjoint from $W_2$, which
includes $O(\beta_1)$, which includes $\rho(O(\alpha_1))$.
So, by $\nu_4 \xrightarrow{\eta_1} \nu_5$ we are down to showing that
$\nu_4(\rho(y))=\nu_2(y)$ for $y \in O(\alpha)$. We distinguish
two cases.

If $y \in O(\alpha_1)$, since $\nu_3 \xrightarrow{\gamma_1}
\nu_4$ and $\gamma_1$ does not change variables in
$\rho(O(\alpha_1))$, we have $\nu_4(\rho(y))=\nu_3(\rho(y))$,
which equals $\nu_2(y)$ by definition of $\nu_3$.

If $y \in O(\alpha_2) - O(\alpha_1)$, then $\nu_4(\rho(y)) =
\nu_3(y)$ by $\nu_3 \xrightarrow{\gamma_1} \nu_4$.  Now since
\[\nu_3 \xleftarrow{\beta_1} \nu_1 \xrightarrow{\alpha_1} \nu_2
\]and $y \notin O(\beta_1) \cup O(\alpha_1)$, we get
$\nu_3(y)=\nu_2(y)$ as desired.  (The reason for $y \notin
O(\beta_1)$ is that by induction, $O(\beta_1)$ is disjoint from
$W_1$ which includes $\var(\alpha)$.)

\paragraph{Soundness} Assume $(\nu_1,\nu_2) \in \sem{\beta_1
\comp \gamma_1 \comp \eta_1}$; the reasoning for $\beta_2 \comp
\gamma_2 \comp \eta_2$ is analogous.  Then there exist $\nu_3$ and
$\nu_4$ such that \[\nu_1 \xrightarrow{\beta_1} \nu_3
\xrightarrow{\gamma_1} \nu_4 \xrightarrow{\eta_1} \nu_2. \eqno
(*) \]By induction, there exists $\nu$ such that $(\nu_1,\nu)
\in \sem{\alpha_1} \subseteq \sem\alpha$
and $\nu(y)=\nu_3(\rho(y))$ for $y \in
O(\alpha_1)$.  As observed in the completeness proof, $\gamma_1$
and $\eta_1$ do not touch variables in $\rho(O(\alpha_1))$.
Since $(*)$ shows that $\gamma_1$ followed by $\eta_1$ maps $\nu_3$ to $\nu_2$,
also $\nu(y)=\nu_2(\rho(y))$ for $y \in O(\alpha_1)$.

If we can show the same for $y \in O(\alpha_2)-O(\alpha_1)$, we
have covered all $y \in O(\alpha)$ and we are done.  This is
verified as follows.  By inertia, we have
$\nu(y)=\nu_1(y)=\nu_3(y)$, the latter equality because
$O(\beta_1)$ is disjoint from $\var(\alpha)$.  From $\nu_3
\xrightarrow{\gamma_1} \nu_4$ we have $\nu_3(y) =
\nu_4(\rho(y))$.  Now the latter equals $\nu_2(\rho(y))$ since
$\nu_4 \xrightarrow{\eta_1} \nu_2$ and $\eta_1$ does not touch
variables in $\rho(O(\alpha_2))$. 
\newpage
\subsubsection{Difference} \label{subdifference}

\paragraph{Inputs} Let $\{i,j\}=\{1,2\}$.  We begin by
noting:
\begin{align*}
I(\gamma_i) &= O(\alpha_j) - O(\alpha_i) \\
O(\gamma_i) &= \rho_j(O(\alpha_j)-O(\alpha_i))
\end{align*}
Slightly adapting the calculation of inputs in the proof for union
(Section~\ref{subunion}),
we next note:
\begin{align*}
I(\beta_i \comp \gamma_i) &= I(\alpha_i) \cup (O(\alpha_j) -
O(\alpha_i)) \\
O(\beta_i \comp \gamma_i) &=
O(\beta_i) \cup \rho_j(O(\alpha_j) - O(\alpha_i))
\end{align*}

We next analyze $\eta_1$. Recall that this expression was defined
by two cases.
\begin{enumerate}
\item[(a)]
If $O(\alpha_1)$ and $O(\alpha_2)$
intersect, $I(\eta_1) \subseteq \rho(O(\alpha_1) \cap O(\alpha_2))$.
\item[(b)]
Otherwise, $I(\eta_1) \subseteq \var(\alpha_2)$.  However, note in
this case that $I(\alpha)=\var(\alpha)$, so that $I(\eta_1)
\subseteq I(\alpha)$.
\end{enumerate}
Regardless of the case, \[O(\eta_1) = O(\beta_2) -
\rho_2(O(\alpha_2)). \]

Composing $\beta_i \comp \gamma_i$ with $\eta_1$, we
continue with
the above two cases.
\begin{enumerate}
\item[(a)]
By induction, $O(\beta_i)$
contains $\rho_i(O(\alpha_i))$, and $\rho$ agrees
$\rho_i$ on $O(\alpha_1) \cap
O(\alpha_2)$.  Hence $I(\eta_1) \subseteq O(\beta_i \comp
\gamma_i)$ and thus
\begin{align*}
I(\beta_i \comp \gamma_i \comp \eta_1) 
= I(\beta_i \comp
\gamma_i) 
= I(\alpha_i) \cup (O(\alpha_j) - O(\alpha_i)).
\end{align*}
\item[(b)]
In this case $I(\eta_1) \subseteq I(\alpha)$
which is disjoint from $O(\beta_i \comp \gamma_i)$.  Note that
also $I(\beta_i \comp \gamma_i) \subseteq I(\alpha)$.
\end{enumerate}
In both cases, we can state that \[I(\alpha_i) \cup
(O(\alpha_j) - O(\alpha_i)) \subseteq I(\beta_i \comp \gamma_i
\comp \eta_1) \subseteq I(\alpha). \]
For outputs, we have
\[
O(\beta_i \comp \gamma_i \comp \eta_1) = {} 
O(\beta_i) \cup \rho_j(O(\alpha_j) - O(\alpha_i)) \cup
(O(\beta_2) - \rho_2(O(\alpha_2))).
\]

Composing further with $\eta_2$, which is defined similarly to
$\eta_1$, we can reason similarly and still state that
\[I(\alpha_i) \cup
(O(\alpha_j) - O(\alpha_i)) \subseteq I(\beta_i \comp \gamma_i
\comp \eta_1 \comp \eta_2) \subseteq I(\alpha). \]
For outputs, note that
$O(\eta_2)=O(\beta_1)-\rho_1(O(\alpha_1))$.  Uniting this to the
expression for $O(\beta_i \comp \gamma_i \comp \eta_1)$ above, we
obtain \[O(\beta_i \comp \gamma_i \comp \eta_1 \comp \eta_2) =
O(\beta_1) \cup O(\beta_2). \] Indeed, the only part of
$O(\beta_1) \cup O(\beta_2)$ that is not obviously there is
$\rho_j(O(\alpha_j)\cap O(\alpha_i))$.  However,
that part is contained in $\rho(O(\alpha_i))$, because
$\rho_j$ agrees with $\rho$ on $O(\alpha_1)\cap O(\alpha_2)$.
Since $\rho(O(\alpha_i)) \subseteq O(\beta_i)$, the part is
included after all.

With the above results we can reason exactly as in the
proof for union and obtain that $I(\beta)=I(\alpha)$ as desired.

\paragraph{Outputs} From the above we have
$O(\beta)=O(\beta_1)\cup O(\beta_2)$. Since $O(\beta_1) \supseteq
\rho_1(O(\alpha_1))$ by induction, and $\rho_1=\rho$ and
$O(\alpha)=O(\alpha_1)$, we obtain $O(\beta) \supseteq
\rho(O(\alpha))$ as desired.

\paragraph{io-disjointness} Let $i=1,2$. Expression
$\gamma_i$ is io-disjoint by the choice of $\rho_j$.  Then
$\beta_i\comp \gamma_i$ is io-disjoint because both $O(\beta_i)$ and
the image of $\rho_j$ are disjoint from $\var(\alpha)$.  Regarding
$\eta_1$, we have seen that either (a) $I(\eta_1) \subseteq
\rho(O(\alpha_1)\cap O(\alpha_2)) \subseteq \rho_2(O(\alpha_2))$,
or (b) $I(\eta_1) \subseteq I(\alpha)$. In case (a) $I(\eta_1)$
is clearly disjoint from $O(\eta_1)=O(\beta_2) -
\rho_2(O(\alpha_2))$. Also in case (b) $\eta_1$ is io-disjoint
because $O(\beta_2)$ is disjoint from $\var(\alpha)$.
Using similar reasoning, the expressions $\beta_i \comp \gamma_i
\comp \eta_1$, $\eta_2$, $\beta_i \comp \gamma_i \comp \eta_1
\comp \eta_2$, and finally $\beta$, are seen to be io-disjoint.

\paragraph{No clashes} Note that $\rho(O(\alpha)) =
\rho_1(O(\alpha_1))$, and recall that $O(\beta)=O(\beta_1) \cup
O(\beta_2)$.  Hence we can write $O(\beta) -
\rho(O(\alpha))$ as \[(O(\beta_1)-\rho_1(O(\alpha_1))) \cup
(O(\beta_2) - \rho_1(O(\alpha_1))). \]The first term is disjoint
from $W$ by construction and induction.  For the second term,
note that $O(\beta_2)$ can be written as a disjoint union
\[\rho_2(O(\alpha_2) \cap O(\alpha_1)) \cup
\rho_2(O(\alpha_2)-O(\alpha_1)) \cup (O(\beta_2) -
\rho_2(O(\alpha_2))). \]Again by construction and induction,
the second and third terms are disjoint from $W_2$, which
includes $O(\beta_1)$, which includes $\rho_1(O(\alpha_1))$.
On the other hand, the first term is included in
$\rho_1(O(\alpha_1))$ since $\rho_1$ and $\rho_2$ agree on
$O(\alpha_1)\cap O(\alpha_2)$.
Hence, $O(\beta_2)-\rho_1(O(\alpha_1))$ reduces to the union of
the second and third terms, which are disjoint from $W_2$, which
includes $W$, as desired.

\paragraph{Completeness}

Since $(\nu_1,\nu_2) \in \sem{\alpha_1-\alpha_2}$, in particular
$(\nu_1,\nu_2) \in \sem{\alpha_1}$, so by induction there exists
$\nu_3$ such that $(\nu_1,\nu_3) \in \sem{\beta_1}$ and
$\nu_3(\rho_1(y)) = \nu_2(y)$ for $y \in O(\alpha_1)$.

Recall the output variables of $\gamma_i$
and $\eta_i$ for $i=1,2$:
\begin{align*}
O(\gamma_1) &= \rho_2(O(\alpha_2) - O(\alpha_1)) \\
O(\gamma_2) &= \rho_1(O(\alpha_1) - O(\alpha_2)) \\
O(\eta_1) &= O(\beta_2) - \rho_2(O(\alpha_2)) \\
O(\eta_2) &= O(\beta_1) - \rho_1(O(\alpha_1))
\end{align*}
We observe:
\begin{enumerate}
\item\label{obs:item1}
\emph{None of the assignments in $\gamma_1$, $\eta_1$ or $\eta_2$
affect variables in $\rho_1(O(\alpha_1))$.}

This claim is clear for $\eta_2$.  For $\gamma_1$ it holds since
$\rho_2$ was chosen such that its image on
$O(\alpha_2)-O(\alpha_1)$ is disjoint from $W_2$,
which includes $O(\beta_1)$, which includes $\rho_1(O(\alpha_1))$.
For $\eta_1$ the claim
holds because, by induction, $O(\eta_1)$ is again disjoint from $W_2$. 
\item\label{obs:item2}
\emph{None of the assignments in $\gamma_2$, $\eta_1$ or $\eta_2$
affect variables in $\rho_2(O(\alpha_2))$.}

This claim is clear for $\eta_1$. Next consider $\gamma_2$.
On $O(\alpha_2) - O(\alpha_1))$, we just noted that the image of
$\rho_2$ is disjoint from $\rho_1(O(\alpha_1))$. Now let $y \in
O(\alpha_2) \cap O(\alpha_1)$.  Then $\rho_2(y)=\rho_1(y)$ and
clearly $\rho_1(y) \notin \rho_1(O(\alpha_1)-O(\alpha_2))$.
Finally, consider $\eta_2$.  On $O(\alpha_2)-O(\alpha_1)$, we
again use that the image of $\rho_2$ is disjoint from
$O(\beta_1)$.  On $O(\alpha_1)\cap O(\alpha_2)$, again the image
of $\rho_2$ agrees with the image of $\rho_1$, which clearly
is disjoint from $O(\eta_2)$.
\end{enumerate}

Now, using the notation introduced in the completeness proof for
union (Section~\ref{subunion}), let 
\[
\nu_3 \xrightarrow{\gamma_1} \nu_4
\xrightarrow{\eta_1} \nu_5 \xrightarrow{\eta_2} \nu_6
\]
so that $(\nu_1,\nu_6) \in \sem{\beta_1 \comp \gamma_1 \comp
\eta_1 \comp \eta_2}$.  By Observation~(\ref{obs:item1}), for $y
\in O(\alpha)=O(\alpha_1)$, we still have
$\nu_6(\rho_1(y))=\nu_3(\rho_1(y))=\nu_2(y)$.  Thus, completeness is
proved provided we can show that $(\nu_1,\nu_6) \notin
\sem{\beta_2 \comp \gamma_2 \comp \eta_1 \comp \eta_2}$.

For the sake of contradiction, assume $(\nu_1,\nu_6) \in
\sem{\beta_2 \comp \gamma_2 \comp \eta_1 \comp \eta_2}$.  By the
identity property (Proposition~\ref{ip}), also $(\nu_6,\nu_6) \in
\sem{\beta_2 \comp \gamma_2 \comp \eta_1 \comp \eta_2}$.  Hence,
there exists $\nu_7$ such that $(\nu_6,\nu_7) \in \sem{\beta_2}$
and $(\nu_7,\nu_6) \in \sem{\gamma_2 \comp \eta_1 \comp \eta_2}$.
By inertia, $\nu_6$ and $\nu_7$ can differ only on $O(\beta_2)$,
and among $\gamma_2$, $\eta_1$ and $\eta_2$, only $\eta_1$
can change variables in $O(\beta_2)$.  Hence we have 
\[
\nu_7 \xrightarrow{\gamma_2} \nu_7 \xrightarrow{\eta_1} 
\nu_6 \xrightarrow{\eta_2} \nu_6.
\]

Since $(\nu_6,\nu_7) \in \sem{\beta_2}$, by induction there
exists $\nu_8$ such that $(\nu_6,\nu_8) \in \sem{\alpha_2}$ and
$\nu_8(y)=\nu_7(\rho_2(y))$ for $y \in O(\alpha_2)$.  Recall that
$(\nu_1,\nu_6) \in \sem{\beta_1\comp \gamma_1 \comp \eta_1 \comp
\eta_2}$, so $\nu_1$ and $\nu_6$ agree outside $O(\beta_1) \cup
O(\beta_2)$, which is disjoint from $\var(\alpha)$ which includes
$I(\alpha_2)$.  Hence we can apply input-output determinacy, yielding
a valuation $\nu$ such that $(\nu_1,\nu) \in \sem{\alpha_2}$ and
$\nu$ agrees with $\nu_8$ on $O(\alpha_2)$.  If we can show that
$\nu = \nu_2$ we have arrived at a contradiction, since
$(\nu_1,\nu_2) \notin \sem{\alpha_2}$.

By inertia, $\nu$ and $\nu_1$ agree outside $O(\alpha_2)$, and
$\nu_1$ and $\nu_2$ agree outside $O(\alpha_1)$.  Thus 
$\nu$ and $\nu_2$ already agree outside
$O(\alpha_1) \cup O(\alpha_2)$ and we can focus on
that set of variables.  We distinguish three cases.

First, let $y \in O(\alpha_1) \cap O(\alpha_2)$.  Note that
$\rho(y)=\rho_1(y)=\rho_2(y)$.  By definition of $\nu$ and
$\nu_8$, we have $\nu(y)=\nu_8(y)=\nu_7(\rho(y))$.  Since 
\begin{equation}
\nu_3 \xrightarrow{\gamma_1} \nu_4 \xrightarrow{\eta_1} \nu_5
\xrightarrow{\eta_2} \nu_6 \xleftarrow{\eta_1} \nu_7, \label{eq:star}    
\end{equation}
by Observation~(\ref{obs:item1}), we have $\nu_7(\rho(y))=\nu_3(\rho(y))$.
The latter indeed equals $\nu_2(y)$, by definition of $\nu_3$.

Second, let $y \in O(\alpha_2)-O(\alpha_1)$.  As before we have
$\nu(y)=\nu_7(\rho_2(y))$.  By~(\ref{eq:star}) and Observation~(\ref{obs:item2}),
$\nu_7(\rho_2(y))=\nu_4(\rho_2(y))$.  The latter equals
$\nu_3(y)$ since $\nu_3 \xrightarrow{\gamma_1} \nu_4$.  Now since
$(\nu_1,\nu_3) \in \sem{\beta_1}$ and $(\nu_1,\nu_2) \in
\sem{\alpha_1}$ and $y$ is neither in $O(\beta_1)$ nor in
$O(\alpha_1)$, we get $\nu_3(y)=\nu_1(y)=\nu_2(y)$.

Third, let $y \in O(\alpha_1)-O(\alpha_2)$.
Since $(\nu_1,\nu) \in \sem{\alpha_2}$ and $y \notin
O(\alpha_2)$, by inertia $\nu(y)=\nu_1(y)$.  Likewise, since \[\nu_1
\xrightarrow{\beta_1 \comp \gamma_1 \comp \eta_1 \comp \eta_2}
\nu_6 \xrightarrow{\beta_2} \nu_7 \]and $y \notin O(\beta_1) \cup
O(\beta_2)$, we get $\nu_1(y)=\nu_6(y)=\nu_7(y)$.  Since $\nu_7
\xrightarrow{\gamma_2} \nu_7$ we have
$\nu_7(y)=\nu_7(\rho_1(y))$.  In the first case we already noted
that $\nu_7(\rho_1(y))=\nu_3(\rho_1(y))$.
Now the latter equals $\nu_2(y)$ by definition of $\nu_3$, and we
are done.

\paragraph{Soundness} Since $(\nu_1,\nu_2) \in \sem \beta$,
we have $(\nu_1,\nu_2) \in
\sem{\beta_1 \comp \gamma_1 \comp \eta_1 \comp \eta_2}$.
By the identity property, also
$(\nu_2,\nu_2) \in 
\sem{\beta_1 \comp \gamma_1 \comp \eta_1 \comp \eta_2}$.
Hence there exists $\nu_3$ such that $(\nu_2,\nu_3) \in
\sem{\beta_1}$ and $(\nu_3,\nu_2) \in \sem{\gamma_1 \comp \eta_1
\comp \eta_2}$.  By inertia, $\nu_2$ and $\nu_3$ can differ only
on $O(\beta_1)$, and among $\gamma_1$, $\eta_1$ and $\eta_2$,
only $\eta_2$ can change variables in $O(\beta_1)$.  Hence we
have 
\[
\nu_3 \xrightarrow{\gamma_1} \nu_3 \xrightarrow{\eta_1}
\nu_3 \xrightarrow{\eta_2} \nu_2.
\]

Since $(\nu_2,\nu_3) \in \sem{\beta_1}$, by induction there
exists $\nu_4$ such that $(\nu_2,\nu_4) \in \sem{\alpha_1}$ and
$\nu_4(y)=\nu_3(\rho_1(y))$ for $y \in O(\alpha_1)$.  Note that
$\nu_1$ and $\nu_2$ agree outside $O(\beta)$ which is disjoint
from $I(\alpha_1)$.  Hence we can apply input-output determinacy,
yielding a valuation $\nu$ such that $(\nu_1,\nu) \in
\sem{\alpha_1}$ and $\nu$ agrees with $\nu_4$ on $O(\alpha_1)$.
Our goal is to show that $(\nu_1,\nu) \notin \sem{\alpha_2}$.

For the sake of contradiction, assume $(\nu_1,\nu) \in
\sem{\alpha_2}$.  Then by induction, there exists $\nu_5$ such
that $(\nu_1,\nu_5) \in \sem{\beta_2}$ and $\nu_5(\rho_2(y)) =
\nu(y)$ for $y \in O(\alpha_2)$.  Let 
\begin{equation}
\nu_5 \xrightarrow{\gamma_2} \nu_6 \xrightarrow{\eta_1} \nu_7
\xrightarrow{\eta_2} \nu_8 \label{eq:dag}
\end{equation}
so that $(\nu_1,\nu_8)
\in \sem{\beta_2 \comp \gamma_2 \comp \eta_1 \comp \eta_2}$.  If
we can show that $\nu_8=\nu_2$, we have arrived at the desired
contradiction since $(\nu_1,\nu_2) \notin
\sem{\beta_2 \comp \gamma_2 \comp \eta_1 \comp \eta_2}$.

By inertia, $\nu_8$ and $\nu_1$, and $\nu_1$ and $\nu_2$,
agree outside $O(\beta_1) \cup O(\beta_2)$.  Thus $\nu_8$ and
$\nu_2$ already agree outside $O(\beta_1) \cup O(\beta_2)$ and we
can focus on that set of variables.  Note that $O(\beta_1)$ contains
$\rho_1(O(\alpha_1))$ and $O(\beta_2)$ contains
$\rho_2(O(\alpha_2))$.  Accordingly, we distinguish five cases.
\begin{enumerate}
\item\label{item:case1} $\rho(O(\alpha_1) \cap O(\alpha_2))$.
Let $y \in O(\alpha_1) \cap O(\alpha_2)$.  By~(\ref{eq:dag}) and
Observation~(\ref{obs:item2}), $\nu_8(\rho(y)) = \nu_5(\rho(y))$.  
By definition of $\nu_5$, $\nu$ and $\nu_4$ respectively,
$\nu_5(\rho(y))=\nu(y)=\nu_4(y)=\nu_3(\rho(y))$.  The latter 
equals $\nu_2(\rho(y))$ since $\nu_3 \xrightarrow{\eta_2} \nu_2$.

\item\label{item:case2} $\rho_2(O(\alpha_2)-O(\alpha_1))$.
Let $y \in O(\alpha_2)-O(\alpha_1)$.  As in case~(\ref{item:case1}),
$\nu_8(\rho_2(y))=\nu(y)$.  Since \[\nu \xleftarrow{\alpha_1}
\nu_1 \xrightarrow{\beta_1 \comp \gamma_1 \comp \eta_1 \comp
\eta_2} \nu_2 \xrightarrow{\beta_1} \nu_3 \]and $y \notin
O(\alpha_1) \cup O(\beta_1) \cup O(\beta_2)$, we have
$\nu(y)=\nu_3(y)$.  The latter equals $\nu_3(\rho_2(y))$ by 
$\nu_3 \xrightarrow{\gamma_1} \nu_3$.
Now by $\nu_3 \xrightarrow{\eta_2} \nu_2$ and 
Observation~(\ref{obs:item2}) we
get $\nu_3(\rho_2(y))=\nu_2(\rho_2(y))$.

\item\label{item:case3} $\rho_1(O(\alpha_1)-O(\alpha_2))$. 
Let $y \in O(\alpha_1)-O(\alpha_2)$.  By~(\ref{eq:dag}) and 
Observation~(\ref{obs:item1}), $\nu_8(\rho_1(y))=\nu_6(\rho_1(y))$.  
The latter equals $\nu_5(y)$ by $\nu_5 \xrightarrow{\gamma_2} \nu_6$. 
Since 
\[ \nu_5 \xleftarrow{\beta_2} \nu_1 \xrightarrow{\alpha_2} \nu \]
and $y \notin O(\beta_2) \cup O(\alpha_2)$, we get
$\nu_5(y) = \nu(y)$.  The latter equals $\nu_2(\rho_1(y))$ as in
case~(\ref{item:case1}).

\item\label{item:case4} $O(\beta_2) - \rho_2(O(\alpha_2))$.
Let $y \in O(\beta_2) - \rho_2(O(\alpha_2))$.
We distinguish two further cases following the definition of
$\eta_1$, which involves the choice of a variable $z$.
\begin{enumerate}
\item $z = \rho(x)$ for some $x \in O(\alpha_1) \cap
O(\alpha_2)$.
Since $\nu_7 \xrightarrow{\eta_2} \nu_8$, we have
$\nu_8(y) = \nu_7(\rho(x))$. The latter equals $\nu_3(\rho(x))$
as in case~(\ref{item:case1}). Now $\nu_3 \xrightarrow{\eta_2} \nu_2$ 
yields $\nu_3(\rho(x))=\nu_2(y)$.
\item In this case $z \in I(\alpha)$ (see the Inputs part 
of this proof). 
Since $\nu_7 \xrightarrow{\eta_2} \nu_8$, we have
$\nu_8(y)=\nu_7(z)$.  Since 
\[
\nu_2 \xleftarrow{\beta_1 \comp
\gamma_1 \comp \eta_1 \comp \eta_2} \nu_1 \xrightarrow{\beta_2
\comp \gamma_2 \comp \eta_1} \nu_7 \]
and $z$, being in $I(\alpha)$, is not an output variable 
of the involved expressions, we have $\nu_7(z)=\nu_2(z)$.  
Since $\nu_3 \xrightarrow{\eta_2} \nu_2$ with $\eta_2$ not 
touching $z$, we obtain $\nu_2(z)=\nu_3(z)=\nu_2(y)$.
\end{enumerate}

\item\label{item:case5} $O(\beta_1) - \rho_1(O(\alpha_1))$. 
This case is symmetrical to the previous one.
\end{enumerate}

The above five cases confirm $\nu_8=\nu_2$ which gives the
contradiction, showing $(\nu_1,\nu) \notin \sem{\alpha_2}$ 
whence $(\nu_1,\nu) \in \sem\alpha$.  In case~(\ref{item:case1}) 
and case~(\ref{item:case3}) we already observed that 
$\nu(y)=\nu_2(\rho_1(y))$ for $y \in O(\alpha_1)=O(\alpha)$.  
Thus, soundness is proved.

\subsection{From io-disjoint FLIF to Executable FO}
\label{proofflio2exfo}

In this section we prove Theorem~\ref{flio2exfo}.
Recall the translation given in Table~\ref{tabtrans}.
In order to prove Theorem~\ref{flio2exfo}, using Lemma~\ref{ip},
it suffices to prove the following:

\begin{clm}
For each $\alpha$, the formula $\varphi_\alpha$ is
$I(\alpha)$-executable and $\FV(\varphi_\alpha)=\var(\alpha)$.
Moreover, for any instance $D$ and any
valuation $\nu$, we have 
\[(\nu,\nu) \in \sem \alpha \iff
\satD{\varphi}{\nu}.\]
\end{clm}
\begin{proof}
By structural induction.  The implication from left to right is
referred to as \emph{completeness}, and the other implication as
\emph{soundness}.

\paragraph{Atomic expressions}
If $\alpha$ is $R(\bar x;\bar y)$, only soundness is not
immediate.  If $\satD{\varphi}{\nu}$, then $\nu(\bar{x}) \conc
\nu(\bar{y}) \in \inst(R)$. Hence, $(\nu, \nu) \in \sem{\alpha}$,
since two identical valuations agree trivially outside
$\Out(\alpha)$.
The cases where $\alpha$ is of the form $(x=y)$, $(x:=y)$,
$(x=c)$, or $(x:=c)$, are immediate.

Next, we verify the inductive cases.  In each step of the induction, 
we refer to $\varphi_\alpha$ simply as $\varphi$.

\paragraph{Composition}
Consider $\alpha$ of the form $\alpha_1 \comp \alpha_2$.
We begin by checking that $\varphi_\alpha$ is
$I(\alpha)$-executable.   By Proposition~\ref{checkio},
$I(\alpha_1)$ is disjoint from both $O(\alpha_1)$ and
$O(\alpha_2)$.

Let $\varphi'_1 = \exists_{\Out(\alpha_1) \cap \Out(\alpha_2)} \varphi_1$. Then 
$\FV(\varphi'_1) = \In(\alpha_1) \cup (\Out(\alpha_1) \setminus \Out(\alpha_2))$. 
Indeed,  
		\begin{align*}
		\FV(\varphi'_1) &  = \FV(\varphi_1) \setminus (\Out(\alpha_1) \cap \Out(\alpha_2)) \\ 
		& = \var(\alpha_1) \setminus (\Out(\alpha_1) \cap \Out(\alpha_2)) \\ 
		& = (\In(\alpha_1) \cup \Out(\alpha_1)) \setminus
		(\Out(\alpha_1) \cap \Out(\alpha_2))  \\
		& = \In(\alpha_1) \cup (\Out(\alpha_1) \setminus \Out(\alpha_2))
		\end{align*}
By induction, $\varphi_1$ is $\In(\alpha_1)$-executable and
$\varphi_2$ is $\In(\alpha_2)$-executable. Let $\VarSet =
\In(\alpha) = \In(\alpha_1) \cup (\In(\alpha_2)
\setminus \Out(\alpha_1))$. For $\varphi$ to be \vex{}, it must
be the case that:
\begin{itemize}
\item $\varphi'_1$ is \vex{}, which means that $\varphi_1$ should be 
\vex{\setminus (\Out(\alpha_1) \cap \Out(\alpha_2))}.
Since $\In(\alpha_1) \cap \Out(\alpha_1) = \emptyset$, we have 
$\In(\alpha_1) \cap (\Out(\alpha_1) \cap \Out(\alpha_2)) = \emptyset$.
This shows that 
$\In(\alpha_1) \subseteq$ $\VarSet \setminus (\Out(\alpha_1) \cap \Out(\alpha_2))$. 
Consequently, $\varphi'_1$ is \vex{}.
			
\item $\varphi_2$ is \vex{\cup \FV(\varphi'_1)}, which means that
$\varphi_2$ should be \vex{\cup (\In(\alpha_1) \cup
(\Out(\alpha_1) \setminus \Out(\alpha_2)) )}.
We know that 
\begin{align*}
\VarSet \cup \FV(\varphi_1') 
& = \VarSet \cup (\In(\alpha_1) \cup (\Out(\alpha_1) \setminus \Out(\alpha_2)) ) \\
& = \In(\alpha_1) \cup (\In(\alpha_2) \setminus \Out(\alpha_1)) \cup (\Out(\alpha_1) \setminus \Out(\alpha_2)) \\
& = \In(\alpha_1) \cup ((\In(\alpha_2) \cup \Out(\alpha_1) )\setminus \Out(\alpha_2))
\end{align*}
Since $\disj{\In(\alpha_2)}{\Out(\alpha_2)}$, we have $\In(\alpha_2) \subseteq \VarSet \cup \FV(\varphi'_1)$. Hence, $\varphi_2$ is \vex{\cup \FV(\varphi'_1)}.
\end{itemize}

We next prove completeness.
To this end, assume that
$(\nu, \nu) \in \sem{\alpha}$. Then there
exists a valuation $\nu'$ such that
	\begin{enumerate}
		\item\label{it:citem1} $(\nu, \nu') \in \sem{\alpha_1}$;
		\item\label{it:citem2} $(\nu', \nu) \in \sem{\alpha_2}$.
	\end{enumerate}
Since $\In(\alpha_2) \cap \Out(\alpha_2) = \emptyset$,
Lemma~\ref{ip} implies
$(\nu, \nu) \in
\sem{\alpha_2}$. Thus, by induction,
$\satD{\varphi_2}{\nu}$. Similarly from~(\ref{it:citem1}), we know that
$\satD{\varphi_1}{\nu'}$. Consequently,
\begin{equation}
\satD{\varphi_1'}{\nu'} \label{eq:c1}
\end{equation}
Additionally, we know from~(\ref{it:citem1}) and~(\ref{it:citem2}) and the law of inertia
$\nu = \nu'$ outside ${\Out(\alpha_1)}$ and outside
${\Out(\alpha_2)}$. Hence, 
\begin{equation}
\nu = \nu' \text{ outside } {\Out(\alpha_1) \cap \Out(\alpha_2)} \label{eq:c2}
\end{equation}
From~(\ref{eq:c1}) and~(\ref{eq:c2}), we obtain $\satD{\varphi_1'}{\nu}$, 
whence $\satD{\varphi}{\nu}$.
	
To show soundness, assume $\satD{\varphi}{\nu}$.  Then
\begin{enumerate}
\item\label{it:citem3} $\satD{\varphi_1'}{\nu}$, which means that there exists $\nu'$ such that
\begin{enumerate}[label=(\roman*)]
\item\label{it:citem31} $\nu' = \nu$ outside
${\Out(\alpha_1) \cap \Out(\alpha_2)}$;
\item\label{it:citem32} $\satD{\varphi_1}{\nu'}$.
\end{enumerate}
\item\label{it:citem4} $\satD{\varphi_2}{\nu}$.
\end{enumerate}
By induction from~\ref{it:citem32}, we know that $(\nu', \nu') \in
\sem{\alpha_1}$. Since $\disj{\In(\alpha_1)}{\Out(\alpha_1)}$,
we know
from~\ref{it:citem31} that $\nu$ agrees with $\nu'$ on $\In(\alpha_1)$ and
outside
${\Out(\alpha_1)}$. Hence, we know by Lemma~\ref{altidlemma}
that $(\nu, \nu') \in \sem{\alpha_1}$. 
Similarly, from~(\ref{it:citem4}), $(\nu', \nu) \in \sem{\alpha_2}$. 
Consequently, $(\nu, \nu) \in \sem{\alpha}$.

\intersection{
\paragraph{Intersection}
For the remaining cases we only check
$I(\alpha)$-executability, as soundness and
completeness are immediate.

By induction, we know that $\varphi_1$ is
$\In(\alpha_1)$-executable and
$\varphi_2$ is $\In(\alpha_2)$-executable. Let $\VarSet=
\In(\alpha) = \In(\alpha_1) \cup \In(\alpha_2)
\cup (\Out(\alpha_1) \symdif \Out(\alpha_2))$. For $\varphi$ to
be \vex{}, we need $\varphi_1$ to be \vex{} and
$\varphi_2$ to be \vex{\cup \FV(\varphi_1)}, which both hold
since $\In(\alpha_i) \subseteq \VarSet$ for $i \in \{1,2\}$.
}

\paragraph{Difference}
By induction, we know that $\varphi_1$ is
$\In(\alpha_1)$-executable and
$\varphi_2$ is $\In(\alpha_2)$-executable. Let $\VarSet=
\In(\alpha)=\In(\alpha_1) \cup \In(\alpha_2)
\cup ( \Out(\alpha_1) \symdif \Out(\alpha_2))$.  By 
Proposition~\ref{checkio}, we have
$\Out(\alpha_1) \subseteq \Out(\alpha_2)$, so
$\VarSet=\In(\alpha)=\In(\alpha_1) \cup \In(\alpha_2) \cup
(O(\alpha_2)-O(\alpha_1))$.

For $\varphi$ to
be \vex{}, we must verify the following:
\begin{itemize}
\item $\varphi_1$ is \vex{} and $\varphi_2$ is \vex{\cup \FV(\varphi_1)}.
This holds since $\In(\alpha_i) \subseteq \VarSet$ for $i \in \{1,2\}$.

\item $\FV(\varphi_2) \subseteq \VarSet \cup \FV(\varphi_1)$.
We verify this as follows.
\begin{align*}
\FV(\varphi_2) & = I(\alpha_2) \cup O(\alpha_2) \\
& \subseteq I(\alpha_1) \cup I(\alpha_2) \cup
(O(\alpha_2)-O(\alpha_1)) \cup O(\alpha_1) \\
& = \VarSet \cup I(\alpha_1) \cup O(\alpha_1)
= \VarSet \cup \FV(\varphi_1).
\end{align*}
\end{itemize}

\paragraph{Union}

By induction, we know that $\varphi_1$ is
$\In(\alpha_1)$-executable and
$\varphi_2$ is $\In(\alpha_2)$-executable. Let $\VarSet=
\In(\alpha) = \In(\alpha_1) \cup \In(\alpha_2)
\cup (\Out(\alpha_1) \symdif \Out(\alpha_2))$.  By
Proposition~\ref{checkio} we have $O(\alpha_1)=O(\alpha_2)$,
so $\VarSet= \In(\alpha_1) \cup \In(\alpha_2)$.

For $\varphi$ to
be \vex{}, we must verify the following:
\begin{itemize}
\item $\varphi_1$ is \vex{} and $\varphi_2$ is \vex{}. This holds
since $\In(\alpha_i) \subseteq \VarSet$ for $i \in \{1,2\}$.
\item $\FV(\varphi_1) \symdif \FV(\varphi_2) \subseteq \VarSet$.
We verify this as follows.
Since $\Out(\alpha_1) = \Out(\alpha_2)$ and
$\In(\alpha_i) \cap \Out(\alpha_i) = \emptyset$ for $i =
1,2$, we can reason as follows: (we use $O$ to abbreviate
$O(\alpha_1)$)
\begin{align*}
\FV(\varphi_1) \symdif \FV(\varphi_2) &=
(\In(\alpha_1) \cup O) \symdif
(\In(\alpha_2) \cup O) 
=\In(\alpha_1) \symdif \In(\alpha_2) \\
&\subseteq I(\alpha_1) \cup \In(\alpha_2)
=\VarSet.\qedhere
\end{align*}
\end{itemize}
\end{proof}

\section{Relational algebra plans for io-disjoint FLIF}
\label{secplan}

In this section we show how the evaluation problem for $\flio$
expressions can be solved in a very direct manner, using a
translation into a particularly simple form of relational algebra
plans.

We generalize the evaluation problem so that it can take a set of
valuations as input, rather than just a single valuation.
Formally, for an $\flio$ expression $\alpha$ over a database schema 
$\Sch$, an instance $D$ of $\Sch$, and a set $N$ of valuations on
$I(\alpha)$, we want to compute 
\[\EvalL \alpha D N :=
\bigcup \{\EvalL \alpha D \nuin \mid \nuin \in N\}.\]

Viewing variables as attributes, we can view a set of
valuations on a finite set of variables $Z$, like the set $N$
above, as a relation with relation schema $Z$.  Consequently, it
is convenient to use the named perspective of the relational
algebra \cite{ahv_book}, where every expression has an output
relation schema (a finite set of attributes; variables in our
case).  We briefly review the well-known operators of the
relational algebra and their behavior on the relation schema
level:
\begin{itemize}
\item
Union and difference are allowed only on relations with the same
relation schema.
\item
Natural join ($\bowtie$) can be applied on two relations with relation
schemas $Z_1$ and $Z_2$, and produces a relation with relation
schema $Z_1 \cup Z_2$.
\item
Projection ($\pi$) produces a relation with a relation schema that is
a subset of the input relation schema.
\item
Selection ($\sigma$) does not change the schema.
\item
Renaming will not be needed.  Instead, however, to accommodate
the assignment expressions present in $\flio$, we will need the
generalized projection operator that adds a new
attribute with the same value as an existing attribute, or a
constant.  Let $N$ be a relation with relation schema $Z$, let
$y\in Z$, and let $x$ be a variable not in $Z$. Then
\begin{align*}
& \pi_{Z,x:=y}(N) = \{\nu[x:=\nu(y)] \mid \nu \in N\} \\
& \pi_{Z,x:=c}(N) = \{\nu[x:=c] \mid \nu \in N\}
\end{align*}

\end{itemize}

Plans are based on \emph{access methods}, which have the
following syntax and semantics.  Let $R(\bar x;\bar y)$ 
be an atomic $\flio$-expression. Let $X$ be the set of variables
in $\bar x$ and let $Y$ be the set of variables in $\bar y$ (in
particular, $X$ and $Y$ are disjoint).
Let $N$ be a relation with a relation schema $Z$ that contains
$X$ but is disjoint from $Y$.
Let $D$ be a database instance.  We define the result of the
\emph{access join} of $N$ with $R(\bar x;\bar y)$,
evaluated on $D$, to be the
following relation with relation schema $Z \cup Y$:
\[N \access R(\bar x;\bar y) := 
\{\text{$\nu$ valuation on $Z \cup Y$}
\mid \nu|_Z \in N \text{ and } \nu(\bar x) \conc
\nu(\bar y) \in D(R)\} \]
This result relation can clearly be computed respecting the limited access
pattern on $R$.  Indeed, we iterate through the valuations in $N$, feed
their $X$-values to the source $R$, and extend the valuations
with the obtained $Y$-values.

Formally, over any database schema $\Sch$ and
for any finite set of variables $I$, we define a
\emph{plan over $\Sch$ with input variables $I$}
as an expression that can be built up as follows:
\begin{itemize}
\item
The special relation name $\MyIn$, with relation schema $I$, is a plan.
\item
If $R(\bar x;\bar y)$ is an atomic $\flio$ 
expression over $\Sch$, with sets of variables $X$ and $Y$ as
above, and $E$ is a plan with output relation schema $Z$ as above, then
also $E \access R(\bar x;\bar y)$ is a
plan, with output relation schema $Z \cup Y$.
\item
Plans are closed under union, difference, natural join, and
projection.
\end{itemize}

Given a database instance $D$, a set $N$ of valuations on $I$,
and a plan $E$ with input variables $I$, we can instantiate the
relation name $\MyIn$ by $N$ and evaluate $E$ on $(D,N)$ in the
obvious manner.  We denote the result by $E(D,N)$.

We establish:

\begin{thm} \label{theorplan}
For every $\flio$ expression $\alpha$ over a database schema
$\Sch$ there exists a plan
$E_\alpha$ over $\Sch$ with input variables $I(\alpha)$,
such that $\EvalL \alpha D N = E_\alpha(D,N)$, for every instance
$D$ of $\Sch$ and set $N$ of valuations on $I(\alpha)$.
\end{thm}
\begin{exa}\hfill
\begin{itemize}
\item Let $\alpha$ be $R(x;y) \comp S(y;z)$.  Recall that 
$I(\alpha) = \{x\}$.  A plan for $\alpha$ can be taken to be 
\[(\MyIn \access R(x;y)) \access S(y;z).\]
\item Let $\alpha$ be $R(x_1;y,u) \comp S(x_2,y; z,u)$.
Recall that $I(\alpha) = \{x_1,x_2\}$.  A plan for $\alpha$
can be taken to be
\[\pi_{x_1,x_2,y}(\MyIn \access
R(x_1;y,u)) \access S(x_2,y; z,u).\]
\item
Recall the expression
$R(x;y_1) \cup S(x;y_2)$ from Example~\ref{exio}, which has
input variables $\{x,y_1,y_2\}$ and no output variables.  A plan
for this expression is
\[(\pi_{x,y_2}(\MyIn) \access R(x;y_1)) \bowtie \MyIn
\; \cup \;
(\pi_{x,y_1}(\MyIn) \access S(x;y_2)) \bowtie \MyIn.\]
The joins with $\MyIn$ ensure that the produced output values
are equal to the given input values, which may be needed in case
$N$ has multiple tuples.
\end{itemize}
\end{exa}
\begin{proof}
To prove the theorem we need a stronger induction hypothesis,
where we allow $N$ to have a larger relation schema 
$Z \supseteq I(\alpha)$, while still being disjoint with $O(\alpha)$.  
The claim then is that 
\[
E_\alpha(D,N) = \{\nu \text{ on } Z\cup
O(\alpha) \mid \nu|_{\var(\alpha)} \in \EvalL \alpha D
{\nu|_{I(\alpha)}}\}.
\]

The base cases are clear.
If $\alpha$ is $R(\bar x;\bar y)$, then $E_\alpha$ is 
$\MyIn \access R(\bar x;\bar y)$.  If $\alpha$ is $(x=y)$, 
then $E_\alpha$ is the selection $\sigma_{x=y}(\MyIn)$.  
If $\alpha$ is $(x:=y)$, then $E_\alpha$ is the generalized 
projection $\pi_{y,x:=y}(\MyIn)$.

In what follows we use the following notation.  Let $P$ 
and $Q$ be plans.  By $Q(P)$ we mean the plan obtained 
from $Q$ by substituting $P$ for $\MyIn$.

Suppose $\alpha$ is $\alpha_1 \comp \alpha_2$.  Plan
$E_{\alpha_1}$, obtained by induction, assumes an input 
relation schema that contains $I(\alpha_1)$ and is disjoint 
from $O(\alpha_1)$.  Since $I(\alpha)=I(\alpha_1) \cup
(I(\alpha_2)-O(\alpha_1))$, $I(\alpha_1)\cap O(\alpha_1)=
\emptyset$, and $Z$ is disjoint from $O(\alpha)=O(\alpha_1)\cup
O(\alpha_2)$, we can apply $E_{\alpha_1}$ with input relation
schema $Z$.  Let $P_1$ be the plan
$\pi_{Z-O(\alpha_2)}(E_{\alpha_1})$.  Then $E_\alpha$ is the 
plan $E_{\alpha_2}(P_1)$.  (One can again verify that this
is a legal plan.)

Next, suppose $\alpha$ is $\alpha_1 \cup \alpha_2$.  Then
$I(\alpha)=I(\alpha_1)\cup I(\alpha_2)$, which is disjoint from
$O(\alpha_1)=O(\alpha_2)$ (compare Proposition~\ref{checkio}).
Hence, for $E_\alpha$ we can simply take the plan 
$E_{\alpha_1} \cup E_{\alpha_2}$.

\intersection{Next, suppose $\alpha$ is $\alpha_1 \cap \alpha_2$.
Note that $I(\alpha)=I(\alpha_1)\cup I(\alpha_2)\cup(O(\alpha_1)
\symdif O(\alpha_2))$. Now $E_\alpha$ is 
\[ 
E_{\alpha_1}(\pi_{I(\alpha)-O(\alpha_1)}(\MyIn)) \bowtie \MyIn 
\; \cap \;
E_{\alpha_2}(\pi_{I(\alpha)-O(\alpha_2)}(\MyIn)) \bowtie \MyIn.
\]
}

Finally, suppose $\alpha$ is $\alpha_1 - \alpha_2$.  Then
$E_\alpha$ is 
\[
E_{\alpha_1} -
(E_{\alpha_2}(\pi_{I(\alpha)-O(\alpha_2)}(\MyIn)) \bowtie \MyIn.
\] 

In general, in the above translations, we follow the principle
that the result of a subplan $E_{\alpha_i}$ must be joined with
$\MyIn$ whenever $O(\alpha_i)$ may intersect with $I(\alpha)$.
\end{proof}
\begin{rem}
When we extend plans with assignment statements such that 
common expressions can be given a name \cite{benedikt_book},
the translation given in the above proof leads to a plan 
$E_\alpha$ of size linear of the length of $\alpha$.  Each 
time we do a substitution of a subexpression for $\MyIn$ 
in the proof, we first assign a name to the
subexpression and only substitute the name.
\end{rem}

\begin{exa}\label{exp12}
Recall the query from Example~\ref{friends1} expressed in 
$\flio$ slightly differently as follows:
\[
F(x;y_1) \comp F(x;y_2) \comp (F(y_1;z) \cap 
F(y_2;z)) \comp (y_1 \neq y_2)
\]
The plan equivalent to this expression is:
\[
\sigma_{y_1 \neq y_2}((E \bowtie F(y_1;z)) 
\bowtie (E \bowtie F(y_2;z)))
\]
where $E$ is the partial plan 
$\MyIn \bowtie F(x;y_1) \bowtie F(x;y_2)$ with $\MyIn$ a 
relation name over $\{x\}$ providing input values.
\end{exa}

\section{Related Work}\label{secrel}

Much of the work on the topic of information sources with access
limitations has been on processing queries expressed in generic
query languages, such as conjunctive queries, unions of
conjunctive queries, conjunctive queries with negation,
first-order logic (relational calculus), or Datalog.  Here, the
query is written as if the database has no access limitations;
the challenge then is to find a query plan that does respect the
limitations, but produces, ideally, the same answers, or, failing
that, produces only correct answers (also known as sound
rewritings) \cite{dgl_rec_int}.

Query plans could take the form of syntactically ordered
fragments of the query languages that are used, like executable
FO considered in the present paper
\cite{rsu_bindingpatterns,lilimited,nl_accesspatterns}.  Query
plans can also be directly described in relational algebra, like
the plans defined here in Section~\ref{secplan}
\cite{ylgcu_capabilities_lim_med,flms_qo_limited,benedikt_planproofs_tods}.
An alternative approach to query processing under access
limitations is to first retrieve the ``accessible part'' of the
database; after that we can simply execute the original query on that
part, which is a sound strategy for monotone queries.  Computing
the accessible part may require recursion; on the other hand, the
computation can be optimized so as to
contain only information needed for the specific query
\cite{cali_query_limit}.

When the query language used is first-order logic, the planning
and optimization problems mentioned above are, of course,
undecidable.  Yet, a remarkable preservation theorem \cite{benedikt_book}
states that, assuming a given first-order query only depends on
the accessible part of the database (for any database; this is a
semantic and undecidable property), then, that query can actually
be rewritten into an executable FO formula.

Interestingly, a variant of our translation results from FLIF
to executable FO in Proposition~\ref{prop:FLIF-BoundedFO} can be seen to follow from the preservation theorem just mentioned. It would suffice to express a given FLIF expression $\alpha$ by any first-order logic formula $\varphi_\alpha$
in the free variables $\SVars_x \cup \SVars_y$, without taking care that $\varphi_\alpha$ is executable. Indeed, FLIF expressions are readily seen to be access-determined by the variables in $\SVars_x$, so, the preservation theorem would imply that $\varphi_\alpha$ can be equivalently written by an executable formula.  Of course, our result provides a much more direct translation, and moreover, shows a bound on the number of variables (free or bound) needed in $\varphi_\alpha$.

Furthermore, it is natural to expect (although we have not verified it formally) that any FLIF expression $\alpha$ is already access-determined by the set of its input variables $I(\alpha)$. In this manner, also Theorem~\ref{flio2exfo} would be implied by the preservation theorem.  Again, our theorem provides a direct and actually very efficient translation.

Incidentally, in the cited work \cite{benedikt_book}, Benedikt et
al.\ define their own notion of executable FO, 
syntactically rather different from the one we use in the
present paper (which was introduced much earlier by Nash and
Lud\"ascher \cite{nl_accesspatterns}).  We prefer the
language we use for its elegance, and because its treatment of input
variables matches well with input variables for FLIF expressions.
Still, both executable-FO languages are
equivalent in expressive power, as they are both equivalent to
the plans used here in Section~\ref{secplan} and also used by
Benedikt et al.

In a companion paper \cite{lif-tocl}, first presented at the KR
2020 conference, we consider LIF in a broader (but still
first-order) context, independently of access limitations.
The companion paper considers the problem of sensitivity analysis
for general LIF expressions and introduces semantic as well as
syntactic definitions of input and output variables.  the
syntactic definitions were shown to be optimal approximations of
the semantic definitions.
We have adopted the syntactic treatment in this paper, and have shown its
relevance, when applied it to FLIF, to querying information
sources with access limitations.  Propositions \ref{propinertia},
\ref{propfv} and \ref{propindet} are adopted here from our companion
paper (there, numbered Proposition~4.3, Lemma~4.4, and
Lemma~4.6, respectively); the proofs can be found there.

\section{Conclusion}\label{seconcl}

We have presented a connection between executable queries on
databases with access limitations on the one hand, and first-order
dynamic logic frameworks on the other hand.  Specifically, we
have defined Forward LIF (FLIF), an instantiation of the Logic of
Information Flows (LIF\@).  FLIF presents itself as an XPath-like
language for graphs of valuations, where edges
represent information accesses.  The main novelty of FLIF lies in
its graph-navigational nature (without explicit quantification),
its input-output mechanism, and the law of inertia that it obeys.
Specifically for io-disjoint FLIF expressions, our work also
presents a more transparent alternative to the result by Benedikt
et al.\ on translating (their version of) executable first-order
formulas to plans.  We have also given renewed attention to Nash
and Lud\"ascher's elegant executable FO language, which seemed to have
been overlooked by more recent research in the field.

Figure~\ref{fig:summary} illustrates our main technical results,
which offer translations between various languages.  Most of the
translations are simple in their formulation, although the
rigorous proof of correctness is not always that simple. 

\begin{figure}
    \centering 
    \includegraphics[width=\textwidth]{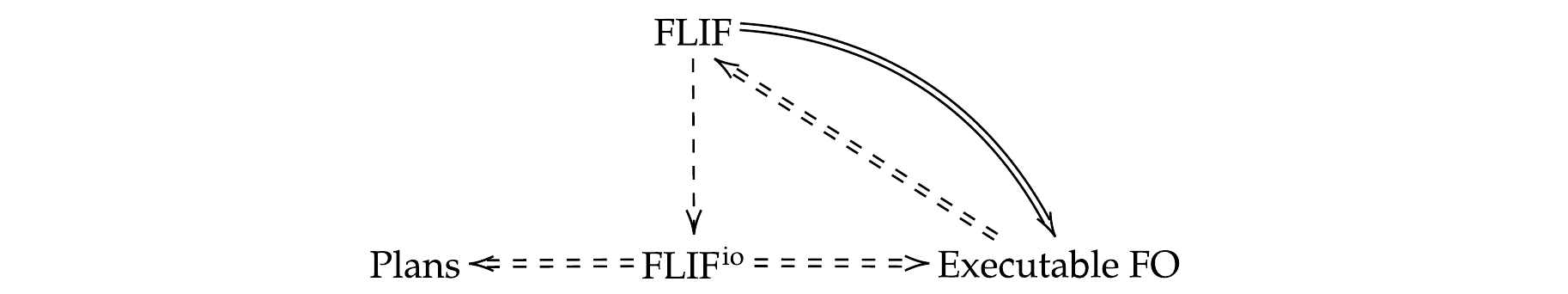}
    \caption{Summary of the translations shown in the paper, where
    a dashed arrow denotes an input-respecting translation, a double arrow 
    denotes a simple translation, and finally, a double dashed arrow denotes both.}
    \label{fig:summary}
\end{figure}

We are not claiming that FLIF is necessarily more user-friendly
than previous languages, or necessarily easier to implement or
optimize.  Both of these aspects should be the topic of further research.
Still we believe it offers a novel perspective.  That FLIF can
express all executable FO queries is something that is not
obvious at first sight.

In closing, we note that querying under limited access patterns
has applicability beyond traditional data or information sources.
For instance in the context of distributed data, when performing
tasks involving the composition of external services, functions,
or modules, limited access patterns are a way for service
providers to protect parts of their data, while still allowing
their services to be integrated seamlessly in other applications.
Limited access patterns also have applications in active
databases, where we like to think of FLIF as an analog of
Active XML \cite{activexml} for the relational data model.

\section*{Acknowledgment}

This research was partially supported by the Flanders AI Research Program. We thank the anonymous reviewers for their critical comments on an earlier version of this paper, which prompted us to significantly improve the presentation of our results.


\bibliographystyle{alphaurl}

\newcommand{\etalchar}[1]{$^{#1}$}




\end{document}